\documentclass[12pt]{amsart}

\usepackage[a4paper, 
            inner = 1.5cm, 
            outer = 1.5cm, 
            top = 2cm, 
            bottom = 3cm]{geometry}
\usepackage{amssymb, amsmath, amsthm}
\usepackage{hyperref}

\newcommand{\be}{\begin{equation}}
\newcommand{\ee}{\end{equation}}

\newcommand{\bC}{\mathbb{C}}
\newcommand{\bN}{\mathbb{N}}
\newcommand{\bR}{\mathbb{R}}
\newcommand{\bZ}{\mathbb{Z}}

\newcommand{\cA}{\mathcal{A}}
\newcommand{\cB}{\mathcal{B}}
\newcommand{\cC}{\mathcal{C}}
\newcommand{\cD}{\mathcal{D}}
\newcommand{\cE}{\mathcal{E}}
\newcommand{\cF}{\mathcal{F}}
\newcommand{\cG}{\mathcal{G}}
\newcommand{\cH}{\mathcal{H}}
\newcommand{\cK}{\mathcal{K}}
\newcommand{\cL}{\mathcal{L}}
\newcommand{\cM}{\mathcal{M}}
\newcommand{\cN}{\mathcal{N}}
\newcommand{\cO}{\mathcal{O}}
\newcommand{\cP}{\mathcal{P}}
\newcommand{\cR}{\mathcal{R}}
\newcommand{\cS}{\mathcal{S}}
\newcommand{\cU}{\mathcal{U}}
\newcommand{\cV}{\mathcal{V}}
\newcommand{\cW}{\mathcal{W}}
\newcommand{\cZ}{\mathcal{Z}}

\newcommand{\mfa}{\mathfrak{a}}
\newcommand{\mfc}{\mathfrak{c}}
\newcommand{\mfg}{\mathfrak{g}}
\newcommand{\mfgl}{\mathfrak{gl}}
\newcommand{\mfk}{\mathfrak{k}}
\newcommand{\mfm}{\mathfrak{m}}
\newcommand{\mfp}{\mathfrak{p}}
\newcommand{\mfu}{\mathfrak{u}}
\newcommand{\mfM}{\mathfrak{M}}
\newcommand{\mfN}{\mathfrak{N}}
\newcommand{\mfX}{\mathfrak{X}}

\newcommand{\bszero}{\boldsymbol{0}}
\newcommand{\bsone}{\boldsymbol{1}}
\newcommand{\bslambda}{\boldsymbol{\lambda}}
\newcommand{\bsLambda}{\boldsymbol{\Lambda}}
\newcommand{\bsTheta}{\boldsymbol{\Theta}}
\newcommand{\bsX}{\boldsymbol{X}}

\newcommand{\bscA}{\boldsymbol{\cA}}
\newcommand{\bscB}{\boldsymbol{\cB}}
\newcommand{\bscV}{\boldsymbol{\cV}}
\newcommand{\bscW}{\boldsymbol{\cW}}
\newcommand{\bscZ}{\boldsymbol{\cZ}}
\newcommand{\bsxi}{\boldsymbol{\xi}}
\newcommand{\bseta}{\boldsymbol{\eta}}

\newcommand{\ri}{\mathrm{i}}
\newcommand{\dd}{\mathrm{d}}
\newcommand{\diag}{\mathrm{diag}}
\newcommand{\reg}{\mathrm{reg}}
\newcommand{\tr}{\mathrm{tr}}
\newcommand{\Id}{\mathrm{Id}}
\newcommand{\ad}{\mathrm{ad}}
\newcommand{\wad}{\widetilde{\mathrm{ad}}}
\newcommand{\Real}{\mathrm{Re}}
\newcommand{\Imag}{\mathrm{Im}}
\newcommand{\Spec}{\mathrm{Spec}}

\newcommand{\hg}{{\hat{g}}}
\newcommand{\hu}{{\hat{u}}}
\newcommand{\hy}{{\hat{y}}}
\newcommand{\hz}{{\hat{z}}}
\newcommand{\hF}{{\hat{F}}}
\newcommand{\hH}{{\hat{H}}}
\newcommand{\hL}{{\hat{L}}}
\newcommand{\hlambda}{{\hat{\lambda}}}
\newcommand{\htheta}{{\hat{\theta}}}
\newcommand{\hmu}{{\hat{\mu}}}
\newcommand{\hnu}{{\hat{\nu}}}
\newcommand{\hDelta}{\hat{\Delta}}
\newcommand{\hTheta}{{\hat{\Theta}}}
\newcommand{\hbsTheta}{{\hat{\bsTheta}}}

\newcommand{\tx}{\tilde{x}}
\newcommand{\tS}{\tilde{S}}
\newcommand{\tmfM}{\tilde{\mfM}}

\newcommand{\half}{\frac{1}{2}}

\newcommand{\PB}[2]{ \{ #1, #2 \} }

\newcommand{\PD}[2]{\frac{\partial #1}{\partial #2}}

\newcommand{\midand}{\quad \text{and} \quad}

\numberwithin{equation}{section}

\theoremstyle{plain}
\newtheorem{THEOREM}{Theorem}
\newtheorem{LEMMA}[THEOREM]{Lemma}
\newtheorem{PROPOSITION}[THEOREM]{Proposition}

\begin{document}

\title[]{Self-duality and scattering map for the hyperbolic van Diejen systems 
with two coupling parameters (with an appendix by S.~Ruijsenaars)}

\author[]{B\'ela G\'abor Pusztai}

\address{
    Bolyai Institute, University of Szeged \\
    Aradi v\'ertan\'uk tere 1, 
    H-6720 Szeged, Hungary
    \emph{and}
    MTA Lend\"ulet Holographic QFT Group, Wigner RCP \\
    H-1525 Budapest 114, P.O.B. 49, Hungary}

\email[B.G.~Pusztai]{gpusztai@math.u-szeged.hu}

\address{
    School of Mathematics, University of Leeds \\
    Leeds LS2 9JT, UK}

\email[S.~Ruijsenaars]{siru@maths.leeds.ac.uk}



\begin{abstract}
In this paper, we construct global action-angle variables for a certain 
two-parameter family of hyperbolic van Diejen systems. Following Ruijsenaars'
ideas on the translation invariant models, the proposed action-angle variables 
come from a thorough analysis of the commutation relation obeyed by the Lax 
matrix, whereas the proof of their canonicity is based on the study of the 
scattering theory. As a consequence, we show that the van Diejen system of 
our interest is self-dual with a factorized scattering map. Also, in an 
appendix by S.~Ruijsenaars, a novel proof of the spectral asymptotics of 
certain exponential type matrix flows is presented. This result is of crucial 
importance in our scattering-theoretical analysis.
\end{abstract}

\maketitle
\tableofcontents

\section{Introduction}
\label{SECTION_Introduction}
The action-angle duality, or Ruijsenaars duality, is one of the most 
attractive features of the Calogero--Moser--Sutherland (CMS) 
\cite{Calogero, Sutherland, Moser_1975, Olsha_Pere_1976, Olsha_Pere_1981,
Sutherland_2004} and the Ruijsenaars--Schneider--van Diejen (RSvD) 
\cite{Ruij_Schneider, Ruij_FiniteDimSolitonSystems, van_Diejen_ComposMath, 
van_Diejen_TMP1994, van_Diejen_JMP1995} integrable many-particle systems. 
At the classical level, in the context of the translation invariant rational 
and hyperbolic models associated with the $A$-type root systems, this 
fascinating phenomenon was discovered by Ruijsenaars in the seminal paper 
\cite{Ruij_CMP1988}. Due to the importance of this observation, using 
various advanced techniques ranging from the methods of gauge theory to the 
machinery of symplectic reductions, by now the duality properties have been 
reinterpreted, and also exhibited for a much wider class of $A$-type models 
(see e.g. \cite{Ruij_RIMS_2, Ruij_RIMS_3, Nekrasov_1999, Fock_Rosley_1999, 
Fock_et_al_JHEP07, Feher_Klimcik_JPA2009, Feher_Ayadi_JMP2010, 
Feher_Klimcik_CMP2011, Feher_Klimcik_NPB2012, Feher_Kluck_NPB2014}).

Although the translation invariant members of the RSvD family have drawn 
considerable attention, the theory of the classical multi-parametric 
non-$A$-type integrable deformations introduced by van Diejen 
\cite{van_Diejen_TMP1994, van_Diejen_JMP1995} is far less developed. For 
brevity, in the rest of the paper the non-$A$-type members of the RSvD family 
are simply referred to as the van Diejen systems. Now, the transparent 
asymmetry between the maturation of the theories of the classical $A$-type 
and non-$A$-type models is probably best elucidated by the somewhat surprising 
fact that no Lax representation is known for the most general van Diejen 
dynamics. Of course, it was observed already in the early stage of the 
developments that by folding the $A$-type root systems one can construct Lax 
matrices of type $C$ and $BC$, but only with a single coupling parameter 
\cite{Ruij_CMP1988}. As concerns the $D$-type models, only partial results are 
available for small values of the number of particles \cite{Chen_Hou_JPA2001}.
Nevertheless, to close the gap, the last couple of years have witnessed 
the emergence of some new ideas in the literature to cope with the intricacies 
posed by the classical van Diejen models. Indeed, by working out Lax matrices 
for the most general rational variants of the RSvD family associated with the
$BC$-type root systems, the duality properties of these special non-$A$-type 
van Diejen models are also settled completely (see \cite{Pusztai_NPB2011, 
Pusztai_NPB2012, Feher_Gorbe_JMP2014}). Prior to our present work, at the 
level of the classical hyperbolic systems, non-trivial results could be found 
only in \cite{Ruij_TMP2008}, where the $1$-particle $BC_1$ model is studied
by direct techniques. However, as the most recent progress in this research 
area, in a joint work with G\"orbe \cite{Pusztai_Gorbe} we constructed Lax 
pairs for certain two-parameter family of \emph{hyperbolic} van Diejen 
systems, too. As a natural step forward, in this paper we wish to uncover 
the self-duality property of these special hyperbolic systems. 

In order to describe the van Diejen models of our interest, it proves 
convenient to start with the shorthand notation
\be
    \bN_m = \{ 1, \ldots, m \} \subset \bN
    \qquad
    (m \in \bN).
\label{bN_m}
\ee
Furthermore, take an arbitrary $n \in \bN$, let $N = 2 n$, and consider the 
open subset
\be
    P = \{ 
        p = (\bsxi, \bseta) 
        = (\xi_1, \ldots, \xi_n, \eta_1, \ldots, \eta_n) \in \bR^N
        \mid
        \xi_1 > \ldots > \xi_n > 0 
    \},
\label{P}
\ee
that we endow with the smooth manifold structure inherited from the ambient 
space $\bR^N$. Note that the functions 
$\lambda_1, \ldots, \lambda_n, \theta_1, \ldots, \theta_n \in C^\infty(P)$
defined by the formulae
\be
    \lambda_a(p) = \xi_a
    \midand
    \theta_a(p) = \eta_a
    \qquad
    (a \in \bN_n, \, p \in P)
\label{lambda_and_theta}
\ee
provide a global coordinate system on $P$. For consistency with the 
terminology of the Ruijsenaars--Schneider models, the coordinates $\lambda_a$ 
and $\theta_a$ are usually called the particle positions and the particle 
rapidities, respectively. For brevity, we also introduce the notations
\be
    x_a = \lambda_a
    \midand
    x_{n + a} = \theta_a
    \qquad
    (a \in \bN_n).
\label{x}
\ee
Note that the even dimensional manifold $P$ can be equipped with the 
symplectic form
\be
    \omega = \sum_{a = 1}^n \dd \lambda_a \wedge \dd \theta_a,
\label{omega}
\ee
which is natural in the sense that thereby we can think of $(P, \omega)$ as 
a model of the cotangent bundle of the configuration space
\be
    Q 
    = \{ \bsxi = (\xi_1, \ldots, \xi_n) \in \bR^n 
        \mid 
        \xi_1 > \ldots > \xi_n > 0 \}.
\label{Q}
\ee
Upon introducing the $N \times N$ matrix
\be
    \Omega
    = [\Omega_{k, l}]_{1 \leq k, l \leq N}
    = \begin{bmatrix}
        \bszero_n & \bsone_n \\
        -\bsone_n & \bszero_n
    \end{bmatrix},    
\label{Omega}
\ee
observe that the Poisson bracket associated with $\omega$ \eqref{omega} 
takes the form
\be
    \PB{f}{g} = \sum_{k, l = 1}^N \Omega_{k, l} \PD{f}{x_k} \PD{g}{x_l}
    \qquad
    (f, g \in C^\infty(P)).
\label{PB}
\ee
In particular, the distinguished coordinates \eqref{lambda_and_theta} are
canonical, i.e.,
\be
    \PB{\lambda_a}{\lambda_b} = 0,
    \quad
    \PB{\theta_a}{\theta_b} = 0,
    \quad
    \PB{\lambda_a}{\theta_b} = \delta_{a, b}
    \qquad
    (a, b \in \bN_n).
\label{lambda_and_theta_PBs}
\ee

To proceed further, let $g = (\mu, \nu) \in \bR^2$ be an arbitrary point 
satisfying 
\be
    \sin(\mu) \neq 0 \neq \sin(\nu), 
\label{mu_nu_genericity}
\ee    
and for each $a \in \bN_n$ define the function
\be
    u^g_a 
    = \left( 1 + \frac{\sin(\nu)^2}{\sinh(2 \lambda_a)^2} \right)^\half
        \prod_{\substack{c = 1 \\ (c \neq a)}}^n
            \left( 
                1 + \frac{\sin(\mu)^2}{\sinh(\lambda_a - \lambda_c)^2} 
            \right)^\half
            \left( 
                1 + \frac{\sin(\mu)^2}{\sinh(\lambda_a + \lambda_c)^2} 
            \right)^\half
    \in C^\infty(P).
\label{u}
\ee
In passing we mention that in the above expression the values of the 
parameters $\mu$ and $\nu$ matter only modulo $\pi$. Now, the $n$-particle 
van Diejen model of our interest is the classical Hamiltonian system 
$(P, \omega, H^g)$ characterized by the smooth Hamiltonian
\be
    H^g = \sum_{a = 1}^n \cosh(\theta_a) u^g_a.
\label{H}
\ee
Remembering the integrable many-particle models introduced in 
\cite{van_Diejen_TMP1994, van_Diejen_JMP1995}, it is clear that this 
Hamiltonian system does belong to the hyperbolic RSvD family with two 
independent coupling parameters $\mu$ and $\nu$. 

Having defined the key player, now we wish to summarize the content of the 
rest of the paper. In the next section we gather the necessary background 
material underlying the study of the van Diejen system \eqref{H}. Since our 
present work can be seen as a natural continuation of \cite{Pusztai_Gorbe}, 
the reader may find it convenient to have a copy of \cite{Pusztai_Gorbe} on 
hand for the proofs of the results outlined in Section 
\ref{SECTION_Preliminaries}. 

Turning to the study of the commutation relation \eqref{commut_rel} obeyed by 
the Lax matrix $L$ \eqref{L} of the van Diejen system \eqref{H}, in Section 
\ref{SECTION_the_dual_objects} we construct a diffeomorphism $\Psi^g$ 
\eqref{Psi} from the spectral data associated with the Lax matrix. Since 
its components $\htheta^g_a$ \eqref{htheta} and $\hlambda^g_a$ \eqref{hlambda}
are defined with the aid of the diagonalization of $L$, the construction of 
$\Psi^g$ is of purely algebraic nature. We wish to emphasize that this 
construction is a direct generalization of Ruijsenaars' approach on the 
translation invariant models \cite{Ruij_CMP1988}, so it is natural to expect 
that the globally defined smooth functions $\htheta^g_a$ and $\hlambda^g_a$ 
$(a \in \bN_n)$ are actually action-angle coordinates for the Hamiltonian 
system \eqref{H}. 

To prove that they do form a Darboux system, in Section 
\ref{SECTION_Scattering_theory_and_duality} we turn our attention to the 
scattering properties of \eqref{H}. Of course, this idea was the other 
cornerstone of the developments presented in \cite{Ruij_CMP1988}, but our 
implementation is quite different. Indeed, rather than using techniques 
from the theory of functions of several complex variables, in Section 
\ref{SECTION_Scattering_theory_and_duality} we apply straightforward 
dynamical system techniques, and a bit real analysis, to prove the 
canonicity of the proposed action-angle variables. The main result of 
the paper is formulated in Theorem \ref{THEOREM_main_theorem}. This theorem 
allows us to complete the study of the scattering theory, too. Namely, 
as formulated in Theorem \ref{THEOREM_wave_maps}, we show that the M\o ller 
wave transformations are symplectomorphisms. Moreover, from the explicit 
formula \eqref{scattering_formula} we see that the scattering map $S$ has 
a factorized form. In Section \ref{SECTION_Discussion} we briefly discuss 
the consequences of our results. Also, we pose some open problems related 
to the theory of the van Diejen type particle systems. We conclude the paper 
with an appendix by S.~Ruijsenaars, on the spectral asymptotics of certain 
exponential type matrix flows.

\section{Preliminaries}
\label{SECTION_Preliminaries}
In this short section we wish to summarize the fundamental algebraic 
properties of the Lax matrix of the van Diejen system \eqref{H} we 
constructed in \cite{Pusztai_Gorbe}. To make the presentation essentially 
self-contained, and also to fix the notations, it is expedient to start 
with a brief overview on the underlying Lie theoretical objects. As a rule, 
the manifolds appearing in this paper are real and smooth.

\subsection{Background material from Lie theory}
\label{SUBSECTION_Lie_theory_background}
One of the most important observations we made in \cite{Pusztai_Gorbe} is 
that many properties of the van Diejen systems \eqref{H} can be understood 
in a geometric setup based on the non-compact reductive matrix Lie group
\be
    G = U(n, n) = \{ y \in GL(N, \bC) \mid y^* C y = C \},
\label{G}
\ee
where 
\be
    C 
    = \begin{bmatrix}
        \bszero_n & \bsone_n \\
        \bsone_n & \bszero_n
    \end{bmatrix}
    \in GL(N, \bC).
\label{C} 
\ee
It is clear that the Lie algebra of $G$ \eqref{G} can be identified with
\be
    \mfg 
    = \mfu(u, n) 
    = \{ Y \in \mfgl(N, \bC) \mid Y^* C + C Y = 0 \}.
\label{mfg}
\ee
Notice also that the set of unitary elements
\be
    K = \{ y \in G \mid y^* y = \bsone_N \} \cong U(n) \times U(n)
\label{K}
\ee
forms a maximal compact subgroup in $G$ \eqref{G}, and the corresponding
Lie subalgebra takes the form
\be
    \mfk = \{ Y \in \mfg \mid Y^* + Y = 0 \} \cong \mfu(n) \oplus \mfu(n).
\label{mfk}
\ee
By taking the complementary subspace
\be
    \mfp = \{ X \in \mfg \mid X^* = X \},
\label{mfp}
\ee
we end up with the $\bZ_2$-gradation 
\be
    \mfg = \mfk \oplus \mfp, 
\label{gradation}
\ee    
which is actually orthogonal with respect to the trace pairing defined on 
the matrix Lie algebra $\mfg$ \eqref{mfg}. Note that at the Lie group level 
the natural analogue of this decomposition is the diffeomorphism
\be
    \cC \colon \mfp \times K \rightarrow G,
    \quad 
    (X, k) \mapsto e^X k.
\label{cC}    
\ee
As a further ingredient, recall that the restriction of the exponential map 
onto the subspace $\mfp$ \eqref{mfp} is injective. Moreover, by taking the 
image of $\mfp$ under the exponential map, in $G$ we obtain the closed 
embedded submanifold
\be
    \cP = \exp(\mfp) = \{ e^X \in G \mid X \in \mfp \}.
\label{cP}
\ee
As a matter of fact, it coincides with the set of positive definite elements 
of $G$; that is,
\be
    \cP = \{ y \in U(n, n) \, | \, y > 0 \}.
\label{cP_identification}
\ee
We mention in passing that, due to the global Cartan decomposition 
\eqref{cC}, $\cP$ can be identified with the non-compact symmetric space 
associated with $(G, K)$, i.e., $\cP \cong G / K$.

Besides the above basic objects, in the following we shall also need some 
finer elements from the structure theory of $G$ \eqref{G}. As the first
step toward this goal, in $\mfp$ \eqref{mfp} we introduce the maximal Abelian 
subspace
\be
    \mfa 
    = \{ X = \diag(x_1, \ldots, x_n, -x_1, \ldots, -x_n) 
        \mid 
        x_1, \ldots, x_n \in \bR \}.
\label{mfa}
\ee
Moreover, for the subset of the off-diagonal elements of $\mfp$ we introduce 
the notation $\mfa^\perp$. Clearly the centralizer of $\mfa$ in $K$ \eqref{K} 
is the Abelian Lie group
\be
    M = 
    \{ \diag(e^{\ri \chi_1}, \ldots, e^{\ri \chi_n}, 
                e^{\ri \chi_1}, \ldots, e^{\ri \chi_n}) 
        \mid
        \chi_1, \ldots, \chi_n \in \bR \}
\label{M}
\ee
with Lie algebra
\be
    \mfm 
    = \{ \diag(\ri \chi_1, \ldots, \ri \chi_n, 
                \ri \chi_1, \ldots, \ri \chi_n) 
        \mid
        \chi_1, \ldots, \chi_n \in \bR \}.
\label{mfm}
\ee
Now, if $\mfm^\perp$ denotes the set of the off-diagonal elements of the 
subalgebra $\mfk$ \eqref{mfk}, then by \eqref{gradation} we can write the 
refined orthogonal decomposition
\be
    \mfg = \mfm \oplus \mfm^\perp \oplus \mfa \oplus \mfa^\perp.
\label{refined_decomposition}
\ee

We proceed by noting that for all
\be
    X = \diag(x_1, \ldots, x_n, -x_1, \ldots, -x_n) \in \mfa
    \qquad
    (x_1, \ldots, x_n \in \bR)
\label{X_matrix} 
\ee
the subspace $\mfm^\perp \oplus \mfa^\perp$ consisting of the off-diagonal
elements of $\mfg$ is invariant under the action of the linear operator
\be
    \ad_X \colon \mfg \rightarrow \mfg, 
    \quad
    Y \mapsto [X, Y].
\label{ad}
\ee
Thus, the restriction
\be
    \wad_X 
    = \ad_X |_{\mfm^\perp \oplus \mfa^\perp}
    \in \mfgl(\mfm^\perp \oplus \mfa^\perp)
\label{wad}
\ee
is well-defined, and for its spectrum we have
\be
    \mathrm{Spec}(\wad_X) 
    = \{x_a - x_b, \pm (x_a + x_b), \pm 2 x_c 
        \mid 
        a, b, c \in \bN_n, \, a \neq b\}.
\label{wad_X_spectrum}
\ee
Recall that in the study of the reductive Lie groups, the regular part of
$\mfa$ \eqref{mfa} is usually defined by the open subset
\be
    \mfa_\reg 
    = \{ X \in \mfa \mid \det(\wad_X) \neq 0 \}.
\label{mfa_reg}
\ee
From \eqref{mfa}, \eqref{wad_X_spectrum} and \eqref{mfa_reg} we see that
\be
    \mfc 
    = \{ X = \diag(x_1, \ldots, x_n, -x_1, \ldots, -x_n) 
        \mid 
        x_1 > \ldots > x_n > 0 \}
\label{mfc}
\ee
is a connected component of $\mfa_\reg$. So, giving a glance at \eqref{Q},
it is clear the configuration space $Q$ of the van Diejen systems \eqref{H} 
can be naturally identified with the standard Weyl chamber $\mfc$ \eqref{mfc}, 
i.e., $Q \cong \mfc$.

Due to its importance in the study of the Lax matrix of the van Diejen
system \eqref{H}, we conclude this subsection with some important facts 
related to the notion of regularity. Our discussion is based on the smooth
surjective map
\be
    \mfa \times K \ni (X, k)
    \mapsto
    k X k^{-1} \in \mfp,
\label{mfa_and_K_and_mfp}
\ee
that allows us to introduce the regular part of $\mfp$ \eqref{mfp} by
\be
    \mfp_\reg = \{ k X k^{-1} \mid k \in K \text{ and } X \in \mfc \}. 
\label{mfp_reg}
\ee
Recall that $\mfp_\reg$ is an open and dense subset of $\mfp$, and the map
\be
    \mfc \times (K / M) \ni (X, k M)
    \mapsto
    k X k^{-1} \in \mfp_\reg
\label{mfp_reg_characterization}
\ee
is a diffeomorphism. Next, by taking the image of $\mfp_\reg$ under the
exponential map, we define
\be
    \cP_\reg 
    = \exp(\mfp_\reg) = \{ e^X \in G \mid X \in \mfp_\reg \} 
        \subset \cP.
\label{cP_reg}
\ee
Since the exponential map is a diffeomorphism from $\mfp$ \eqref{mfp} onto 
$\cP$ \eqref{cP}, the subset $\cP_\reg$ is open and dense in $\cP$. Therefore,
$\cP_\reg$ can be also seen as an embedded submanifold of $G$. Moreover, the
map
\be
    \varUpsilon \colon \mfc \times (K / M) \rightarrow \cP_\reg,
    \quad
    (X, k M) \mapsto k e^X k^{-1}
\label{varUpsilon}
\ee
is a diffeomorphism from $\mfc \times (K / M)$ onto $\cP_\reg$. As a closing
remark, for the proofs of the Lie theoretic facts appearing in this subsection 
we recommend \cite{Knapp}.

\subsection{Fundamental properties of the Lax matrix}
Having reviewed the necessary objects from Lie theory, now we wish to describe
the Lax matrix we constructed in \cite{Pusztai_Gorbe} for the two-parameter
family of hyperbolic van Diejen systems \eqref{H}. As a preliminary step, for 
the space of the admissible model parameters we introduce the notation 
\be
    \mfM
    = \{ g = (\mu, \nu) \in \bR^2 \mid \sin(\mu) \neq 0 \neq \sin(\nu) \}.
\label{mfM}
\ee
Next, for each $a \in \bN_n$ we define the smooth function
\be
    z_a \colon P \times \mfM \rightarrow \bC,
    \quad
    (p, g) \mapsto z_a(p, g) = z_a^g(p),
\label{z_a_def}
\ee
where for each $g = (\mu, \nu) \in \mfM$ the $g$-section $z_a^g$ is given by
\be
    z_a^g = - \frac{\sinh(\ri \nu + 2 \lambda_a)}{\sinh(2 \lambda_a)}
            \prod_{\substack{c = 1 \\ (c \neq a)}}^n
                \frac{\sinh(\ri \mu + \lambda_a - \lambda_c)}
                    {\sinh(\lambda_a - \lambda_c)}
                \frac{\sinh(\ri \mu + \lambda_a + \lambda_c)}
                    {\sinh(\lambda_a + \lambda_c)}
    \in C^\infty(P, \bC).
\label{z_a}
\ee
By taking the modulus of \eqref{z_a_def}, we let 
\be
    u_a = \vert z_a \vert \in C^\infty(P \times \mfM)
    \qquad
    (a \in \bN_n).
\label{u_a_def}
\ee 
This is consistent with our earlier notation, since for all 
$g = (\mu, \nu) \in \mfM$ the $g$-section of $u_a$ does coincide with the 
function we defined in \eqref{u}. An equally important ingredient in the 
construction of the Lax matrix is the column vector valued smooth function
\be
    F \colon P \times \mfM \rightarrow \bC^{N \times 1},
    \quad
    (p, g) \mapsto F(p, g) = F^g(p)    
\label{F_def}
\ee
with components
\be
    F_a(p, g) = e^{\frac{\theta_a(p)}{2}} u_a(p, g)^\half
    \quad \text{and} \quad
    F_{n + a}(p, g) 
    = e^{-\frac{\theta_a(p)}{2}} \bar{z}_a(p, g) u_a(p, g)^{-\half}
    \qquad
    (a \in \bN_n).
\label{F}
\ee
Finally, upon introducing the shorthand notation
\be
    \Lambda_a = \lambda_a
    \midand
    \Lambda_{n + a} = - \lambda_a
    \qquad
    (a \in \bN_n),
\label{Lambda}
\ee
let us recall that the Lax matrix we worked out in \cite{Pusztai_Gorbe} is 
the matrix valued smooth function
\be
    L \colon P \times \mfM \rightarrow \cP,
    \quad
    (p, g) \mapsto L(p, g) = L^g(p),
\label{L}
\ee
where for each $g = (\mu, \nu) \in \mfM$ the $g$-section $L^g$ is given by 
the entries
\be
    L^g_{k, l} 
    = \frac{\ri \sin(\mu) F^g_k \bar{F}^g_l 
        + \ri \sin(\mu - \nu) C_{k, l}}
        {\sinh(\ri \mu + \Lambda_k - \Lambda_l)}
    \qquad
    (k, l \in \bN_N).
\label{L_entries}
\ee
Concerning the relationship between the Hamiltonian $H^g$ \eqref{H} and the
Lax matrix $L$ \eqref{L}, let us note that
\be
    H^g = \half \tr(L^g).
\label{H_and_L}
\ee

Remember that the fact that $L$ \eqref{L} takes values in $\cP$ \eqref{cP}
is itself not entirely trivial (see the proofs of Proposition 1 and Lemma 2 
in \cite{Pusztai_Gorbe}). However, for our present purposes it is much more
important that $L$ obeys a Ruijsenaars type commutation relation (see equation 
(2.4) and the surrounding ideas in \cite{Ruij_CMP1988}). Indeed, utilizing 
the matrix valued smooth function
\be
    \bsLambda 
    = \diag(\Lambda_1, \ldots, \Lambda_N) 
        \in C^\infty(P, \mfc),
\label{bsLambda}
\ee
one can easily show that $\forall g = (\mu, \nu) \in \mfM$ we have
\be
    e^{\ri \mu} e^{\bsLambda} L^g e^{-\bsLambda}
    - e^{- \ri \mu} e^{-\bsLambda} L^g e^{\bsLambda}
    = 2 \ri \sin(\mu) F^g (F^g)^* + 2 \ri \sin(\mu - \nu) C.
\label{commut_rel} 
\ee
Although the proof of this statement is quite straightforward (see Lemma 3
in \cite{Pusztai_Gorbe}), this commutation relation proved to be one of the 
cornerstones of the developments presented in \cite{Pusztai_Gorbe}. As it 
turns out, this relation can be seen as the starting point of the present 
work, too. However, compared to \cite{Pusztai_Gorbe}, in this paper we take 
a completely different route by turning our attention to the `dual systems'. 
Before presenting the new results, let us not forget about an important 
regularity property of $L$. Under the slight technical assumption that the 
coupling parameter belongs to
\be
    \tmfM = \{ g = (\mu, \nu) \in \mfM \mid \sin(2 \mu - \nu) \neq 0 \},
\label{tmfM}
\ee
in \cite{Pusztai_Gorbe} we proved that the spectrum of the Lax matrix 
\eqref{L} is simple. More precisely, the statement we made in Lemma 4 of 
\cite{Pusztai_Gorbe} is that $\forall (p, g) \in P \times \tmfM$ we 
have $L(p, g) \in \cP_\reg$ \eqref{cP_reg}. During the whole paper we shall 
assume that the map $L$ \eqref{L} is regular in this sense.

\section{Construction of the dual objects}
\label{SECTION_the_dual_objects}
Starting with this section, we wish to present our new results related to 
the $2$-parameter family of classical hyperbolic van Diejen systems \eqref{H}. 
Following the lead of the paper \cite{Ruij_CMP1988}, we shall utilize the 
commutation relation \eqref{commut_rel} to infer the relevant spectral 
properties of the Lax matrix $L$ \eqref{L}. As a consequence, our analysis 
shall naturally give rise to the `dual objects' playing the fundamental role 
in the construction of action-angle variables for the Hamiltonian system 
\eqref{H}.

\subsection{Diagonalization of the Lax matrix}
\label{SUBSECTION_Diagonalization_of_L}
As we discussed at the end of the previous section, for all $p \in P$ and 
$g \in \tmfM$ the self-adjoint Lax matrix $L(p, g)$ \eqref{L_entries} is 
regular in the sense that it belongs to $\cP_\reg$ \eqref{cP_reg}. Bearing 
in mind \eqref{varUpsilon}, we see that the matrix $L(p, g)$ can be conjugated 
into a unique element of the subset $\exp(\mfc) \subset \cP_\reg$ by certain 
elements of the compact subgroup $K$ \eqref{K}. In particular, the spectrum 
of $L(p, g)$ is simple. More precisely, due to the identification 
$Q \cong \mfc$ (see \eqref{Q} and \eqref{mfc}), there is a \emph{unique} 
element
\be
    \htheta(p, g) = (\htheta_1 (p, g), \ldots, \htheta_n(p, g)) \in Q
\label{htheta}
\ee
such that we can write that
\be
    \Spec(L(p, g)) = \{ e^{\pm 2 \htheta_a(p, g)} \mid a \in \bN_n \}.
\label{L_spec}
\ee
Utilizing the spectrum of $L$, let us define the function
\be
    \htheta \colon P \times \tmfM \rightarrow Q,
    \quad
    (p, g) \mapsto \htheta(p, g),
\label{htheta_function}
\ee
together with the matrix valued function
\be
    \hbsTheta
    = \diag(\hTheta_1, \ldots, \hTheta_N),
\label{hbsTheta}
\ee
where the diagonal entries are given by
\be
    \hTheta_a = \htheta_a
    \midand
    \hTheta_{n + a} = -\htheta_a
    \qquad
    (a \in \bN_n).
\label{hTheta}
\ee
Notice that $\hbsTheta$ is smooth, i.e.,
\be 
    \hbsTheta \in C^\infty(P \times \tmfM, \mfc). 
\label{hbsTheta_smooth}
\ee
Indeed, recalling the fact that the smooth function $L$ \eqref{L} takes 
values in $\cP_\reg$, it is evident that with the aid of the diffeomorphism 
$\varUpsilon$ \eqref{varUpsilon} we can write 
\be
    \hbsTheta = \half \text{pr}_\mfc \circ \varUpsilon^{-1} \circ L,
\label{hbsTheta_eq}
\ee
where $\text{pr}_\mfc \colon \mfc \times (K / M) \twoheadrightarrow \mfc$ is 
the (smooth) canonical projection from the product $\mfc \times (K / M)$ onto 
the first factor $\mfc$. As a trivial consequence, the smoothness of $\htheta$ 
\eqref{htheta_function} also follows.

As concerns the diagonalization of the matrix $L(p, g)$ at any given point 
$(p, g) \in P \times \tmfM$, later on we shall also need information about 
those elements of $K$ \eqref{K} that can be used to transform $L(p, g)$ 
into $\exp(\mfc)$ by conjugation. Of course, there is a plethora of such 
diagonalizing elements. Indeed, from \eqref{varUpsilon} it is evident that the 
freedom of choice is completely characterized by the $n$-dimensional subgroup 
$M$ \eqref{M}. Thus, aiming for uniqueness, on the diagonalizing matrices we 
shall impose  that the first $n$ components of a certain column vector be 
strictly positive (see \eqref{hy_2} below). Namely, keeping in mind the 
objects defined in \eqref{F} and \eqref{bsLambda}, our observation can be 
formulated as follows.

\begin{LEMMA}
\label{LEMMA_hy}
For all $p \in P$ and $g \in \tmfM$ there is a unique element 
$\hy(p, g) \in K$ such that
\be
    L(p, g) = \hy(p, g) e^{2 \hbsTheta(p, g)} \hy(p, g)^{-1},
\label{hy_1}
\ee
and also $\forall a \in \bN_n$ we have
\be
    ( e^{-\hbsTheta(p, g)} \hy(p, g)^{-1} e^{\bsLambda(p)} F(p, g) )_a > 0.
\label{hy_2}
\ee
Moreover, the resulting function
\be
    \hy \colon P \times \tmfM \rightarrow K,
    \quad
    (p, g) \mapsto \hy(p, g)
\label{hy_function}
\ee
is smooth, i.e., $\hy \in C^\infty(P \times \tmfM, K)$.
\end{LEMMA}

\begin{proof}
We start with the existence part of the Lemma. Let $p \in P$ 
and $g = (\mu, \nu) \in \tmfM$ be arbitrary elements and keep them fixed. 
Recalling the construction of the diagonal matrix \eqref{hbsTheta}, from 
\eqref{varUpsilon} it follows that there is an element $y \in K$ such that
\be
    L(p, g) = y e^{2 \hbsTheta(p, g)} y^{-1}.
\label{y_def}
\ee
Plugging it into the commutation relation \eqref{commut_rel}, we get
\be
\begin{split}
    & e^{\ri \mu} 
        e^{\bsLambda(p)} y e^{2 \hbsTheta(p, g)} y^{-1} e^{-\bsLambda(p)}
    - e^{-\ri \mu} 
        e^{-\bsLambda(p)} y e^{2 \hbsTheta(p, g)} y^{-1} e^{\bsLambda(p)}
    \\
    & \quad 
    = 2 \ri \sin(\mu) F(p, g) F(p, g)^* + 2 \ri \sin(\mu - \nu) C.
\end{split}
\label{AAA_2}
\ee
By multiplying this equation with the matrices
\be 
    e^{-\hbsTheta(p, g)} y^{-1} e^{\bsLambda(p)}
    \midand 
    e^{\bsLambda(p)} y e^{-\hbsTheta(p, g)}
\label{multipliers}
\ee
from the left and the right, respectively, one finds immediately that
\be
\begin{split}
    & e^{\ri \mu} 
        e^{-\hbsTheta(p, g)} y^{-1} e^{2 \bsLambda(p)} y e^{\hbsTheta(p, g)} 
    - e^{-\ri \mu}
        e^{\hbsTheta(p, g)} y^{-1} e^{2 \bsLambda(p)} y e^{-\hbsTheta(p, g)} 
    \\
    & \quad 
    = 2 \ri \sin(\mu) 
        (e^{-\hbsTheta(p, g)} y^{-1} e^{\bsLambda(p)} F(p, g)) 
        (e^{-\hbsTheta(p, g)} y^{-1} e^{\bsLambda(p)} F(p, g))^* 
    + 2 \ri \sin(\mu - \nu) C.
\end{split}
\label{commut_rel_and_y}
\ee
Now, focusing on the diagonal entries on the above matrix equation, it 
follows that $\forall k \in \bN_N$ we have
\be
    (y^{-1} e^{2 \bsLambda(p)} y)_{k, k} 
    = \vert (e^{-\hbsTheta(p, g)} y^{-1} e^{\bsLambda(p)} F(p, g))_k \vert^2.
\label{AAA_4}
\ee
Notice that the matrix $y^{-1} e^{2 \bsLambda(p)} y$ is positive definite, 
whence its diagonal entries are strictly positive. Therefore, we conclude 
that
\be
    (e^{-\hbsTheta(p, g)} y^{-1} e^{\bsLambda(p)} F(p, g))_k \neq 0
    \qquad
    (k \in \bN_N).   
\label{hF_components_nonzero}
\ee

To proceed further, $\forall a \in \bN_n$ we define the complex number
\be
    m_a 
    = \frac{(e^{-\hbsTheta(p, g)} y^{-1} e^{\bsLambda(p)} F(p, g))_a}
        {\vert 
            (e^{-\hbsTheta(p, g)} y^{-1} e^{\bsLambda(p)} F(p, g))_a 
        \vert}.
\label{m_a}
\ee 
Since $\vert m_a \vert = 1$, from \eqref{M} it is clear that
\be
    m = \diag(m_1, \ldots, m_n, m_1, \ldots, m_n) \in M
\label{m}
\ee
is well-defined, just as the group element
\be
    \hy(p, g) = y m \in K. 
\label{y_hat}
\ee
Since the subgroup $M$ \eqref{M} is the centralizer of $\mfa$ in $K$, from 
\eqref{y_def} it is immediate that
\be
    L(p, g) = \hy(p, g) e^{2 \hbsTheta(p, g)} \hy(p, g)^{-1}.    
\label{y_hat_OK_1}
\ee
Keeping in mind the definition \eqref{m_a}, it is also evident that 
$\forall a \in \bN_n$ we can write
\be
\begin{split}
    & (e^{-\hbsTheta(p, g)} \hy(p, g)^{-1} e^{\bsLambda(p)} F(p, g))_a
    = (e^{-\hbsTheta(p, g)} m^{-1} y^{-1} e^{\bsLambda(p)} F(p, g))_a
    \\
    & \quad
    = e^{-\hTheta_a(p, g)} m_a^{-1} (y^{-1} e^{\bsLambda(p)} F(p, g))_a
    = m_a^{-1} (e^{-\hbsTheta(p, g)} y^{-1} e^{\bsLambda(p)} F(p, g))_a
    \\
    & \quad
    = \vert (e^{-\hbsTheta(p, g)} y^{-1} e^{\bsLambda(p)} F(p, g))_a \vert 
    > 0,   
\end{split}
\label{y_hat_OK_2}
\ee
whence the matrix \eqref{y_hat} does meet the requirements we imposed in 
the Lemma.

As concerns the uniqueness part of the Lemma, in the light of the above 
calculations it is trivial. As a matter of fact, the smoothness of the 
function $\hy$ \eqref{hy_function} is also quite straightforward from 
the above construction. Since $\varUpsilon$ \eqref{varUpsilon} is a 
diffeomorphism, and since the smooth map $L$ \eqref{L} takes values in 
$\cP_\reg$, the map 
\be
    \varUpsilon^{-1} \circ L 
    \colon P \times \tmfM \rightarrow \mfc \times (K / M)
\label{varUpsilon_and_L}
\ee
is well-defined and smooth. Now, by composing
$\varUpsilon^{-1} \circ L$ from the left with appropriate \emph{smooth} local 
sections of the principal $M$-bundle
\be
    \mfc \times K \ni (X, k) \mapsto (X, k M) \in \mfc \times (K / M),
\label{M_bundle}
\ee
it is evident that the diagonalizing matrix $y$ appearing in \eqref{y_def} 
can be chosen smoothly in a small neighborhood of any point 
$(p, g) \in P \times \tmfM$. It entails that the complex numbers $m_a$ 
\eqref{m_a} also depend smoothly in the same neighborhood of $(p, g)$. 
However, by virtue of the uniqueness of $\hy$, at each point of this 
neighborhood it must coincide with the product given in \eqref{y_hat}. 
As a consequence, the function defined in \eqref{hy_function} is smooth 
in a small neighborhood of any point, whence it is smooth everywhere.
\end{proof}

\subsection{The form of the dual Lax matrix}
\label{SUBSECTION_The_form_of_the_dual_Lax_matrix}
With the aid of the functions $L$ \eqref{L} and $\hy$ \eqref{hy_function} 
let us define the matrix valued function
\be
    \hL \colon P \times \tmfM \rightarrow \cP_\reg,
    \quad
    (p, g) \mapsto \hL(p, g) = \hy(p, g)^{-1} e^{2 \bsLambda(p)} \hy(p, g).
\label{hL}
\ee 
Utilizing $F$ \eqref{F_def} and $\bsLambda$ \eqref{bsLambda}, together with 
the recently introduced functions $\hbsTheta$ \eqref{hbsTheta} and $\hy$ 
\eqref{hy_function}, we also define the column vector valued function
\be
    \hF \colon P \times \tmfM \rightarrow \bC^{N \times 1},
    \quad
    (p, g) 
    \mapsto 
    \hF(p, g) = e^{-\hbsTheta(p, g)} y(p, g)^{-1} e^{\bsLambda(p)} F(p, g).
\label{hF}
\ee
By virtue of \eqref{hF_components_nonzero} we have $\hF_k(p, g) \neq 0$ 
for all $k \in \bN_N$. Now, from the observations we made in the previous 
subsection (see \eqref{hbsTheta_smooth} and Lemma \ref{LEMMA_hy}) it is 
evident that
\be
    \hL \in C^\infty(P \times \tmfM, \cP_\reg)
    \midand
    \hF \in C^\infty(P \times \tmfM, \bC^{N \times 1}).
\label{hL_and_hF_smooth}
\ee
Remembering \eqref{commut_rel_and_y}, we also see that for all 
$g = (\mu, \nu) \in \tmfM$ the $g$-sections of these new objects obey the 
commutation relation
\be
    e^{\ri \mu} e^{-\hbsTheta^g} \hL^g e^{\hbsTheta^g}
    - e^{- \ri \mu} e^{\hbsTheta^g} \hL^g e^{-\hbsTheta^g}
    = 2 \ri \sin(\mu) \hF^g (\hF^g)^* + 2 \ri \sin(\mu - \nu) C.
\label{commut_rel_hat} 
\ee
To proceed further, it is expedient to introduce the shorthand notation
\be
    \hg = (\hmu, \hnu) = (-\mu, -\nu).
\label{g_hat}
\ee
Notice that the map
\be
    \tmfM \ni g \mapsto \hg \in \tmfM
\label{tmfM_involution}
\ee
is a well-defined involution on the space of the admissible parameters 
\eqref{tmfM}, i.e., 
\be
    \hat{\hg} = g.
\label{hathat} 
\ee
More importantly, by taking the matrix entries of the equation 
\eqref{commut_rel_hat}, we find immediately that
\be
    \hL^g_{k, l} 
    = \frac{\ri \sin(\hmu) \hF^g_k \bar{\hF}^g_l 
        + \ri \sin(\hmu - \hnu) C_{k, l}}
        {\sinh\big( \ri \hmu + \hTheta_k^g - \hTheta_l^g \big)}
    \qquad
    (k, l \in \bN_N).
\label{hL_entries}
\ee
Due to its striking similarity with \eqref{L_entries}, it is natural to 
expect intimate relationships between the matrix valued functions $L$ and 
$\hL$. All we need to find is the connection between the column vector valued 
functions $F$ and $\hF$. 

In order to reveal this missing relationship, we follow the same strategy 
we applied in our paper \cite{Pusztai_JPA2011} to understand the structure of
the Lax matrix of the rational $C_n$ RSvD system. As the first step, let us 
observe that by multiplying both sides of the commutation relation 
\eqref{commut_rel_hat} with the matrix $C$ \eqref{C} we get that
\be
    e^{\ri \mu} e^{\hbsTheta^g} (\hL^g)^{-1} e^{-\hbsTheta^g}
    - e^{- \ri \mu} e^{-\hbsTheta^g} (\hL^g)^{-1} e^{\hbsTheta^g}
    = 2 \ri \sin(\mu) (C \hF^g) (C \hF^g)^* + 2 \ri \sin(\mu - \nu) C.
\label{commut_rel_hat_inverse} 
\ee
Therefore, for the matrix entries of the inverse of $\hL^g$ we obtain at 
once that
\be
    (\hL^g)^{-1}_{k, l} 
    = \frac{\ri \sin(\hmu) (C \hF^g)_k \overline{(C \hF^g)}_l 
        + \ri \sin(\hmu - \hnu) C_{k, l}}
        {\sinh\big( \ri \hmu - (\hTheta_k^g - \hTheta_l^g) \big)}
    \qquad
    (k, l \in \bN_N). 
\label{hL_inv_entries}
\ee
Now, the main idea is that from the relationships between certain minors of 
$\hL^g$ and $(\hL^g)^{-1}$ provided by Jacobi's theorem (see e.g. Theorem 
2.5.2 in \cite{Prasolov}) we can deduce characterizing equations for the 
smooth functions
\be
    \hz_c = \hF_c \bar{\hF}_{n + c} \in C^\infty(P \times \tmfM, \bC)
    \qquad
    (c \in \bN_n).
\label{z_hat}
\ee
In complete analogy with equation (42) of \cite{Pusztai_JPA2011}, for 
any $g = (\mu, \nu) \in \tmfM$ and $c \in \bN_n$ it proves handy to introduce 
the temporary shorthand notations
\begin{align}
    & \cD^g_c 
    = \prod_{\substack{d = 1 \\ (d \neq c)}}^n \vert \hF^g_{n + d} \vert^2
        \prod_{\substack{a, b = 1 \\ (c \neq a \neq b \neq c)}}^n
        \frac{\sinh(\htheta^g_a - \htheta^g_b)}
            {\sinh(\ri \hmu + \htheta^g_a - \htheta^g_b)}
    \in C^\infty(P, \bR \setminus \{ 0 \}),
    \label{cD_c}
    \\
    & \omega^g_c = \prod_{\substack{d = 1 \\ (d \neq c)}}^n
        \frac{\sinh(\htheta^g_c - \htheta^g_d) 
                \sinh(\htheta^g_c + \htheta^g_d)}
            {\sinh(\ri \hmu + \htheta^g_c - \htheta^g_d) 
                \sinh(\ri \hmu + \htheta^g_c + \htheta^g_d)}
    \in C^\infty(P, \bC \setminus \{ 0 \}).
    \label{omega_c}
\end{align}
Furthermore, in order to compute the minors of $\hL^g$ and $(\hL^g)^{-1}$, 
we shall need an appropriate hyperbolic variant of the Cauchy type determinant 
formulae that we borrow from the paper \cite{Ruij_CMP1987}. Namely, by letting 
$\lambda \to \infty$ in equation (B28) of \cite{Ruij_CMP1987}, one can easily 
see that if $\alpha \in \bR$ such that $\sin(\alpha) \neq 0$, and if 
$m \in \bN$, $\xi_1, \ldots, \xi_m \in \bR$ and also 
$\eta_1, \ldots, \eta_m \in \bR$, then we can write
\be
    \det 
        \left( 
            \left[ 
                \frac{\sinh(\ri \alpha)}{\sinh(\ri \alpha + \xi_k - \eta_l)} 
            \right]_{1 \leq k, l \leq m} 
        \right)
    = \sinh(\ri \alpha)^m
    \frac{\prod_{1 \leq k < l \leq m} 
            \sinh(\xi_k - \xi_l) \sinh(\eta_l - \eta_k)}
        {\prod_{k, l = 1}^m \sinh(\ri \alpha + \xi_k - \eta_l)}.
\label{Cauchy_det_formula}
\ee
Finally, let us keep in mind that, if $m \in \bN$ and $k \in \bN_m$, then for 
any $m \times m$ matrix $A$ its $k \times k$ minor determinant corresponding 
to the rows $r_1, \ldots, r_k \in \bN_m$ and the columns 
$c_1, \ldots, c_k \in \bN_m$ is given by
\be
    A 
    \begin{pmatrix} 
        r_1 & \ldots & r_k \\
        c_1 & \ldots & c_k
    \end{pmatrix}
    = \det \left( [A_{r_a, c_b}]_{1 \leq a, b \leq k}\right).
\label{minor_determinant}
\ee
Now, we are in a position to present two particularly useful relationships
for the function $\hz_c$ \eqref{z_hat}. Before proving them, the reader may 
find it convenient to skim through the formulae appearing in the Appendix of 
\cite{Pusztai_JPA2011}.

\begin{PROPOSITION}
\label{PROPOSITION_z_hat_relations}
For any $c \in \bN_n$ and $g = (\mu, \nu) \in \tmfM$ the $g$-section of the 
function $\hz_c$ \eqref{z_hat} obey the equations
\begin{align}
    & \frac{\sin(\hmu)}{\sinh(\ri \hmu + 2 \htheta^g_c)} \omega^g_c \hz^g_c
    + \frac{\sin(\hmu)}{\sinh(\ri \hmu - 2 \htheta^g_c)} 
        \bar{\omega}^g_c \bar{\hz}^g_c
    + \frac{\sin(\hmu - \hnu)}{\sinh(\ri \hmu + 2 \htheta^g_c)}
    + \frac{\sin(\hmu - \hnu)}{\sinh(\ri \hmu - 2 \htheta^g_c)}
    = 0,
    \label{z_hat_linear_relation}
    \\
    & \sinh(2 \htheta^g_c)^2 \vert \omega^g_c \hz^g_c \vert^2
    - \sin(\hmu) \sin(\hmu - \hnu) 
        (\omega^g_c \hz^g_c + \bar{\omega}^g_c \bar{\hz}^g_c) 
    = \sin(\hmu)^2 + \sin(\hmu - \hnu)^2 + \sinh(2 \htheta^g_c)^2 .
    \label{z_hat_quadratic_relation}
\end{align}
\end{PROPOSITION}

\begin{proof}
We can be brief here, since our proof is modeled on the ideas presented in 
subsection 3.2 of \cite{Pusztai_JPA2011}. Fix an arbitrary 
$g = (\mu, \nu) \in \tmfM$ and let $c \in \bN_n$. From the definition of 
$\hL$ \eqref{hL} we know that $\hL^g$ takes values in $\exp(\mfp_\reg)$, 
whence it is clear that $\det(\hL^g) = 1$. Thus, recalling the notation
we introduced in \eqref{minor_determinant}, the application of Jacobi's
theorem, as formulated in Theorem A1 of \cite{Pusztai_JPA2011}, leads to 
the relationship
\be
    ((\hL^g)^{-1})^\top
    \begin{pmatrix} 
        1 & \ldots & c & \ldots & n \\
        1 & \ldots & n + c & \ldots & n
    \end{pmatrix}
    = -\hL^g
        \begin{pmatrix} 
        n + 1 & \ldots & n + c & \ldots & 2 n \\
        n + 1 & \ldots & c & \ldots & 2 n
    \end{pmatrix},
\label{BBB_1}
\ee
where $\top$ is the shorthand for taking transpose. Now, for brevity let 
$X^{(c)}$ denote the $n \times n$ matrix corresponding to the minor 
determinant on the left hand side of the above equation. By inspecting 
the entries of $(\hL^g)^{-1}$ \eqref{hL_inv_entries}, one finds immediately 
that
\be
    X^{(c)} 
    = Y^{(c)} 
        + \frac{\ri \sin(\hmu - \hnu)}
            {\sinh(\ri \hmu + 2 \htheta^g_c)} E_{c, c},
\label{BBB_2}
\ee
where $Y^{(c)} = [Y^{(c)}_{a, b}]_{1 \leq a, b \leq n}$ is a Cauchy type 
matrix with entries
\be
    Y^{(c)}_{a, b} 
    = \begin{cases}
        \bar{\hF}^g_{n + a} 
        \dfrac{\sinh(\ri \hmu)}
            {\sinh(\ri \hmu + \hTheta^g_a - \hTheta^g_b)} 
        \hF^g_{n + b},
        & \text{if $b \neq c$}, 
        \\
        \bar{\hF}^g_{n + a} 
        \dfrac{\sinh(\ri \hmu)}
            {\sinh(\ri \hmu + \hTheta^g_a - \hTheta^g_{n + c})} 
        \hF^g_{c},
        & \text{if $b = c$},
    \end{cases}
\label{BBB_3}
\ee
whereas $E_{c, c}$ is the $n \times n$ elementary matrix 
$E_{c, c} = [ \delta_{c, a} \delta_{c, b} ]_{1 \leq a, b \leq n}$. Keeping 
in mind the definitions \eqref{cD_c} and \eqref{omega_c}, the application 
of the determinant formula \eqref{Cauchy_det_formula} immediately yields
\be
    \det(Y^{(c)})
    = \frac{\sinh(\ri \hmu)}{\sinh(\ri \hmu + 2 \htheta^g_c)} 
        \cD^g_c \omega^g_c \hz^g_c.
\label{BBB_4}
\ee
Next, giving a glance at \eqref{BBB_2}, we see that matrix $X^{(c)}$ is 
actually a rank one perturbation of $Y^{(c)}$. As is known (see e.g. 
equation (A.6) in \cite{Pusztai_JPA2011}), in such cases we can write that
\be
    \det(X^{(c)}) 
    = \det(Y^{(c)}) 
    + \frac{\ri \sin(\hmu - \hnu)} {\sinh(\ri \hmu + 2 \htheta^g_c)} 
        \cC^{(c)}_{c, c},
\label{BBB_5}
\ee
where $\cC^{(c)}_{c, c}$ is the \emph{cofactor} of $Y^{(c)}$ associated with 
the entry $Y^{(c)}_{c, c}$. In other words, $\cC^{(c)}_{c, c}$ can be computed
by taking $(-1)^{c + c} = 1$ times the determinant of the 
$(n - 1) \times (n - 1)$ submatrix obtained by deleting the $c$-th row 
and the $c$-th column of $Y^{(c)}$. Since this submatrix is also of Cauchy 
type, by applying \eqref{Cauchy_det_formula} we get immediately that 
$\cC^{(c)}_{c, c} = \cD^g_c$. Thus, by putting the above formulae together, 
we end up with the expression
\be
    \det(X^{(c)})
    = \frac{\ri \sin(\hmu)}{\sinh(\ri \hmu + 2 \htheta^g_c)} 
        \cD^g_c \omega^g_c \hz^g_c
    + \frac{\ri \sin(\hmu - \hnu)} {\sinh(\ri \hmu + 2 \htheta^g_c)} \cD^g_c.
\label{BBB_6}
\ee
To proceed further, we turn to the study of the right hand side of 
\eqref{BBB_1}. Along the same lines as above, one obtains that
\be
    \hL^g
        \begin{pmatrix} 
        n + 1 & \ldots & n + c & \ldots & 2 n \\
        n + 1 & \ldots & c & \ldots & 2 n
    \end{pmatrix}
    = \overline{\det(X^{(c)})},
\label{BBB_7}
\ee
and so the relationship \eqref{BBB_1} can be rewritten as
\be
    \det(X^{(c)}) + \overline{\det(X^{(c)})} = 0.
\label{BBB_8}
\ee
Since $\cD^g_c$ \eqref{cD_c} is non-zero and real at each point of the phase
space, the relationship displayed in \eqref{z_hat_linear_relation} emerges at 
once as an immediate consequence of the equations \eqref{BBB_6} and 
\eqref{BBB_8}.

In order to get an independent relationship for $\hz_c$ \eqref{z_hat}, 
notice that by Jacobi's theorem we can also write that
\be
    ((\hL^g)^{-1})^\top
    \begin{pmatrix} 
        1 & \ldots & n & n + c \\
        1 & \ldots & n & n + c 
    \end{pmatrix}
    = \hL^g
        \begin{pmatrix} 
        n + 1 & \ldots & \widehat{n + c} & \ldots & 2 n \\
        n + 1 & \ldots & \widehat{n + c} & \ldots & 2 n
    \end{pmatrix},
\label{CCC_1}
\ee
where $\widehat{n + c}$ on the right means that the indicated row and column 
are omitted. To make it practical, let $Z^{(c)}$ denote the 
$(n + 1) \times (n + 1)$ matrix corresponding to the minor determinant on the
left hand side of \eqref{CCC_1}. Upon introducing the temporary shorthand 
notation
\be
    \xi_k = \begin{cases}
        \hTheta^g_k, & \text{if $1 \leq k \leq n$}, \\
        \hTheta^g_{n + c}, & \text{if $k = n + 1$},
    \end{cases}
\label{CCC_2}
\ee
together with
\be
    f_k = \begin{cases}
        \hF^g_{n + k}, & \text{if $1 \leq k \leq n$}, \\
        \hF^g_{c}, & \text{if $k = n + 1$},
    \end{cases}
\label{CCC_3}
\ee
one finds immediately that
\be
    Z^{(c)} 
    = S^{(c)} 
    + \frac{\ri \sin(\hmu - \hnu)}
        {\sinh(\ri \hmu + 2 \htheta^g_c)} 
    E_{c, n + 1} 
    + \frac{\ri \sin(\hmu - \hnu)}
        {\sinh(\ri \hmu - 2 \htheta^g_c)} 
    E_{n + 1, c},
\label{CCC_4}
\ee
where $S^{(c)}$ is a Cauchy type $(n + 1) \times (n + 1)$ matrix with entries
\be
    S^{(c)}_{k, l} 
    = \bar{f}_k \frac{\sinh(\ri \hmu)}{\sinh(\ri \hmu + \xi_k - \xi_k)} f_l
    \qquad
    (k, l \in \bN_{n + 1}),
\label{CCC_5}
\ee
whereas 
\be
    E_{c, n + 1} = [\delta_{c, k} \delta_{n + 1, l}]_{1 \leq k, l \leq n + 1}
    \midand 
    E_{n + 1, c} = [\delta_{n + 1, k} \delta_{c, l}]_{1 \leq k, l \leq n + 1}.
\label{CCC_elementary_matrices}
\ee
Now, by applying \eqref{Cauchy_det_formula}, it is straightforward to verify 
that
\be
    \det(S^{(c)}) 
    = \frac{\sinh(2 \htheta^g_c)^2}
        {\vert \sinh(\ri \hmu + 2 \htheta^g_c) \vert^2}
    \cD^g_c \vert \omega^g_c \hz^g_c \vert^2.
\label{CCC_6}
\ee
Of course, the computation of the determinant of $Z^{(c)}$ is a bit more 
subtle. Nevertheless, the relevant determinant formula for rank two 
perturbations of invertible Hermitian matrices can be also found in the 
Appendix of \cite{Pusztai_JPA2011}. Indeed, equation (A.8) in 
\cite{Pusztai_JPA2011} tells us that, upon introducing the complex valued
function
\be
    \alpha_c = \frac{\ri \sin(\hmu - \hnu)}{\sinh(\ri \hmu + 2 \htheta^g_c)},
\label{CCC_7}
\ee
for the determinant of $Z^{(c)}$ \eqref{CCC_4} we can write
\be
    \det(Z^{(c)}) 
    = \det(S^{(c)}) 
    + \alpha_c \cC^{(c)}_{c, n + 1} 
    + \bar{\alpha}_c \bar{\cC}^{(c)}_{c, n + 1} 
    + \vert \alpha_c \vert^2 
        \frac{\vert \cC^{(c)}_{c, n + 1} \vert^2 
            - \cC^{(c)}_{c, c} \cC^{(c)}_{n + 1, n + 1}}
            {\det(S^{(c)})},
\label{CCC_8}
\ee
where $\cC^{(c)}_{k, l}$ now stands for the cofactor of $S^{(c)}$ associated 
with the entry $S^{(c)}_{k, l}$ $(k, l \in \bN_{n + 1})$. Invoking
\eqref{Cauchy_det_formula} again, it is a routine exercise to verify that
\be
    \cC^{(c)}_{n + 1, c} = \bar{\cC}^{(c)}_{c, n + 1} 
    = - \frac{\ri \sin(\hmu)}
        {\sinh(\ri \hmu + 2 \htheta^g_c)} 
    \cD^g_c \omega^g_c \hz^g_c,
\label{CCC_9_1}
\ee
whilst the other two relevant cofactors are given by
\begin{align}
    & \cC^{(c)}_{c, c} 
    = \cD^g_c \vert \hF^g_c \vert^2 
        \prod_{\substack{d = 1 \\ (d \neq c)}}^n 
            \left\vert
                \frac{\sinh(\htheta^g_c + \htheta^g_d)}
                    {\sinh(\ri \hmu + \htheta^g_c + \htheta^g_d)}
            \right\vert^2,
    \label{CCC_9_2}
    \\
    & \cC^{(c)}_{n + 1, n + 1} 
    = \cD^g_c \vert \hF^g_{n + c} \vert^2 
        \prod_{\substack{d = 1 \\ (d \neq c)}}^n 
            \left\vert
                \frac{\sinh(\htheta^g_c - \htheta^g_d)}
                    {\sinh(\ri \hmu + \htheta^g_c - \htheta^g_d)}
            \right\vert^2.
    \label{CCC_9_3}
\end{align}
Inserting the above formulae into \eqref{CCC_8}, for the left hand side 
of \eqref{CCC_1} we obtain
\be
    \det(Z^{(c)})
    = \cD^g_c 
    \frac{
        \sinh(2 \htheta^g_c)^2 \vert \omega^g_c \hz^g_c \vert^2
        - \sin(\hmu) \sin(\hmu - \hnu) 
            (\omega^g_c \hz^g_c + \bar{\omega}^g_c \bar{\hz}^g_c)
        - \sin(\hmu - \hnu)^2}
        {\vert \sinh(\ri \hmu + 2 \htheta^g_c) \vert^2}.
\label{CCC_10}
\ee
As concerns the right hand side of \eqref{CCC_1}, from \eqref{hL_entries} it is
clear that the corresponding $(n - 1) \times (n - 1)$ matrix is of Cauchy
type, whence by invoking \eqref{Cauchy_det_formula} we obtain at once that
\be
    \hL^g
        \begin{pmatrix} 
        n + 1 & \ldots & \widehat{n + c} & \ldots & 2 n \\
        n + 1 & \ldots & \widehat{n + c} & \ldots & 2 n
    \end{pmatrix}
    = \cD^g_c.
\label{CCC_11}
\ee
Now, simply by plugging the formulae \eqref{CCC_10} and \eqref{CCC_11} into
\eqref{CCC_1}, we end up with the quadratic relationship 
\eqref{z_hat_quadratic_relation}.
\end{proof}

Our next goal is to solve the system of equations given in Proposition 
\ref{PROPOSITION_z_hat_relations} for the $g$-section of $\hz_c$ 
\eqref{z_hat}. By taking the real and the imaginary parts of the smooth 
function $\omega^g_c \hz^g_c$, we introduce the shorthand notation
\be
    R^g_c = \Real(\omega^g_c \hz^g_c) \in C^\infty(P) 
    \midand
    I^g_c = \Imag(\omega^g_c \hz^g_c) \in C^\infty(P)
    \qquad
    (c \in \bN_n).
\label{DDD_1}
\ee
Notice that \eqref{z_hat_linear_relation} can be cast into the form
\be
    \sin(\hmu) \cosh(2 \htheta^g_c) R^g_c 
        - \cos(\hmu) \sinh(2 \htheta^g_c) I^g_c 
        + \sin(\hmu - \hnu) \cosh(2 \htheta^g_c)
    = 0,
\label{DDD_2}
\ee
and so we get
\be
    R^g_c 
    = \cot(\hmu) \tanh(2 \htheta^g_c) I^g_c 
        - \frac{\sin(\hmu - \hnu)}{\sin(\hmu)}.
\label{R_and_I}
\ee
On the other hand, \eqref{z_hat_quadratic_relation} can be rewritten as
\be
    \sinh(2 \htheta^g_c)^2 ((R^g_c)^2 + (I^g_c)^2) 
        - 2 \sin(\hmu) \sin(\hmu -\hnu) R^g_c 
        - \sin(\hmu)^2 - \sin(\hmu - \hnu)^2 - \sinh(2 \htheta^g_c)^2
    = 0,
\label{DDD_4}
\ee
thus by exploiting \eqref{R_and_I}, for the smooth function $I^g_c$ we end 
up with 
\be
    (I^g_c)^2 - 2 \cos(\hmu) \sin(\hmu - \hnu) \coth(2 \htheta^g_c) I^g_c 
        + (\sin(\hmu - \hnu)^2 - \sin(\hmu)^2) \coth(2 \htheta^g_c) 
    = 0.
\label{I_quadratic_eq}
\ee
Solving this quadratic equation for $I_c^g$, it is clear that there is a 
function
\be
    S_c \colon P \times \tmfM \rightarrow \bR,
    \quad
    (p, g) \mapsto S_c(p, g) = S^g_c(p),
\label{S_c}
\ee
such that $\vert S_c \vert = 1$ and
\be
    I^g_c 
    = \left( 
        \cos(\hmu) \sin(\hmu - \hnu) 
        + S^g_c \sin(\hmu) \cos(\hmu - \hnu) 
    \right) 
    \coth(2 \htheta^g_c).
\label{I_and_S}
\ee
Note that, a priori, the `sign' $S^g_c$ may change from point to point in 
an uncontrolled manner. To cure the problem, in $\tmfM$ \eqref{tmfM} we 
define the open subset
\be
    \tmfM_0 = \{ g = (\mu, \nu) \in \tmfM \mid \cos(\mu - \nu) \neq 0 \},
\label{tmfM_0}
\ee
and for the restriction of the function $S_c$ \eqref{S_c} to 
$P \times \tmfM_0$ we introduce the notation $\tS_c$. The point is that, 
if $g = (\mu, \nu) \in \tmfM_0$, then from \eqref{I_and_S} we get
\be
    \tS^g_c 
    = \frac{\tanh(2 \htheta^g_c) I^g_c - \cos(\hmu) \sin(\hmu - \hnu)}
        {\sin(\hmu) \cos(\hmu - \hnu)}.
\label{tS_c}
\ee
From this formula it is also clear that the `sign' function $\tS_c$ is 
\emph{smooth}, i.e., 
\be
    \tS_c \in C^\infty(P \times \tmfM_0). 
\label{tS_c_smooth}
\ee
As a consequence, the function $\tS_c$ is constant on each connected 
component of the product manifold $P \times \tmfM_0$. Thus, keeping in mind 
\eqref{I_and_S}, \eqref{R_and_I}, \eqref{DDD_1} and \eqref{omega_c}, one 
finds immediately that on any connected component of $P \times \tmfM_0$ where 
$\tS_c = 1$ we have
\be
    \hz_c(p, g)
    = \frac{\sinh(\ri (2 \hmu - \hnu) + 2 \htheta^g_c(p))}
        {\sinh(2 \htheta^g_c(p))}    
    \prod_{\substack{d = 1 \\ (d \neq c)}}^n
        \frac{\sinh(\ri \hmu + \htheta^g_c(p) - \htheta^g_d(p)) 
                \sinh(\ri \hmu + \htheta^g_c(p) + \htheta^g_d(p))}
            {\sinh(\htheta^g_c(p) - \htheta^g_d(p)) 
                \sinh(\htheta^g_c(p) + \htheta^g_d(p))}.
\label{hz_bad}
\ee 
On the other hand, on a connected component of $P \times \tmfM_0$ where 
$\tS_c = -1$, we can write that
\be
    \hz_c(p, g)
    = - \frac{\sinh(\ri \hnu + 2 \htheta^g_c(p))}
        {\sinh(2 \htheta^g_c(p))}    
    \prod_{\substack{d = 1 \\ (d \neq c)}}^n
        \frac{\sinh(\ri \hmu + \htheta^g_c(p) - \htheta^g_d(p)) 
                \sinh(\ri \hmu + \htheta^g_c(p) + \htheta^g_d(p))}
            {\sinh(\htheta^g_c(p) - \htheta^g_d(p)) 
                \sinh(\htheta^g_c(p) + \htheta^g_d(p))}.
\label{hz_good}
\ee
To proceed further, the following trivial observation proves to be 
instrumental in selecting the correct form of $\hz_c$.

\begin{PROPOSITION}
\label{PROPOSITION_real_part_of_z}
The functions $z_c$ \eqref{z_a_def} and $\hz_c$ \eqref{z_hat} obey the 
equation
\be
    \sum_{c = 1}^n \Real(\hz_c) = \sum_{c = 1}^n \Real(z_c).
\label{real_part_of_z}
\ee
\end{PROPOSITION}

\begin{proof}
Looking back to \eqref{hF}, we see that 
$\hF = e^{-\hbsTheta} \hy^{-1} e^{\bsLambda} F$, and so
\be
    (C \hF)^* 
    = (e^{\hbsTheta} \hy^{-1} e^{-\bsLambda} C F)^*
    = (C F)^* e^{-\bsLambda} \hy e^{\hbsTheta}.
\label{EEE_1}
\ee
As a consequence, we obtain
\be
    \tr (\hF (C \hF)^*)
    = \tr (e^{-\hbsTheta} \hy^{-1} e^{\bsLambda} 
            F (C F)^* e^{-\bsLambda} \hy e^{\hbsTheta})
    = \tr (F (C F)^*).
\label{EEE_2}
\ee
However, keeping in mind the definition \eqref{F}, we can write
\be
\begin{split}
    \tr (F (C F)^*)
    & = \sum_{j = 1}^N F_j \overline{(C F)}_j 
    = \sum_{c = 1}^n 
        \left( 
            F_c \overline{(C F)}_c + F_{n + c} \overline{(C F)}_{n + c} 
        \right)
    \\
    & = \sum_{c = 1}^n (F_c \bar{F}_{n + c} + F_{n + c} \bar{F}_c)
    = \sum_{c = 1}^n (z_c + \bar{z}_c) 
    = 2 \sum_{c = 1}^n \Real(z_c).
\end{split}
\label{EEE_3}
\ee
Along the same lines, from the definition \eqref{z_hat} it is also evident 
that 
\be
    \tr (\hF (C \hF)^*) = 2 \sum_{c = 1}^n \Real(\hz_c), 
\label{EEE_4}
\ee
and so by \eqref{EEE_2} the proof is complete.
\end{proof}

\begin{LEMMA}
\label{LEMMA_form_of_z_hat}
For all $c \in \bN_n$ and $g = (\mu, \nu) \in \tmfM$ the $g$-section of the
function $\hz_c$ \eqref{z_hat} has the form
\be
    \hz^g_c
    = - \frac{\sinh(\ri \hnu + 2 \htheta^g_c)}
        {\sinh(2 \htheta^g_c)}    
    \prod_{\substack{d = 1 \\ (d \neq c)}}^n
        \frac{\sinh(\ri \hmu + \htheta^g_c - \htheta^g_d) 
                \sinh(\ri \hmu + \htheta^g_c + \htheta^g_d)}
            {\sinh(\htheta^g_c - \htheta^g_d) 
                \sinh(\htheta^g_c + \htheta^g_d)}.
\label{hz_OK}
\ee
\end{LEMMA}

\begin{proof}
Recalling \eqref{tmfM_0}, it is clear that the open set $P \times \tmfM_0$ 
is \emph{dense} in $P \times \tmfM$. Thus, by continuity, it is enough to 
prove that $\forall c \in \bN_n$ and $\forall g \in \tmfM_0$ the $g$-section 
of the smooth function $\hz_c$ is given by \eqref{hz_OK}. So, recalling our 
discussion surrounding the derivation of the formulae \eqref{hz_bad} and 
\eqref{hz_good}, it is sufficient to show that on each connected component 
of $P \times \tmfM_0$, for all $c \in \bN_n$ we have $\tS_c = -1$. 

Arguing by contradiction, let us \emph{suppose} that there is an index 
$c_0 \in \bN_n$ and a connected component $\cO_0$ of $P \times \tmfM_0$ such 
that on $\cO_0$ we have $\tS_{c_0} = 1$. Since the phase space $P$ \eqref{P} 
is connected, we have $\cO_0 = P \times \mfN_0$, where $\mfN_0$ is a connected 
component of $\tmfM_0$ \eqref{tmfM_0}. Now, let $x_0 \in \cO_0$ be arbitrary 
and define
\be
    \cN_0 = \{ c \in \bN_n \mid \tS_c(x_0) = 1 \}.
\label{cN_0}
\ee
Since $\cO_0$ is connected, $\cN_0$ is \emph{independent} of the choice 
of the point $x_0$. Moreover, it is a non-empty subset of $\bN_n$, since 
$c_0 \in \cN_0$.

Equipped with the above objects, let $g = (\mu, \nu) \in \mfN_0$, and 
keep it fixed. Furthermore, from the configuration space $Q$ \eqref{Q} take 
an arbitrary $\bsxi = (\xi_1, \ldots, \xi_n)$, and in the phase space 
\eqref{P} for all $r, s \in \bN$ define a point
\be
    p_{r, s} = (r \bsxi, s \bsxi) \in P.
\label{p_r_s}
\ee
Looking back to the definition \eqref{z_a}, it is clear that 
$\forall c \in \bN_n$ and $\forall s \in \bN$ we have the limit relation
\be
    z^g_c(p_{r, s}) \to - \exp(\ri \nu + (N - 2 c) \ri \mu)
    \qquad
    (r \to \infty).
\label{FFF_1}
\ee
Thus, it is straightforward to verify that
\be
    \sum_{c = 1}^n \Real(z^g_c(p_{r, s})) 
    \to 
    - \cos(\nu + (n - 1) \mu) \frac{\sin(n \mu)}{\sin(\mu)}
    \qquad
    (r \to \infty).
\label{sum_re_z_limit}
\ee

Turning to the study of $\hz^g_c$ \eqref{z_hat}, from the definition of 
$u^g_c$ \eqref{u} it is obvious that $\forall c \in \bN_n$ and 
$\forall s \in \bN$ we have $u^g_c(p_{r, s}) \to 1$ as $r \to \infty$, and 
so recalling \eqref{F} we obtain the limit relations
\be
    F^g_c(p_{r, s}) \to e^{s \xi_c / 2}
    \midand
    F^g_{n + c} (p_{r, s}) 
    \to 
    - e^{-s \xi_c / 2} e^{-\ri \nu - (N - 2 c) \ri \mu}
    \qquad
    (r \to \infty).
\label{FFF_3}
\ee
Looking back to the entries of the Lax matrix $L$ \eqref{L}, it is clear that
\be
    \cP_\reg \ni L(p_{r, s}, g) 
    \to 
    \diag(e^{s \xi_1}, \ldots, e^{s \xi_n}, 
        e^{-s \xi_1}, \ldots, e^{-s \xi_n})
        \in \cP_\reg
    \qquad
    (r \to \infty).
\label{L_limit}
\ee
Now, let us observe that the diffeomorphism $\varUpsilon$ \eqref{varUpsilon} 
allows us to write that
\be
    \diag(e^{s \xi_1}, \ldots, e^{s \xi_n}, 
        e^{-s \xi_1}, \ldots, e^{-s \xi_n})
    = \varUpsilon
        \big(
            \diag(s \xi_1, \ldots, s \xi_n, -s \xi_1, \ldots, -s \xi_n), 
            \bsone_N M
        \big).
\label{FFF_4}
\ee
Moreover, due to Lemma \ref{LEMMA_hy}, we can also write that
\be
    L(p_{r, s}, g) 
    = \hy(p_{r, s}, g) e^{2 \hbsTheta(p_{r, s}, g)} \hy(p_{r, s}, g)^{-1}
    = \varUpsilon 
        \big(
            2 \hbsTheta(p_{r, s}, g), \hy(p_{r, s}, g) M
        \big).
\label{FFF_5}
\ee
Thus, from \eqref{L_limit} it follows that
\be
    \big( 2 \hbsTheta(p_{r, s}, g), \hy(p_{r, s}, g) M \big)
    \to
    \big(
        \diag(s \xi_1, \ldots, s \xi_n, -s \xi_1, \ldots, -s \xi_n), 
        \bsone_N M
    \big)
    \qquad
    (r \to \infty).
\label{FFF_6}
\ee
In particular, $\forall c \in \bN_n$ and $\forall s \in \bN$ we obtain
\be
    \htheta_c(p_{r, s}, g) \to \half s \xi_c 
    \qquad
    (r \to \infty).
\label{htheta_limit}
\ee

Now, from the above observations the behavior of $\hz^g_c(p_{r, s})$ 
\eqref{z_hat} for $r \to \infty$ comes almost effortlessly. Indeed, 
recalling \eqref{cN_0}, from \eqref{hz_bad} it follows immediately that 
$\forall c \in \cN_0$ and $\forall s \in \bN$ we have
\be
    \hz^g_c(p_{r, s})
    \to
    \frac{\sinh(\ri (2 \hmu - \hnu) + s \xi_c)}
        {\sinh(s \xi_c)}    
    \prod_{\substack{d = 1 \\ (d \neq c)}}^n
        \frac{\sinh(\ri \hmu + s (\xi_c - \xi_d) / 2) 
                \sinh(\ri \hmu + s (\xi_c + \xi_d) / 2)}
            {\sinh(s (\xi_c - \xi_d) / 2) 
                \sinh(s (\xi_c + \xi_d) / 2)},
\label{hz_bad_limit}
\ee
whereas equation \eqref{hz_good} entails that 
$\forall c \in \bN_n \setminus \cN_0$ and $\forall s \in \bN$ we can write
\be
    \hz^g_c(p_{r, s})
    \to
    - \frac{\sinh(\ri \hnu + s \xi_c)}
        {\sinh(s \xi_c)}    
    \prod_{\substack{d = 1 \\ (d \neq c)}}^n
        \frac{\sinh(\ri \hmu + s (\xi_c - \xi_d) / 2) 
                \sinh(\ri \hmu + s (\xi_c + \xi_d) / 2)}
            {\sinh(s (\xi_c - \xi_d) / 2) 
                \sinh(s (\xi_c + \xi_d) / 2)}.
\label{hz_good_limit}
\ee
Keeping in mind the limit relation \eqref{sum_re_z_limit}, 
notice that $\forall s \in \bN$ from Proposition 
\ref{PROPOSITION_real_part_of_z} we can infer that 
\be
    - \cos(\nu + (n - 1) \mu) \frac{\sin(n \mu)}{\sin(\mu)}
    = \sum_{c \in \cN_0} 
        \Real \bigg( \lim_{r \to \infty} \hz^g_c(p_{r, s}) \bigg)
    + \sum_{c \in \bN_n \setminus \cN_0} 
        \Real \bigg( \lim_{r \to \infty} \hz^g_c(p_{r, s}) \bigg),
\label{key_limit_relation_r_to_infty}
\ee
where the limits on the right can be read off from \eqref{hz_bad_limit} and 
\eqref{hz_good_limit}.

To proceed further, in the following we shall study the behavior of equation 
\eqref{key_limit_relation_r_to_infty} for large values of $s$. From
\eqref{hz_bad_limit} it follows immediately that $\forall c \in \cN_0$ we 
have
\be
    \lim_{s \to \infty} \lim_{r \to \infty} \hz^g_c (p_{r, s})
    = \exp(\ri (2 \hmu - \hnu) + (N - 2 c) \ri \hmu),
\label{hz_bad_s_limit}
\ee
whilst from \eqref{hz_good_limit} we see that 
$\forall c \in \bN_n \setminus \cN_0$ we can write
\be
    \lim_{s \to \infty} \lim_{r \to \infty} \hz^g_c (p_{r, s})
    = - \exp(\ri \hnu + (N - 2 c) \ri \hmu).
\label{hz_good_s_limit}
\ee
Therefore, as $s \to \infty$, the above relationship 
\eqref{key_limit_relation_r_to_infty} yields
\be
    - \cos(\nu + (n - 1) \mu) \frac{\sin(n \mu)}{\sin(\mu)}
    = \sum_{c \in \cN_0} 
        \Real \big( e^{\ri (2 \hmu - \hnu) + (N - 2 c) \ri \hmu} \big)
    - \sum_{c \in \bN_n \setminus \cN_0} 
        \Real \big( e^{\ri \hnu + (N - 2 c) \ri \hmu} \big).
\label{key_limit_relation_s_to_infty}
\ee
It is straightforward to see that this equation is equivalent to 
\be
    \cos(\nu - \mu) \sum_{c \in \cN_0} \cos((2 c - N - 1) \mu) = 0.
\label{FFF_key}
\ee
However, since $g = (\mu, \nu) \in \mfN_0 \subset \tmfM_0$ \eqref{tmfM_0},
we have $\cos(\mu - \nu) \neq 0$, thus we can also write that
\be
    \sum_{c \in \cN_0} 
        \left(
            e^{\ri (2 c - N - 1) \mu} + e^{- \ri (2 c - N - 1) \mu}
        \right) 
    = 0.
\label{FFF_key_OK}
\ee
Now, notice that in the above sum the powers of $e^{\ri \mu}$ are all 
distinct. Recall also that our argument is valid for any 
$g = (\mu, \nu) \in \mfN_0$, which is an open and non-empty condition on 
$\mu$ as well. However, if $\mu$ varies in an open and non-empty subset 
of $\bR$, then the family of functions 
$\{ \mu \mapsto e^{\ri \ell \mu} \}_{\ell \in \bZ}$ is clearly linearly 
independent, thus \eqref{FFF_key_OK} expresses a contradiction. So, 
necessarily, $\cN_0 = \emptyset$, and the proof is complete.
\end{proof}

To proceed further, for each $c \in \bN_n$ we define 
$\hu_c = \vert \hz_c \vert$. Due to Lemma \ref{LEMMA_form_of_z_hat}, it is 
obvious that for all $g = (\mu, \nu) \in \tmfM$ we can write
\be
    \hu^g_c 
    = \left( 1 + \frac{\sin(\hnu)^2}{\sinh(2 \htheta^g_c)^2} \right)^\half
        \prod_{\substack{d = 1 \\ (d \neq c)}}^n
            \left( 
                1 + \frac{\sin(\hmu)^2}{\sinh(\htheta^g_c - \htheta^g_d)^2} 
            \right)^\half
            \left( 
                1 + \frac{\sin(\hmu)^2}{\sinh(\htheta^g_c + \htheta^g_d)^2} 
            \right)^\half
    \in C^\infty(P).
\label{hu}
\ee
Let us observe that $\hu_c > 1$. Recalling Lemma \ref{LEMMA_hy} and the 
definition \eqref{hF}, it is also clear that $\hF_c > 0$, whence the function
\be
    \hlambda_c \colon P \times \tmfM \rightarrow \bR,
    \quad
    (p, g) \mapsto \hlambda_c(p, g) = 2 \ln(\hF_c(p, g)) - \ln(\hu_c(p, g))
\label{hlambda}
\ee
is well-defined and smooth. Keeping in mind \eqref{z_hat}, it is obvious that
\be
    \hF_c = e^{\frac{\hlambda_c}{2}} \hu_c^\half
    \midand
    \hF_{n + c} 
    = \frac{1}{\hF_c} \bar{\hz}_c 
    = e^{-\frac{\hlambda_c}{2}} \bar{\hz}_c \hu_c^{-\half}.
\label{hF_components}
\ee
Now, with the aid of the functions $\htheta_c$ \eqref{htheta} and $\hlambda_c$ 
\eqref{hlambda}, for each $g \in \tmfM$ we introduce the \emph{smooth} map
\be
    \Psi^g \colon P \rightarrow P,
    \quad
    p 
    \mapsto 
    \Psi^g(p) 
    = (\htheta^g_1(p), \ldots, \htheta^g_n(p), 
        \hlambda^g_1(p), \ldots, \hlambda^g_n(p)).
\label{Psi}
\ee
Comparing the formulae \eqref{z_a} and \eqref{hz_OK}, it is clear that the 
components of $F$ \eqref{F} and $\hF$ \eqref{hF_components} look alike. As 
a consequence, now we can uncover the precise connection between the matrix 
entries of $L$ \eqref{L_entries} and $\hL$ \eqref{hL_entries}, too. Indeed, 
making use of $\Psi^g$ \eqref{Psi}, our observations can be formulated as 
follows.

\begin{THEOREM}
\label{THEOREM_L_and_hL}
For all $g \in \tmfM$, $c \in \bN_n$, and $j \in \bN_N$ we can write
\be
    \hz^g_c = z^{\hg}_c \circ \Psi^g 
    \midand
    \hF^g_j = F^{\hg}_j \circ \Psi^g,
\label{hz_and_z__hF_and_F}
\ee
whilst the matrix valued functions $L$ \eqref{L} and $\hL$ \eqref{hL} 
are related by
\be
    \hL^g = L^{\hg} \circ \Psi^g.
\label{L_and_hL}
\ee
\end{THEOREM}

Following the terminology of \cite{Ruij_CMP1988}, at this point we may call 
the matrix valued function $\hL$ \eqref{hL} the \emph{dual Lax matrix}, 
whereas the name \emph{duality map} seems adequate for $\Psi^g$ \eqref{Psi}. 
Of course, to fully justify these names, the map $\Psi^g$ \eqref{Psi} must 
meet some further conditions, that we wish to examine in the rest of the 
paper. Heading toward this goal, take an arbitrary $g \in \tmfM$ and keep it 
fixed. Recalling \eqref{varUpsilon}, \eqref{bsLambda}, \eqref{hbsTheta} and 
\eqref{Psi}, it is evident that $\hbsTheta^g = \bsLambda \circ \Psi^g$, whence 
\eqref{hy_1} can be rewritten as
\be
    L^g 
    = \hy^g e^{2 \bsLambda \circ \Psi^g} (\hy^g)^{-1}
    = \varUpsilon(2 \bsLambda \circ \Psi^g, \hy^g M).
\label{GGG_1}
\ee
Thus, by composing the above functions with $\Psi^{\hg}$, we obtain
\be
    L^g \circ \Psi^{\hg} 
    = \varUpsilon
        \left(
            2 \bsLambda \circ \Psi^g \circ \Psi^{\hg}, 
            (\hy^g \circ \Psi^{\hg}) M
        \right).
\label{GGG_2}
\ee
On the other hand, keeping in mind the involution \eqref{tmfM_involution} and 
the definition \eqref{hL}, the observation we made in \eqref{L_and_hL} entails
that
\be
    L^g \circ \Psi^{\hg} 
    = L^{\hat{\hg}} \circ \Psi^{\hg} 
    = \hL^{\hg}
    = (\hy^{\hg})^{-1} e^{2 \bsLambda} \hy^{\hg}
    = \varUpsilon(2 \bsLambda, (\hy^{\hg})^{-1} M).
\label{GGG_3}
\ee
Since $\varUpsilon$ \eqref{varUpsilon} is a diffeomorphism, the above two 
equations imply that
\be
    \bsLambda \circ \Psi^g \circ \Psi^{\hg} = \bsLambda.
\label{bsLambda_and_Psi}
\ee
Moreover, there is a unique smooth function $m^g \in C^\infty(P, M)$ such that
\be
    (\hy^{\hg})^{-1} = (\hy^g \circ \Psi^{\hg}) m^g.
\label{m_g}
\ee
By applying the involution \eqref{tmfM_involution} on the last two equations,
we get
\be
    \bsLambda \circ \Psi^{\hg} \circ \Psi^{g} 
    = \bsLambda \circ \Psi^{\hg} \circ \Psi^{\hat{\hg}}
    = \bsLambda
    \midand
    (\hy^{g})^{-1} 
    = (\hy^{\hat{\hg}})^{-1} 
    = (\hy^{\hg} \circ \Psi^{\hat{\hg}}) m^{\hg} 
    = (\hy^{\hg} \circ \Psi^{g}) m^{\hg}.
\label{GGG_4}
\ee
Thus, recalling \eqref{GGG_3} and \eqref{GGG_2}, it also follows that
\be
\begin{split}
    & L^g \circ \Psi^{\hg} \circ \Psi^g
    = (L^g \circ \Psi^{\hg}) \circ \Psi^g
    = ((\hy^{\hg})^{-1} \circ \Psi^g) 
        e^{2 \bsLambda \circ \Psi^g} 
        (\hy^{\hg} \circ \Psi^g)
    \\
    & \quad
    = m^{\hg} \hy^g e^{2 \bsLambda \circ \Psi^g} (\hy^g)^{-1} (m^{\hg})^{-1}
    = m^{\hg} L^g (m^{\hg})^{-1}.
\end{split}
\label{GGG_5}
\ee

To proceed further, take an arbitrary point $p = (\bsxi, \bseta) \in P$ and
let 
\be
    p' = (\bsxi', \bseta') = \Psi^{\hg} \circ \Psi^g(p) \in P.
\label{p_prime}
\ee
Applying \eqref{GGG_4}, we can clearly write that
\be
    \bsLambda(p') = \bsLambda(p),
\label{GGG_6}
\ee
so the definition \eqref{bsLambda} leads to the relation $\bsxi' = \bsxi$. 
Next, let us recall that the function $m^{\hg}$ takes values in the subgroup 
$M$ \eqref{M}, whence it is \emph{diagonal}. Therefore, due to the relationship 
\eqref{GGG_5}, for all $c \in \bN_n$ we can write that
\be
    L^g_{c, c}(p') = L^g_{c, c}(p).
\label{GGG_7}
\ee
Recalling \eqref{L_entries} and \eqref{F}, it is clear that 
\be
    L^g_{c, c}(p) = \vert F^g_c(p) \vert^2 = e^{\eta_c} u^g_c(p)
    \midand
    L^g_{c, c}(p') = \vert F^g_c(p') \vert^2 = e^{\eta'_c} u^g_c(p').
\label{GGG_8}
\ee
Since $\bsxi' = \bsxi$, from the definition \eqref{u} it is evident that
$u^g_c(p') = u^g_c(p)$. Consequently, from \eqref{GGG_7} we infer that
\be
    \eta'_c = \eta_c
    \qquad
    (c \in \bN_n).
\label{GGG_9} 
\ee
Now, putting the above observations together, we see that $p' = p$; that is,
\be
    \Psi^{\hg} \circ \Psi^g(p) = p
    \qquad
    (p \in P).
\label{GGG_10}
\ee
In other words, $\Psi^{\hg} \circ \Psi^g = \Id_P$. Applying the involution
\eqref{tmfM_involution} once more, we also obtain that
\be
    \Psi^g \circ \Psi^{\hg} = \Psi^{\hat{\hg}} \circ \Psi^{\hg} = \Id_P.
\label{GGG_11}
\ee

\begin{THEOREM}
\label{THEOREM_Psi_is_a_diffeomorphism}
For all $g \in \tmfM$ the map $\Psi^g$ \eqref{Psi} is a diffeomorphism with
inverse $(\Psi^g)^{-1} = \Psi^{\hg}$.
\end{THEOREM}

To sum up, using essentially algebraic techniques, in this section we proved
that the proposed duality map $\Psi^g$ \eqref{Psi} is a diffeomorphism.
To complete the study of $\Psi^g$, we still have to investigate its 
relationship with the symplectic form $\omega$ \eqref{omega}. Since our
approach is built upon the scattering theory of the van Diejen system 
\eqref{H}, using mainly analytical techniques, it is favorable to relegate 
this material to a separate section.

\section{Scattering theory and duality}
\label{SECTION_Scattering_theory_and_duality}
In this section we provide a rigorous treatment on the scattering theory of 
the particle system governed by the Hamiltonian function $H^g$ \eqref{H}. Our
analysis hinges on Ruijsenaars' theorem about the asymptotic properties of
certain exponential type matrix flows (see Theorem A2 in \cite{Ruij_CMP1988}).
Merging this result with our projection method presented in 
\cite{Pusztai_Gorbe}, we shall work out the asymptotic properties of the 
particle trajectories. By pushing forward this approach, the wave and the 
scattering maps also become accessible. As an added bonus, in this section
we shall complete the study of the self-duality property of the van Diejen
type systems \eqref{H}, too, by proving that the duality map $\Psi^g$ 
\eqref{Psi} is an anti-symplectomorphism. 

\subsection{Recapitulation of Ruijsenaars' theorem}
\label{SUBSECTION_Ruijsenaars_thm}
To make our presentation self-contained, in this subsection we briefly 
recapitulate Ruijsenaars' aforementioned theorem. Given a quadratic matrix 
\be
    M = [ M_{k, l} ]_{1 \leq k, l \leq N} \in \bC^{N \times N}, 
\label{matrix_M}
\ee   
for its leading principal minors we introduce the temporary notation
\be
    \pi_j(M) 
    = \det \left( [ M_{k, l} ]_{1 \leq k, l \leq j} \right) \in \bC
    \qquad
    (j \in \bN_N).
\label{pi_j}
\ee
Clearly the subset
\be
    \cM 
    = \{ M \in \bC^{N \times N} \mid \pi_j(M) \neq 0 \quad (j \in \bN_N) \}
\label{cM_def}
\ee
is open in $\bC^{N \times N}$, and the functions 
$m_j \colon \cM \rightarrow \bC, \, M \mapsto m_j(M)$ defined by the formulae
\be
    m_1(M) = M_{1, 1},
    \quad
    m_j(M) = \frac{\pi_j(M)}{\pi_{j - 1}(M)}
    \quad (2 \leq j \leq N)
\label{m_j}
\ee
are smooth. We also need the special family of diagonal matrices
\be
    \cD 
    = \{ D = \diag(d_1, \ldots, d_N) \in \bC^{N \times N} 
        \mid
        \Real(d_1) > \ldots > \Real(d_N) \},
\label{cD_def}
\ee
and with each pair $X = (M, D) \in \cM \times \cD$ we associate the 
exponential type matrix flow
\be
    E_X \colon \bR \rightarrow \bC^{N \times N},
    \quad
    s \mapsto E_X(s) = M e^{s D}.
\label{E_X}
\ee
Counting with multiplicities, let 
$\bslambda_1^X(s), \ldots, \bslambda_N^X(s)$ 
denote the not necessarily distinct eigenvalues of $E_X(s)$. For convenience, 
we shall assume that
\be
    \vert \bslambda_1^X(s) \vert 
        \geq \ldots \geq \vert \bslambda_N^X(s) \vert.
\label{E_X_evals_assumption}
\ee
Equipped with the above objects, we are now in a position to state 
Ruijsenaars' theorem in the form that is the most convenient in the later 
developments.

\begin{THEOREM}
\label{THEOREM_Ruijsenaars}
For each pair $X_0 = (M_0, D_0) \in \cM \times \cD$ there exist positive 
constants $T_0, C_0, R_0 > 0$ and a compact subset 
$\cK_0 \subset \cM \times \cD$ containing $X_0$ in its topological interior 
such that $\forall s \in [T_0, \infty)$ and $\forall X \in \cK_0$ we have
\be
    \vert \bslambda_1^X(s) \vert > \ldots > \vert \bslambda_N^X(s) \vert.
\label{THM_bslambda_ordering}
\ee
Moreover, $\forall s \in [T_0, \infty)$, $\forall X = (M, D) \in \cK_0$ and 
$\forall j \in \bN_N$ we can write
\be
    \bslambda_j^X(s) = m_j(M) e^{s D_{j, j}} \left( 1 + \rho_j^X(s) \right), 
\label{THM_bslambda}
\ee
whereas for its derivative with respect to $s$ we have
\be
    \dot{\bslambda}_j^X(s) 
    = m_j(M) e^{s D_{j, j}} 
        \left( D_{j, j} + D_{j, j} \rho_j^X(s) + \dot{\rho}_j^X(s) \right),
\label{THM_bslambda_dot}
\ee
where
\be
    \vert \rho_j^X(s) \vert \leq C_0 e^{-s R_0} \leq \half
    \midand 
    \vert \dot{\rho}_j^X(s) \vert \leq C_0 e^{-s R_0} \leq \half.
\label{THM_decay_cond}
\ee
\end{THEOREM}

We mention in passing that the theorem in its original form is sharper in 
the sense that in \cite{Ruij_CMP1988} we find finer estimates on the decay 
\eqref{THM_decay_cond} for each individual eigenvalue. For even sharper
estimates, the reader is kindly advised to study Theorem 
\ref{THEOREM_THM_A1} in Appendix \ref{AppA}.

\subsection{Application of Ruijsenaars' theorem}
\label{SUBSECTION_Applying_Ruijsenaars_thm}
Starting with this subsection, we wish to analyze the dynamics generated by 
the Hamiltonian function $H^g$ \eqref{H}. Notationwise, in the rest of the 
whole section we take an arbitrary $g = (\mu, \nu) \in \tmfM$ and keep it
fixed. Also, for the $g$-section of any function defined on $P \times \tmfM$
we shall omit the superscript $g$, i.e., we shall write $H \equiv H^g$, 
$\Psi \equiv \Psi^g$, $L \equiv L^g$, $\hL \equiv \hL^g$, etc.

Now, recall that the Hamiltonian vector field $\bsX_H \in \mfX(P)$ 
corresponding to the Hamiltonian $H$ \eqref{H} is defined by
\be
    \bsX_H [f] = \PB{f}{H}
    \qquad
    (f \in C^\infty(P)).
\label{bsX_H}
\ee
In our paper \cite{Pusztai_Gorbe} we proved that the vector field $\bsX_H$ is 
complete. In other words, if $p \in P$, $I \subseteq \bR$ is an open interval 
containing $0$, and
\be
    I \ni s \mapsto \gamma_{p}(s) \in P
\label{gamma} 
\ee
is the maximally defined integral curve of $\bsX_H$ satisfying the initial 
condition $\gamma_{p}(0) = p$, then we have $I = \bR$. This fact allows us 
to introduce the flow
\be
    \Phi \colon \bR \times P \rightarrow P,
    \quad
    (s, p) \mapsto \Phi(s, p) = \gamma_{p}(s),
\label{Phi}
\ee
that is a smooth map (see e.g. Theorem 9.12 in \cite{Lee}). 
Conforming with the standard convention, for any $s \in \bR$ and $p \in P$ 
we define the $s$-section $\Phi_s \colon P \rightarrow P$ and the $p$-section
$\Phi^p \colon \bR \rightarrow P$ by
\be
    \Phi_s(p) = \Phi^p(s) = \Phi(s, p).
\label{sections}
\ee
In the following we shall apply this notation for each function defined on 
$\bR \times P$. 

Next, for the natural extensions of the coordinates \eqref{x} onto 
$\bR \times P$ we introduce the notations 
\be
    \tx_j(s, p) = x_j(p)
    \qquad
    (j \in \bN_N, \, (s, p) \in \bR \times P).
\label{tx}
\ee 
Upon defining
\be
    t \colon \bR \times P \rightarrow \bR,
    \quad
    (s, p) \mapsto s,
\label{t}
\ee
that we call the time variable, it is evident that the family of smooth 
functions $t, \tx_1, \ldots, \tx_N$ provides a global coordinate system 
on the product manifold $\bR \times P$. Notice that the defining property 
of $\Phi$ \eqref{Phi} can be rephrased as
\be
    \PD{(f \circ \Phi)}{t} = \PB{f}{H} \circ \Phi
    \qquad
    (f \in C^\infty(P))
\label{Phi_dynamics}
\ee
and $\Phi_0 = \Id_P$. Occasionally, in the rest of the paper the partial 
differentiation with respect to $t$ will be denoted by a dot.

Equipped with the above objects, now we are in a position to make 
a closer inspection of the dynamics. Due to the completeness of $\bsX_H$
\eqref{bsX_H}, it is natural to inquire about the properties of the flow 
\eqref{Phi} for large values of $s$. In this respect our key observation 
is that the trajectories \eqref{gamma} can be recovered from the projection 
method we worked out in \cite{Pusztai_Gorbe}. More precisely, due to Theorem 
12 in \cite{Pusztai_Gorbe}, we have the spectral identification
\be
    \{ e^{\pm 2 \lambda_a \circ \Phi(s, p)} \mid a \in \bN_n \}
    = \Spec \left( e^{2 \bsLambda(p)} e^{s (L(p) - L(p)^{-1})} \right)
    \qquad
    (s \in \bR, \, p \in P).
\label{spectral_identification}
\ee
To put it simple, finding $\lambda_a \circ \Phi$ amounts to determining the 
eigenvalues of the matrix flow
\be
    \bR \ni s 
    \mapsto 
    e^{2 \bsLambda(p)} e^{s (L(p) - L(p)^{-1})} \in GL(N, \bC).
\label{matrix_flow}
\ee
As concerns the $\theta_a$ coordinate of $\Phi$, notice that the equation of 
motion \eqref{Phi_dynamics} entails
\be
    \PD{(\lambda_a \circ \Phi)}{t} = \PB{\lambda_a}{H} \circ \Phi 
    = \PD{H}{\theta_a} \circ \Phi = \sinh(\theta_a(\Phi)) u_a(\Phi).
\label{lambda_dot}
\ee
Therefore, since the function $u_a$ \eqref{u} is independent of the particle 
rapidities, from the knowledge of $\lambda_a \circ \Phi$ it comes effortlessly 
that 
\be
    \theta_a \circ \Phi 
    = \text{arcsinh} 
        \left( \frac{1}{u_a(\Phi)} \PD{(\lambda_a \circ \Phi)}{t} \right).
\label{theta_a_circ_Phi}
\ee
In the light of the above observations our plan is quite straightforward. 
Taking into account Ruijsenaars' theorem, we can get close control over 
the asymptotics of the eigenvalues of the matrix flow \eqref{matrix_flow}. 
Therefore, by exploiting \eqref{spectral_identification} and 
\eqref{theta_a_circ_Phi}, we can squeeze information about the asymptotic 
properties of the flow $\Phi$ \eqref{Phi} as well. The rest is technique.

In order to meet the requirements of Theorem \ref{THEOREM_Ruijsenaars}, we 
still have to diagonalize the exponent of the matrix flow \eqref{matrix_flow}. 
We can be brief here, since the same idea already appeared in Subsection 4.5 
our paper \cite{Pusztai_Gorbe}. Recalling Lemma \ref{LEMMA_hy} and definition 
\eqref{hL}, we can write that
\be
    L = \hy e^{2 \hbsTheta} \hy^{-1}
    \midand
    \hL = \hy^{-1} e^{2 \bsLambda} \hy,
\label{L_and_hL_diagonalization}
\ee
thus 
\be
    L - L^{-1} = 2 \hy \sinh(2 \hbsTheta) \hy^{-1}
    \midand
    e^{2 \bsLambda} = \hy \hL \hy^{-1}.
\label{G1}
\ee
It follows that $\forall s \in \bR$ and $\forall p \in P$ we have
\be
    \Spec \left( e^{2 \bsLambda(p)} e^{s (L(p) - L(p)^{-1})} \right)
    = \Spec \left( \hL(p) e^{2 s \sinh(2 \hbsTheta(p))} \right).
\label{spectral_identification_2}
\ee
To proceed further, for each $m \in \bN$ we introduce the $m \times m$ 
Hermitian matrix $\cR_m$ with entries
\be
    (\cR_m)_{k, l} = \delta_{k + l, m + 1}
    \qquad
    (1 \leq k, l \leq m),
\label{cR}
\ee 
and the $N \times N$ Hermitian matrix
\be
    W = \begin{bmatrix}
        \bsone_n & \bszero_n \\
        \bszero_n & \cR_n
    \end{bmatrix}.
\label{W}
\ee
Since $\cR_m^2 = \bsone_m$, both $\cR_m$ and $W$ are invertible with inverses
\be
    \cR_m^{-1} = \cR_m \midand W^{-1} = W.
\label{cR_and_W_inverse}
\ee
Next, upon introducing the matrix valued functions
\be
    \bsTheta^+ = 2 W \hbsTheta W^{-1}
    \midand
    \tilde{L} = W \hat{L} W^{-1},
\label{bsTheta_+_and_L_tilde}
\ee
from \eqref{spectral_identification} and \eqref{spectral_identification_2} 
we conclude that
\be
    \{ e^{\pm 2 \lambda_a \circ \Phi(s, p)} \mid a \in \bN_n \}
    = \Spec 
        \left( \tilde{L}(p) e^{2 s \sinh(\bsTheta^+(p))} \right)
    \qquad
    (s \in \bR, \, p \in P).
\label{spectral_identification_OK}
\ee

Now, two observations are in order. First, notice that by construction the 
matrix $\tilde{L}(p)$ is positive definite, whence for its leading principal 
minors we have $\pi_j(\tilde{L}(p)) > 0$ $(j \in \bN_N)$. Recalling 
\eqref{m_j}, it follows that $m_j(\tilde{L}(p)) > 0$. As a consequence, 
$\tilde{L}(p) \in \cM$ \eqref{cM_def}, and for each $a \in \bN_n$ the function
\be
    \lambda_a^+ \colon P \rightarrow \bR,
    \quad
    p \mapsto \half \ln(m_a(\tilde{L}(p)))
\label{lambda_+}
\ee
is well-defined and smooth. Second, from \eqref{bsTheta_+_and_L_tilde} it is 
also clear that
\be
    \bsTheta^+ 
    = \diag(\theta_1^+, \ldots, \theta_n^+, 
            - \theta_n^+, \ldots, -\theta_1^+),
\label{bsTheta_+}
\ee
where
\be
    \theta_a^+ = 2 \htheta_a \in C^\infty(P)
    \qquad
    (a \in \bN_n).
\label{theta_+}
\ee
As a consequence, $\forall p \in P$ we have $2 \sinh(\bsTheta^+(p)) \in \cD$ 
\eqref{cD_def}. The upshot of the above discussion is that, for each $p \in P$, 
Theorem \ref{THEOREM_Ruijsenaars} is directly applicable on the matrix flow
\be
    \bR \ni s
    \mapsto
    \tilde{L}(p) e^{2 s \sinh(\bsTheta^+(p))} \in GL(N, \bC),
\label{matrix_flow_OK}
\ee
and so by \eqref{spectral_identification_OK} we can obtain information about 
the asymptotic properties of the flow $\Phi$ \eqref{Phi}. To formulate our 
first result in this direction, for each $a \in \bN_n$ we introduce the smooth 
functions
\begin{align}
    & \cE_a \colon \bR \times P \rightarrow \bR,
    \quad
    (s, p) 
    \mapsto 
    \lambda_a(\Phi(s, p)) - s \sinh(\theta_a^+(p)) - \lambda_a^+(p),
    \label{cE_a}
    \\
    & \cF_a \colon \bR \times P \rightarrow \bR,
    \quad
    (s, p) 
    \mapsto 
    \sinh(\theta_a(\Phi(s, p))) u_a(\Phi(s, p)) - \sinh(\theta_a^+(p)).
    \label{cF_a}
\end{align}

\begin{LEMMA}
\label{LEMMA_asymptotics_1}
For each point $p_0 \in P$ there are strictly positive constants 
$T_0, C_0, R_0 > 0$ and a compact subset $K_0 \subset P$ containing $p_0$ in 
its interior such that $\forall s \in [T_0, \infty)$, $\forall p \in K_0$, 
and $\forall a \in \bN_n$ we have
\be
    \vert \cE_a(s, p) \vert \leq C_0 e^{-s R_0} \leq \half
    \midand 
    \vert \cF_a(s, p) \vert \leq C_0 e^{-s R_0} \leq \half.
\label{asymptotics_1}
\ee
\end{LEMMA}

\begin{proof}
From our earlier discussion it is clear that the map
\be
    \Gamma \colon P \rightarrow \cM \times \cD,
    \quad
    p \mapsto (\tilde{L}(p), 2 \sinh(\bsTheta^+(p)))
\label{Gamma}
\ee
is well-defined and smooth. Now, take an arbitrary point $p_0 \in P$. 
Associated with $X_0 = \Gamma(p_0)$ we also choose the non-negative constants 
$T_0, C_0, R_0 > 0$ and the compact subset $\cK_0 \subset \cM \times \cD$ 
containing $X_0$ in its interior, whose existence is guaranteed by Theorem 
\ref{THEOREM_Ruijsenaars}. Note that by the continuity of $\Gamma$ we can 
find a compact subset $K_0 \subset P$ such that
\be
    p_0 
    \in \text{int}(K_0) 
    \subset K_0 
    \subset \Gamma^{-1}(\text{int}(\cK_0)).
\label{K_0_OK}
\ee

To proceed further, let $s \in [T_0, \infty)$, $p \in K_0$ and $a \in \bN_n$
be arbitrary elements. Since by \eqref{K_0_OK} we have 
$X = \Gamma(p) \in \cK_0$, from \eqref{THM_bslambda} we infer that
\be
    e^{2 \lambda_a(\Phi(s, p))}
    = m_a(\tilde{L}(p)) e^{2 s \sinh(\theta_a^+(p))} 
        \left( 1 + \rho^X_a(s) \right),
\label{asymptotics_1_exp_form}
\ee
whereas \eqref{THM_bslambda_dot} leads to
\be
    \PD{e^{2 \lambda_a \circ \Phi}}{t} (s, p)
    = m_a(\tilde{L}(p)) e^{2 s \sinh(\theta_a^+(p))} 
        \left( 
            2 \sinh(\theta_a^+(p)) (1 + \rho^X_a(s)) 
            + \dot{\rho}_a^X (s)
        \right).
\label{asymptotics_2_exp_form}
\ee
From \eqref{asymptotics_1_exp_form} and \eqref{asymptotics_2_exp_form}
it is clear that both $\rho^X_a(s)$ and $\dot{\rho}_a^X (s)$ are real 
numbers. Moreover, due to the estimates displayed in \eqref{THM_decay_cond}, 
they are of magnitude less than or equal to $1 / 2$. So, by taking the 
logarithm of \eqref{asymptotics_1_exp_form}, from the definitions 
\eqref{lambda_+} and \eqref{cE_a} it is straightforward to verify that
\be
    \cE_a(s, p) = \half \ln(1 + \rho_a^X(s)),
\label{C1}
\ee
and so we may write
\be
    \vert \cE_a(s, p) \vert 
    = \half \vert \ln(1 + \rho_a^X(s)) \vert
    \leq \vert \rho_a^X(s) \vert 
    \leq C_0 e^{-s R_0}
    \leq \half.
\label{cE_estim}
\ee

Next, from the equation of motion \eqref{lambda_dot} and the relationship
\eqref{asymptotics_1_exp_form} we see that
\be
\begin{split}
    \PD{e^{2 \lambda_a \circ \Phi}}{t} (s, p)
    & = 2 \PD{(\lambda_a \circ \Phi)}{t} (s, p) e^{2 \lambda_a(\Phi(s, p))}
    = 2 \sinh(\theta_a(\Phi(s, p))) u_a(\Phi(s, p)) 
        e^{2 \lambda_a(\Phi(s, p))}
    \\
    & = 2 \sinh(\theta_a(\Phi(s, p))) u_a(\Phi(s, p))
    m_a(\tilde{L}(p)) e^{2 s \sinh(\theta_a^+(p))} 
        \left( 1 + \rho^X_a(s) \right).
\end{split}
\label{C2}
\ee
Notice that by plugging this formula into the left hand side of 
\eqref{asymptotics_2_exp_form} we obtain that
\be
    \sinh(\theta_a(\Phi(s, p))) u_a(\Phi(s, p))
    = \sinh(\theta_a^+(p)) + \half \frac{\dot{\rho}_a^X(s)}{1 + \rho_a^X(s)}.
\label{C3}
\ee
Thus, recalling the definition \eqref{cF_a}, we end up with the estimate
\be
    \vert \cF_a(s, p) \vert 
    = \half \frac{\vert \dot{\rho}_a^X(s) \vert}{1 + \rho_a^X(s)}
    \leq \vert \dot{\rho}_a^X(s) \vert
    \leq C_0 e^{-s R_0}
    \leq \half,
\label{cF_estim}
\ee
and so the proof is complete.
\end{proof}

\subsection{Asymptotics along the trajectories}
\label{SUBSECTION_Asymptotics}
Somewhat informally, Lemma \ref{LEMMA_asymptotics_1} tells us that the time 
evolution of the particle positions is asymptotically linear. However, 
looking back to the definition \eqref{cF_a}, the asymptotic characterization 
of the particle rapidities given by Lemma \ref{LEMMA_asymptotics_1} is quite 
implicit. To make it explicit, we need control over the value of the 
function $u_a$ \eqref{u} along the trajectories for large values of time.
\emph{As a rule for this subsection, we find it convenient to fix an 
arbitrary point $p_0 \in P$, together with the corresponding non-negative 
constants $T_0$, $C_0$, $R_0$ and the compact subset $K_0$ given by Lemma 
\ref{LEMMA_asymptotics_1}.} Furthermore, we introduce
\be
    l_0 
    = \max \{ \vert \lambda_a^+(p) \vert \mid a \in \bN_n, \, p \in K_0 \} 
        \geq 0,
\label{l_0}
\ee
and
\be
    \varDelta_0 
    = \min \{ 
        \min \{ 
            \theta_1^+(p) - \theta_2^+(p),
            \ldots,
            \theta_{n - 1}^+(p) - \theta_n^+(p),
            2 \theta_n^+(p)
        \} 
        \mid p \in K_0
    \} > 0.
\label{varDelta_0}
\ee
Since $K_0$ is compact, by the extreme value theorem both constants are 
well-defined.

\begin{PROPOSITION}
\label{PROPOSITION_particle_repulsion}
There is a constant $T_1 \geq T_0$ such that if $s \in [T_1, \infty)$,
$p \in K_0$, $a, b, c \in \bN_n$ and $a \neq b$, then
\be
    \vert \lambda_a(\Phi(s, p)) - \lambda_b(\Phi(s, p)) \vert
    \geq s \frac{\varDelta_0}{2} 
    \geq \ln(\sqrt{2})
    \midand
    2 \lambda_c(\Phi(s, p)) \geq s \frac{\varDelta_0}{2} \geq \ln(\sqrt{2}).
\label{particle_repulsion}
\ee
\end{PROPOSITION}

\begin{proof}
Let $p \in K_0$, $a, b, c \in \bN_n$ and suppose that $a \neq b$. Recalling
\eqref{varDelta_0}, it is obvious that
\be
    \vert \theta_a^+(p) - \theta_b^+(p) \vert \geq \varDelta_0
    \midand
    2 \theta_c^+(p) \geq \varDelta_0.
\label{UU1}
\ee
Thus, remembering the definitions \eqref{cE_a} and \eqref{l_0}, from Lemma 
\ref{LEMMA_asymptotics_1} it is clear $\forall s \in [T_0, \infty)$ we can 
write that
\be
\begin{split}
    & \vert \lambda_a(\Phi(s, p)) - \lambda_b(\Phi(s, p)) \vert
    \\
    & \quad = 
    \left\vert s(\sinh(\theta_a^+(p)) - \sinh(\theta_b^+(p))) 
        + \lambda_a^+(p) - \lambda_b^+(p) + \cE_a(s, p) - \cE_b(s, p)
    \right\vert
    \\
    & \quad \geq
    s \vert \sinh(\theta_a^+(p)) - \sinh(\theta_b^+(p)) \vert 
        - \vert \lambda_a^+(p) \vert - \vert \lambda_b^+(p) \vert 
        - \vert \cE_a(s, p) \vert - \vert \cE_b(s, p) \vert
    \\
    & \quad \geq
    s \varDelta_0 - 2 l_0 - 1 
    = s \frac{\varDelta_0}{2} 
        + \frac{\varDelta_0}{2} 
            \left( s - \frac{4 l_0 + 2}{\varDelta_0} \right).
\end{split}
\label{UU2}
\ee
Along the same lines, one can easily infer that
\be
    2 \lambda_c(\Phi(s, p)) 
    \geq 
    s \frac{\varDelta_0}{2} 
    + \frac{\varDelta_0}{2} 
        \left( s - \frac{4 l_0 + 2}{\varDelta_0} \right).
\label{UU3}
\ee
So, with the constant
\be
    T_1 = \max \{ T_0, (4 l_0 + 2) / \varDelta_0, \ln(2) / \varDelta_0 \}
\label{T_1}
\ee
the Proposition follows.
\end{proof}

As an important remark, note that \emph{in the rest of this subsection 
$T_1$ shall always stand for the constant provided by Proposition 
\ref{PROPOSITION_particle_repulsion}.} 

Before going into the study of the asymptotic properties of $u_a$ \eqref{u},
notice that $\forall y \in \bR$ satisfying $\vert y \vert \geq \ln(\sqrt{2})$
we have the trivial inequalities
\be
    \frac{1}{\vert \sinh(y) \vert} \leq 4 e^{-\vert y \vert} \leq 2 \sqrt{2},
    \quad
    \frac{1}{\sinh(y)^2} \leq 16 e^{- 2 \vert y \vert} \leq 8,
    \quad
    \vert \coth(y) \vert \leq 3.
\label{trivial_ineq}
\ee

\begin{PROPOSITION}
\label{PROPOSITION_u_asymptotics}
There is a constant $C_1 \geq 0$ such that $\forall s \in [T_1, \infty)$,
 $\forall p \in K_0$, $\forall a \in \bN_n$ we have
\be
    u_a(\Phi(s, p)) \leq 1 + C_1 e^{-s \varDelta_0}.
\label{u_asymptotics}
\ee
\end{PROPOSITION}

\begin{proof}
Since $\forall y \geq 0$ we have $\ln(1 + y) \leq y$, by taking the
logarithm of $u_a$ $(a \in \bN_n)$ we can write that
\be
    0 < \ln(u_a) 
    \leq 
    \half \frac{1}{\sinh(2 \lambda_a)^2} 
    + \half \sum_{\substack{c = 1 \\ (c \neq a)}}^n 
        \left( 
            \frac{1}{\sinh(\lambda_a - \lambda_c)^2} 
            + \frac{1}{\sinh(\lambda_a + \lambda_c)^2} 
        \right).
\label{EE1}
\ee
Thus, if $s \geq T_1$ and $p \in K_0$, then from Proposition 
\ref{PROPOSITION_particle_repulsion} and the inequalities appearing in
\eqref{trivial_ineq} we can easily infer that
\be
    \ln(u_a(\Phi(s, p))) 
    \leq 
    \half \frac{1 + 2 (n - 1)}{\sinh \left(s \frac{\varDelta_0}{2} \right)^2}
    \leq
    8 (N - 1) e^{- s \varDelta_0}
    \leq 4 (N - 1).
\label{EE2}
\ee
Since the exponential function is convex on the interval $[0, 4 (N - 1)]$, 
with the non-negative constant
\be
    C_1 = 2 (e^{4 (N - 1)} - 1)
\label{EE3}
\ee
we can write that
\be
    1 < u_a(\Phi(s, p)) 
    \leq \exp \left( 8 (N - 1) e^{- s \varDelta_0} \right) 
    \leq 1 + C_1 e^{- s \varDelta_0},
\label{EE4}
\ee
whence the proof is complete.
\end{proof}

During the later developments we shall also need control over the partial 
derivatives of the smooth functions $u_a$ \eqref{u} along the trajectories. 
For brevity, we introduce the notations
\be
    \cL^a_b = \PD{\ln(u_a)}{\lambda_b} \in C^\infty(P)
    \midand
    \cL^a_{b, c} = \PD{\cL^a_b}{\lambda_c} 
    = \frac{\partial^2 \ln(u_a)}{ \partial \lambda_b \partial \lambda_c}
        \in C^\infty(P)
    \qquad
    (a, b, c \in \bN_n).
\label{cL}
\ee
Utilizing the auxiliary function
\be
    \Xi(y, \alpha) 
    = \frac{\sin(\alpha)^2 \coth(y)}{\sin(\alpha)^2 + \sinh(y)^2}
    \qquad
    (y \in \bR \setminus \{ 0 \}, \, \alpha \in \bR \setminus \bZ \pi),
\label{Xi}
\ee
it is straightforward to verify that
\be
    \cL_a^a 
    = - 2 \Xi(2 \lambda_a, \nu) 
    - \sum_{\substack{d = 1 \\ (d \neq a)}}^n 
        \left( 
            \Xi(\lambda_a - \lambda_d, \mu) + \Xi(\lambda_a + \lambda_d, \mu)
        \right)
    \qquad
    (a \in \bN_n),
\label{cL_a_a_explicit}
\ee
while
\be
    \cL^a_b 
    = - \Xi(\lambda_b - \lambda_a, \mu) - \Xi(\lambda_b + \lambda_a, \mu)
    \qquad
    (a, b \in \bN_n, \, a \neq b).
\label{cL_a_b_explicit}
\ee
As an immediate consequence of \eqref{trivial_ineq}, let us also observe that
\be
    \vert \Xi(y, \alpha) \vert 
    \leq
    \frac{\vert \coth(y) \vert}{\sinh(y)^2}
    \leq 48 e^{- 2 \vert y \vert}
    \qquad
    (\vert y \vert \geq \ln(\sqrt{2})).
\label{Xi_estim}
\ee
Thus, by Proposition \ref{PROPOSITION_particle_repulsion} the following
result is immediate.

\begin{PROPOSITION}
\label{PROPOSITION_cL_a_b_estim}
There is a constant $C_2 \geq 0$ such that $\forall s \in [T_1, \infty)$,
$\forall p \in K_0$, $\forall a, b \in \bN_n$ we have
\be
    \vert \cL^a_b (\Phi(s, p)) \vert \leq C_2 e^{-s \varDelta_0}.
\label{cL_a_b_estim}
\ee
\end{PROPOSITION}

\begin{proof}
Let $s \in [T_1, \infty)$ and $p \in K_0$ be arbitrary elements, and also
introduce the temporary shorthand notation 
\be
    p' = \Phi(s, p) \in P. 
\label{NN_p_prime}
\ee    
Keeping in mind Proposition \ref{PROPOSITION_particle_repulsion} and the 
estimate \eqref{Xi_estim}, from \eqref{cL_a_a_explicit} it is clear that 
$\forall a \in \bN_n$ we can write that
\be
\begin{split}
    \vert \cL^a_a(\Phi(s, p)) \vert
    & \leq
    2 \vert \Xi(2 \lambda_a(p'), \nu) \vert 
    + \sum_{\substack{d = 1 \\ (d \neq a)}}^n 
        \left( 
            \vert 
                \Xi(\lambda_a(p') - \lambda_d(p'), \mu) 
            \vert
            + 
            \vert
                \Xi(\lambda_a(p') + \lambda_d(p'), \mu)
            \vert
        \right)
    \\
    & \leq
    96 e^{-2 \vert 2 \lambda_a(p') \vert}
    + 48 \sum_{\substack{d = 1 \\ (d \neq a)}}^n 
        \left( 
            e^{-2 \vert \lambda_a(p') - \lambda_d(p') \vert}
            + e^{-2 \vert \lambda_a(p') + \lambda_d(p') \vert}
        \right)
    \leq 48 N e^{-s \varDelta_0}.
\end{split}
\label{NN1}
\ee
On the other hand, if $a, b \in \bN_n$ and $a \neq b$, then 
\eqref{cL_a_b_explicit} yields
\be
\begin{split}
    \vert \cL^a_b(\Phi(s, p)) \vert
    & \leq
    \vert \Xi(\lambda_b(p') - \lambda_a(p'), \mu) \vert
    + \vert \Xi(\lambda_b(p') + \lambda_a(p'), \mu) \vert
    \\
    & \leq
    48 e^{-2 \vert \lambda_b(p') - \lambda_a(p') \vert} 
    + 48 e^{-2 \vert \lambda_b(p') + \lambda_a(p') \vert}
    \leq 
    96 e^{-s \varDelta_0}.
\end{split}
\label{NN2}
\ee
Therefore, with $C_2 = 48 N$ the Proposition follows.
\end{proof}

Turning to the second order partial derivatives of $\ln(u_a)$, bear in mind
that the partial derivative of the function $\Xi$ \eqref{Xi} with respect to 
its first variable $y$ is given by
\be
    \Xi'(y, \alpha)
    = - \frac{2 \sin(\alpha)^2 \cosh(y)^2}{(\sin(\alpha)^2 + \sinh(y)^2)^2}
    - \frac{\sin(\alpha)^2}{\sin(\alpha)^2 + \sinh(y)^2} \frac{1}{\sinh(y)^2}.
\label{Xi_der}
\ee
Making use of this derivative, from \eqref{cL_a_a_explicit} and 
\eqref{cL_a_b_explicit} we find immediately that
\be
    \cL^a_{a, a} 
    = - 4 \Xi'(2 \lambda_a, \nu)
    - \sum_{\substack{d = 1 \\ (d \neq a)}}^n 
        \left( 
            \Xi'(\lambda_a - \lambda_d, \mu) 
            + \Xi'(\lambda_a + \lambda_d, \mu)
        \right)
    \qquad
    (a \in \bN_n).
\label{cL_a_a_a}
\ee
Notice that the remaining second order partial derivatives take much simpler 
forms. Indeed, for all $a, b \in \bN_n$ satisfying $a \neq b$ we have
\begin{align}
    & \cL^a_{b, b}  
    = - \Xi'(\lambda_b - \lambda_a, \mu) - \Xi'(\lambda_b + \lambda_a, \mu),
    \label{cL_a_b_b}
    \\
    & \cL^a_{a, b} = \cL^a_{b, a} 
    = \Xi'(\lambda_b - \lambda_a, \mu) - \Xi'(\lambda_b + \lambda_a, \mu),
    \label{cL_a_a_b}
\end{align}
whereas for all $a, b, c \in \bN_n$ satisfying $c \neq a \neq b \neq c$ we
can write that 
\be
    \cL^a_{b, c} = \cL^a_{c, b} = 0.
\label{cL_a_b_c}
\ee
Now, in order to cook up convenient estimates for the above objects, let us
observe that by applying \eqref{trivial_ineq} on \eqref{Xi_der} we obtain
\be
    \vert \Xi'(y, \alpha) \vert
    \leq \frac{2 \cosh(y)^2}{\sinh(y)^4} + \frac{1}{\sinh(y)^2}
    = \frac{2 \coth(y)^2 + 1}{\sinh(y)^2}
    \leq 304 e^{-2 \vert y \vert}
    \qquad
    (\vert y \vert \geq \ln(\sqrt{2})).
\label{Xi_der_estim}
\ee
Therefore, without going into the details, the same ideas we presented in 
the proof of Proposition \ref{PROPOSITION_cL_a_b_estim} lead to the following 
result immediately.

\begin{PROPOSITION}
\label{PROPOSITION_cL_a_b_c_estim}
There is a constant $C_3 \geq 0$ such that $\forall s \in [T_1, \infty)$,
$\forall p \in K_0$, $\forall a, b, c \in \bN_n$ we have
\be
    \vert \cL^a_{b, c} (\Phi(s, p)) \vert \leq C_3 e^{-s \varDelta_0}.
\label{cL_a_b_c_estim}
\ee
\end{PROPOSITION}

Picking an arbitrary $a \in \bN_n$, at this point we are in a position to 
measure effectively the deviation of the rapidity $\theta_a \circ \Phi$ 
from $\theta_a^+$ \eqref{theta_+}. So, for convenience, we define the smooth 
function
\be
    \cG_a \colon \bR \times P \rightarrow \bR,
    \quad
    (s, p) \mapsto \theta_a(\Phi(s, p)) - \theta_a^+(p).
\label{cG_a}
\ee
Furthermore, with the aid of the Hamiltonian $H$ \eqref{H}, we introduce the 
strictly positive constant
\be
    \cH_0 = \max \{ H(p) \mid p \in K_0 \} > 0,
\label{cH_0}
\ee
which is well-defined by the compactness of $K_0$. Since $H$ is a first 
integral, let us note that $\forall s \in \bR$, $\forall p \in K_0$, 
$\forall a \in \bN_n$ we can write that
\be
\begin{split}
    \vert \theta_a(\Phi(s, p)) \vert
    & \leq
    \sinh(\vert \theta_a(\Phi(s, p)) \vert)
    \leq
    \cosh(\theta_a(\Phi(s, p)))
    \\
    & \leq H(\Phi(s, p))
    = H(\Phi(0, p)) = H(p) \leq \cH_0.
\end{split}
\label{theta_estim_0}
\ee
Before formulating our explicit result about the asymptotic properties of the 
rapidities, it proves convenient to bring in the new constant
\be
    \delta_0 = \min \{ R_0, \varDelta_0 \} > 0.
\label{delta_0}
\ee

\begin{LEMMA}
\label{LEMMA_asymptotics_2}
There is a constant $C_4 \geq 0$ such that $\forall s \in [T_1, \infty)$, 
$\forall p \in K_0$, $\forall a \in \bN_n$ we can write
\be
    \vert \cG_a(s, p) \vert \leq C_4 e^{- s \delta_0}.
\label{asymptotics_2}  
\ee
\end{LEMMA}

\begin{proof}
Let $s \in [T_1, \infty)$, $p \in K_0$, $a \in \bN_n$ be arbitrary elements. 
Keeping in mind \eqref{cG_a} and \eqref{cF_a}, it is clear that 
\be
\begin{split}
    \vert \cG_a(s, p) \vert 
    & \leq \vert \sinh(\theta_a(\Phi(s, p))) - \sinh(\theta_a^+(p)) \vert
    \\    
    & \leq \vert \cF_a(s, p) \vert 
        + \vert \sinh(\theta_a(\Phi(s, p))) (1 - u_a(\Phi(s, p))) \vert.
\end{split}
\label{BB1}
\ee
However, recalling \eqref{delta_0}, by Lemma \ref{LEMMA_asymptotics_1} we can 
write
\be
    \vert \cF_a(s, p) \vert \leq C_0 e^{- s R_0} \leq C_0 e^{- s \delta_0},
\label{BB2}
\ee 
whereas from Proposition \ref{PROPOSITION_u_asymptotics} and the observation 
we made in \eqref{theta_estim_0} it is immediate that
\be
\begin{split}
    & \vert \sinh(\theta_a(\Phi(s, p))) (1 - u_a(\Phi(s, p))) \vert
    = \sinh(\vert \theta_a(\Phi(s, p))\vert) (u_a(\Phi(s, p)) - 1)
    \\
    & \quad 
    \leq \cH_0 C_1 e^{- s \varDelta_0} \leq \cH_0 C_1 e^{- s \delta_0}. 
\end{split}
\label{BB3}
\ee
Thus, with the constant $C_4 = C_0 + \cH_0 C_1$ the Lemma follows.
\end{proof}

To conclude this subsection, notice that by Lemmas \ref{LEMMA_asymptotics_1} 
and \ref{LEMMA_asymptotics_2} we may call the functions $\lambda_a^+$ 
\eqref{lambda_+} the \emph{asymptotic positions}, and the name 
\emph{asymptotic rapidities} is also justified for $\theta_a^+$ 
\eqref{theta_+}.

\subsection{Derivation of the equations of variation}
\label{SUBSECTION_Derivation_of_the_variational_equations}
From Lemma \ref{LEMMA_asymptotics_1} the importance of the functions defined 
in \eqref{cE_a} and \eqref{cF_a} is evident. Not surprisingly, the whole 
scattering theory of the van Diejen systems \eqref{H} can be built upon 
the careful study of these objects. The partial differential equations we set 
up in this subsection form the basis of their finer analysis.

\begin{LEMMA}
\label{LEMMA_cE_and_cF_time_derivatives}
For each $a \in \bN_n$ we have $\dot{\cE}_a = \cF_a$ and 
$\dot{\cF}_a = \varphi_a \circ \Phi$, where
\be
    \varphi_a 
    = \sinh(\theta_a) u_a \sum_{c = 1}^n \sinh(\theta_c) u_c \cL^a_c
        - \cosh(\theta_a) u_a \sum_{c = 1}^n \cosh(\theta_c) u_c \cL^c_a 
    \in C^\infty(P).
\label{varphi_a}
\ee
\end{LEMMA}

\begin{proof}
Let $a \in \bN_n$ be arbitrary, then from \eqref{cE_a} and \eqref{lambda_dot} 
it is evident that
\be
    \PD{\cE_a}{t}
    = \PD{ 
        \left( 
            \lambda_a \circ \Phi - t \sinh(\theta_a^+) - \lambda_a^+ 
        \right)}
        {t}
    = \sinh(\theta_a(\Phi)) u_a(\Phi) - \sinh(\theta_a^+)
    = \cF_a,
\label{cE_dot}
\ee
whereas \eqref{Phi_dynamics} and the definition \eqref{cF_a} lead to 
\be
\begin{split}
    \PD{\cF_a}{t}
    & = \PD{ 
        \left( 
            \sinh(\theta_a(\Phi)) u_a(\Phi) - \sinh(\theta_a^+)    
        \right)}
        {t}
    = \cosh(\theta_a(\Phi)) \PD{\theta_a(\Phi)}{t} u_a(\Phi)
        + \sinh(\theta_a(\Phi)) \PD{u_a(\Phi)}{t}
    \\
    & = \left( \cosh(\theta_a) \PB{\theta_a}{H} u_a 
        + \sinh(\theta_a) \PB{u_a}{H} \right) \circ \Phi.
\end{split}
\label{cF_dot}
\ee
Recalling the notation \eqref{cL}, notice that for the constituent Poisson 
brackets we can write that
\begin{align}
    & \PB{\theta_a}{H} 
    = -\PD{H}{\lambda_a} 
    = - \sum_{c = 1}^n \cosh(\theta_c) u_c \cL^c_a,
    \label{T1_1}
    \\ 
    & \PB{u_a}{H} 
    = \sum_{c = 1}^n \PD{u_a}{\lambda_c} \PD{H}{\theta_c}
    = u_a \sum_{c = 1}^n \cL^a_c \sinh(\theta_c) u_c,
    \label{T1_2}
\end{align}
and so the Lemma follows immediately.
\end{proof}

There is no doubt that any change in the initial condition of the maximally 
defined integral curve $\gamma_p$ \eqref{gamma} has an inevitable effect 
on the asymptotic variables \eqref{lambda_+} and \eqref{theta_+}, too.
To measure the effect of small changes, we introduce the smooth functions
\be
    \cV_a^{(j)} = \PD{\cE_a}{\tx_j} \in C^\infty(\bR \times P)
    \midand
    \cW_a^{(j)} = \PD{\cF_a}{\tx_j} \in C^\infty(\bR \times P)
    \qquad
    (a \in \bN_n, \, j \in \bN_N),
\label{cV_and_cW}
\ee
simply by taking the first order partial derivatives of \eqref{cE_a} 
and \eqref{cF_a} with respect to the coordinates \eqref{tx}. In order 
to understand their asymptotic properties, first we focus on their time 
evolution. Remembering the notations introduced in \eqref{cL}, the following
result is immediate.

\begin{PROPOSITION}
\label{PROPOSITION_cV_and_cW_time_derivatives}
For all $a \in \bN_n$ and for all $j \in \bN_N$ we have
\be
    \PD{\cV_a^{(j)}}{t} = \cW_a^{(j)}
    \midand
    \PD{\cW_a^{(j)}}{t} 
    = \sum_{b = 1}^n 
        \left( \PD{\varphi_a}{\lambda_b} \circ \Phi \right) 
        \PD{(\lambda_b \circ \Phi)}{\tx_j}
        + \sum_{b = 1}^n \left( \PD{\varphi_a}{\theta_b} \circ \Phi \right) 
        \PD{(\theta_b \circ \Phi)}{\tx_j}
\label{cV_and_cW_derivative}
\ee
with the partial derivatives
\begin{align}
    \PD{\varphi_a}{\lambda_b} 
    & = \sinh(\theta_a) u_a 
        \sum_{c = 1}^n 
            \sinh(\theta_c) u_c 
            (\cL^a_b \cL^a_c + \cL^c_b \cL^a_c + \cL^a_{c, b}) 
    \nonumber 
    \\
    & \quad 
    -\cosh(\theta_a) u_a 
        \sum_{c = 1}^n 
            \cosh(\theta_c) u_c 
            (\cL^a_b \cL^c_a + \cL^c_b \cL^c_a + \cL^c_{a, b}),
    \label{PD_varphi_a_lambda_b}
    \\
    \PD{\varphi_a}{\theta_b}
    & = \sinh(\theta_a) u_a \cosh(\theta_b) u_b \cL^a_b 
        - \cosh(\theta_a) u_a \sinh(\theta_b) u_b \cL^b_a 
    \nonumber 
    \\
    & \quad 
    + \delta_{a, b} 
        \left( 
            \cosh(\theta_a) u_a \sum_{c = 1}^n \sinh(\theta_c) u_c \cL^a_c 
            - \sinh(\theta_a) u_a \sum_{c = 1}^n \cosh(\theta_c) u_c \cL^c_a
        \right).
    \label{PD_varphi_a_theta_b}
\end{align}
\end{PROPOSITION}

\begin{proof}
Due to the smoothness of the maps $\cE_a$ \eqref{cE_a} and $\cF_a$ 
\eqref{cF_a}, from Lemma \ref{LEMMA_cE_and_cF_time_derivatives} and the 
definitions displayed in \eqref{cV_and_cW} it is evident that
\be
    \PD{\cV_a^{(j)}}{t}
    = \PD{}{t} \PD{\cE_a}{\tx_j} 
    = \PD{}{\tx_j} \PD{\cE_a}{t} 
    = \PD{\cF_a}{\tx_j} 
    = \cW^{(j)}_a,
\label{O1}
\ee
whereas
\be
    \PD{\cW_a^{(j)}}{t}
    = \PD{}{t} \PD{\cF_a}{\tx_j} 
    = \PD{}{\tx_j} \PD{\cF_a}{t}
    = \PD{(\varphi_a \circ \Phi)}{\tx_j}
    = \sum_{k = 1}^N 
        \left( \PD{\varphi_a}{x_k} \circ \Phi \right) 
        \PD{(x_k \circ \Phi)}{\tx_j}.
\label{O2}
\ee
Now, simply by working out the partial derivatives of $\varphi_a$ 
\eqref{varphi_a}, the Proposition follows.
\end{proof}

To proceed further, we wish to sharpen the above Proposition by expressing 
the derivative $\dot{\cW}_a^{(j)}$ in terms of the functions \eqref{cV_and_cW}. 
Elementary algebraic manipulations lead to the following result.

\begin{LEMMA}
\label{LEMMA_cV_and_cW_time_derivatives}
For all $a \in \bN_n$ and for all $j \in \bN_N$ we have
\be
    \PD{\cW_a^{(j)}}{t} 
    = \sum_{c = 1}^n \cA_{a, c} \cV_c^{(j)} 
        + \sum_{c = 1}^n \cB_{a, c} \cW_c^{(j)} + \cZ_a^{(j)},
\label{cV_and_cW_derivative_OK}
\ee
where the coefficients are smooth functions on $\bR \times P$ of the form
\begin{align}
    & \cA_{a, c} 
    = \left( 
        \PD{\varphi_a}{\lambda_c} 
        - \sum_{d = 1}^n \tanh(\theta_d) \cL^d_c \PD{\varphi_a}{\theta_d}
    \right) \circ \Phi,
    \quad
    \cB_{a, c} 
    = \left( 
        \frac{1}{\cosh(\theta_c) u_c} \PD{\varphi_a}{\theta_c} 
    \right) \circ \Phi,
    \label{cA_and_cB}
    \\
    & \cZ_a^{(j)}
    = \sum_{c = 1}^n \cA_{a, c} 
        \left( 
            t \PD{\sinh(\theta_c^+)}{x_j} + \PD{\lambda_c^+}{x_j} 
        \right)
    + \sum_{c = 1}^n \cB_{a, c} \PD{\sinh(\theta_c^+)}{x_j}.
    \label{cZ}
\end{align}
\end{LEMMA}

\begin{proof}
Let $b \in \bN_n$ and $j \in \bN_N$ be arbitrary indices. Looking back to the 
definition \eqref{cE_a}, it is clear that
\be
    \lambda_b \circ \Phi = \cE_b + t \sinh(\theta_b^+) + \lambda_b^+,
\label{L1}
\ee
whence by \eqref{cV_and_cW} we obtain
\be
    \PD{(\lambda_b \circ \Phi)}{\tx_j}
    = \cV^{(j)}_b + t \PD{\sinh(\theta_b^+)}{x_j} + \PD{\lambda_b^+}{x_j}.
\label{PDs_lambda_b_circ_Phi}
\ee
Next, from the definitions \eqref{cF_a} and \eqref{cV_and_cW} it follows that
\be
\begin{split}
    \cW^{(j)}_b 
    & = \PD{(\sinh(\theta_b(\Phi)) u_b(\Phi) - \sinh(\theta_b^+))}{\tx_j}
    \\
    & = \cosh(\theta_b(\Phi)) \PD{(\theta_b \circ \Phi)}{\tx_j} u_b(\Phi)
    + \sinh(\theta_b(\Phi)) \PD{(u_b \circ \Phi)}{\tx_j}
    - \PD{\sinh(\theta_b^+)}{x_j},
\end{split}
\label{L2}
\ee
and so we can write that
\be
    \PD{(\theta_b \circ \Phi)}{\tx_j} 
    = \frac{1}{\cosh(\theta_b(\Phi)) u_b(\Phi)} 
        \left( 
            \cW^{(j)}_b + \PD{\sinh(\theta_b^+)}{x_j} 
            - \sinh(\theta_b(\Phi)) \PD{(u_b \circ \Phi)}{\tx_j}
        \right).
\label{L3}
\ee
However, keeping in mind \eqref{PDs_lambda_b_circ_Phi}, it is immediate that
\be
\begin{split}
    \PD{(u_b \circ \Phi)}{\tx_j}
    = \sum_{d = 1}^n 
        \left( \PD{u_b}{\lambda_d} \circ \Phi \right) 
        \PD{(\lambda_d \circ \Phi)}{\tx_j}
    = u_b(\Phi) \sum_{d = 1}^n 
                    \cL^b_d(\Phi) 
                    \left( 
                        \cV^{(j)}_d + t \PD{\sinh(\theta_d^+)}{x_j} 
                        + \PD{\lambda_d^+}{x_j} 
                    \right).
\end{split}
\label{L4}
\ee
Thus, by inserting \eqref{L4} into \eqref{L3}, we end up with the formula
\be
\begin{split}
    \PD{(\theta_b \circ \Phi)}{\tx_j} 
    & = \frac{1}{\cosh(\theta_b(\Phi)) u_b(\Phi)} \cW^{(j)}_b 
        - \tanh(\theta_b(\Phi)) \sum_{d = 1}^n \cL^b_d(\Phi) \cV^{(j)}_d 
    \\
    & \quad + \frac{1}{\cosh(\theta_b(\Phi)) u_b(\Phi)} 
                \PD{\sinh(\theta_b^+)}{x_j} 
    \\        
    & \quad - \tanh(\theta_b(\Phi)) 
                \sum_{d = 1}^n 
                    \cL^b_d(\Phi) 
                    \left( 
                        t \PD{\sinh(\theta_d^+)}{x_j} + \PD{\lambda_d^+}{x_j} 
                    \right).
\end{split}
\label{PDs_theta_b_circ_Phi}
\ee
Now, by plugging the formulae \eqref{PDs_lambda_b_circ_Phi} and 
\eqref{PDs_theta_b_circ_Phi} into \eqref{cV_and_cW_derivative}, the 
Lemma follows. 
\end{proof}

Now, let us observe that with the aid of the matrix valued functions
\be
    \bscA = [\cA_{a, c}]_{1 \leq a, c \leq n},
    \quad
    \bscB = [\cB_{a, c}]_{1 \leq a, c \leq n},
\label{bscA_and_bscB}
\ee
and the column vector valued functions
\be
    \bscV^{(j)} = [\cV^{(j)}_a]_{1 \leq a \leq n},
    \quad
    \bscW^{(j)} = [\cW^{(j)}_a]_{1 \leq a \leq n},
    \quad
    \bscZ^{(j)} = [\cZ^{(j)}_a]_{1 \leq a \leq n}
    \qquad
    (j \in \bN_N),
\label{bscV_and_bscW_and_bscZ}
\ee
the differential equations displayed in Proposition 
\ref{PROPOSITION_cV_and_cW_time_derivatives} and Lemma 
\ref{LEMMA_cV_and_cW_time_derivatives} can be cast into the more concise 
matrix form
\be
    \dot{\bscV}^{(j)} = \bscW^{(j)},
    \quad
    \dot{\bscW}^{(j)}= \bscA \bscV^{(j)} + \bscB \bscW^{(j)} + \bscZ^{(j)}.
\label{variational_equations}
\ee
Following the standard terminology, this inhomogeneous linear system may 
be called the system of \emph{equations of variation} associated with the 
non-linear differential equations appearing in Lemma 
\ref{LEMMA_cE_and_cF_time_derivatives}. (For background information on
the theory of the equations of variation see. e.g. Chapter V in the classic 
textbook \cite{Hartman}.)

\subsection{Analyzing the equations of variation}
\label{SUBSECTION_Analyzing_the_variational_equations}
By merging our former asymptotic results with the equations of variation
\eqref{variational_equations}, in this subsection we wish to extend our 
asymptotic analysis to the functions \eqref{cV_and_cW}. As a preliminary
step, we begin with two remarks on the notations. 

First, recall that for any $m \times m'$ matrix 
\be
    A = [A_{k, l}]_{\substack{1 \leq k \leq m, \\ 1 \leq l \leq m'}} 
        \in \bC^{m \times m'}
\label{matrix_A}
\ee 
its \emph{Hilbert--Schmidt norm} is given by
\be
    \Vert A \Vert 
    = \sqrt{\sum_{k = 1}^m \sum_{l = 1}^{m'} \vert A_{k, l} \vert^2}.
\label{HS_norm}
\ee
Besides the above simple formula, we shall often exploit the fact
that it is submultiplicative. In other words, if $A \in \bC^{m \times m'}$
and $B \in \bC^{m' \times m''}$, then 
$\Vert A B \Vert \leq \Vert A \Vert \Vert B \Vert$.

Second, \emph{throughout this subsection we fix an arbitrary point 
$p_0 \in P$, and the earlier notations for the associated constants 
introduced in subsections \ref{SUBSECTION_Applying_Ruijsenaars_thm} and 
\ref{SUBSECTION_Asymptotics} are also kept in effect.} In particular,
remembering the compact subset $K_0 \subset P$ given in Lemma 
\ref{LEMMA_asymptotics_1}, the constant $T_1$ provided by Proposition
\ref{PROPOSITION_particle_repulsion}, and the shorthand notation 
\eqref{delta_0}, for the Hilbert--Schmidt norms of the coefficient matrices 
appearing in the system of the equations of variation 
\eqref{variational_equations} one can easily establish the following result.

\begin{PROPOSITION}
\label{PROPOSITION_bscA_bscB_bscZ_norms}
There is a constant $c_0 \geq 0$ such that $\forall s \in [T_1, \infty)$ and
$\forall p \in K_0$ we have
\be
    \Vert \bscA(s, p) \Vert \leq c_0 e^{-s \delta_0},
    \quad
    \Vert \bscB(s, p) \Vert \leq c_0 e^{-s \delta_0},
\label{bscA_bscB_norms}
\ee
and also $\forall j \in \bN_N$ we can write that
\be
    \Vert \bscZ^{(j)}(s, p) \Vert \leq c_0 (s + 1) e^{-s \delta_0}.
\label{bscZ_norm} 
\ee
\end{PROPOSITION}

\begin{proof}
Let $a, b, c \in \bN_n$, $s \in [T_1, \infty)$ and $p \in K_0$. By combining
the estimate \eqref{theta_estim_0} with Propositions 
\ref{PROPOSITION_u_asymptotics}, \ref{PROPOSITION_cL_a_b_estim} and 
\ref{PROPOSITION_cL_a_b_c_estim}, from the explicit expressions 
\eqref{PD_varphi_a_lambda_b} and \eqref{PD_varphi_a_theta_b} it is evident 
that
\be
    \left\vert \PD{\varphi_a}{\lambda_b} (\Phi(s, p)) \right\vert 
        \leq c_1 e^{-s \delta_0}
    \midand
    \left\vert \PD{\varphi_a}{\theta_b} (\Phi(s, p)) \right\vert 
        \leq c_1 e^{-s \delta_0}
\label{S2}
\ee
with some constant $c_1 \geq 0$. Thus, giving a glance at \eqref{cA_and_cB},
it readily follows that there is a constant $c_2 \geq 0$ such that
\be
    \vert \cA_{a, c}(s, p) \vert \leq c_2 e^{-s \delta_0}
    \midand
    \vert \cB_{a, c}(s, p) \vert \leq c_2 e^{-s \delta_0}.
\label{S3}
\ee 
Finally, from \eqref{cZ} it is also clear that $\forall j \in \bN_N$ we can 
write 
\be
\begin{split}
    \vert \cZ_a^{(j)}(s, p) \vert
    & \leq \sum_{c = 1}^n 
            \vert \cA_{a, c}(s, p) \vert 
            \left(
                s \vert \cosh(\theta_c^+(p)) \vert 
                \left\vert \PD{\theta^+_c}{x_j}(p) \right\vert
                + \left\vert \PD{\lambda^+_c}{x_j}(p) \right\vert
            \right)
    \\
    & \quad + \sum_{c = 1}^n 
                \vert \cB_{a, c}(s, p) \vert 
                \vert \cosh(\theta_c^+(p)) \vert 
                \left\vert \PD{\theta^+_c}{x_j}(p) \right\vert.
\end{split}
\label{S4}
\ee
Since the smooth functions $\lambda_c^+$, $\theta_c^+$, and their partial
derivatives are also bounded on the compact subset $K_0$, the estimates 
displayed in \eqref{S3} entail that
\be
    \vert \cZ_a^{(j)}(s, p) \vert \leq (c_3 s + c_4) e^{-s \delta_0}
\label{S5_}
\ee
with some constants $c_3, c_4 \geq 0$. So, recalling \eqref{HS_norm}, with 
$c_0 = \max\{ n c_2, \sqrt{n} c_3, \sqrt{n} c_4 \}$ the Proposition follows.
\end{proof}

As can be seen below, Proposition \ref{PROPOSITION_bscA_bscB_bscZ_norms} 
itself allows us to give a very rough estimate on the growth of the functions 
\eqref{cV_and_cW}.

\begin{LEMMA}
\label{LEMMA_bscV_and_bscW_estim}
There are constants $v_0, w_0 \geq 0$ such that $\forall s \in [T_1, \infty)$,
$\forall p \in K_0$, and $\forall j \in \bN_N$ we can write that
\be
    \Vert \bscV^{(j)}(s, p) \Vert \leq v_0 + (s - T_1) w_0
    \midand
    \Vert \bscW^{(j)}(s, p) \Vert \leq w_0.
\label{bscV_and_bscW_estim}
\ee
\end{LEMMA}

\begin{proof}
Notice that the dependence of the norms $\Vert \bscV^{(j)}(T_1, p) \Vert$ 
and $\Vert \bscW^{(j)}(T_1, p) \Vert$ on $p$ is continuous, thus by the
compactness of $K_0$ there are some constants $v_0, c_1 \geq 0$ such that
\be
    \Vert \bscV^{(j)}(T_1, p) \Vert \leq v_0 
    \midand
    \Vert \bscW^{(j)}(T_1, p) \Vert \leq c_1
    \qquad
    (j \in \bN_N, \, p \in K_0).
\label{bscV_bscW_norms_at_T_1}
\ee
To proceed further, let $s \in [T_1, \infty)$, $p \in K_0$, and $j \in \bN_N$ 
be arbitrary elements and keep them fixed. Upon integrating the equations of
variation \eqref{variational_equations} with respect to time $t$, we obtain
\be
    \bscV^{(j)}(s, p) - \bscV^{(j)}(T_1, p) 
    = \int_{T_1}^s \dot{\bscV}^{(j)} (\tau, p) \, \dd \tau
    = \int_{T_1}^s \bscW^{(j)} (\tau, p) \, \dd \tau,
\label{bscV_integral_form}
\ee
therefore we can write 
\be
    \Vert \bscV^{(j)}(s, p) \Vert 
    \leq \Vert \bscV^{(j)}(T_1, p) \Vert
    + \int_{T_1}^s \Vert \bscW^{(j)}(\tau, p) \Vert \, \dd \tau
    \leq v_0 + \int_{T_1}^s \Vert \bscW^{(j)}(\tau, p) \Vert \, \dd \tau.
\label{bscV_norm_estim}
\ee
Along the same lines, utilizing the integral equation
\be
    \bscW^{(j)}(s, p) - \bscW^{(j)}(T_1, p) 
    = \int_{T_1}^s 
        \left( 
            \bscA(\tau, p) \bscV^{(j)}(\tau, p) 
            + \bscB(\tau, p) \bscW^{(j)}(\tau, p) 
            + \bscZ^{(j)}(\tau, p)
        \right) \, \dd \tau,
\label{bscW_integral_form}
\ee
we infer that
\be
\begin{split}
    \Vert \bscW^{(j)}(s, p) \Vert
    & \leq \Vert \bscW^{(j)}(T_1, p) \Vert
    + \int_{T_1}^s 
        \Vert \bscA(\tau, p) \Vert 
        \Vert \bscV^{(j)}(\tau, p) \Vert 
        \, \dd \tau
    \\
    & \quad + \int_{T_1}^s 
                \Vert \bscB(\tau, p) \Vert 
                \Vert \bscW^{(j)}(\tau, p) \Vert 
                \, \dd \tau
    + \int_{T_1}^s \Vert \bscZ^{(j)}(\tau, p) \Vert \, \dd \tau.
\end{split}
\label{bscW_norm_estim}
\ee

Now, let us examine each term appearing on the right hand side of the above 
inequality. Due to the observation we made in \eqref{bscV_bscW_norms_at_T_1}, 
the first term is under control. Turning to the second term, let us note that
by exploiting \eqref{bscV_norm_estim} we can write 
\be
\begin{split}
    \int_{T_1}^s 
        \Vert \bscA(\tau, p) \Vert 
        \Vert \bscV^{(j)}(\tau, p) \Vert 
        \, \dd \tau
    & \leq v_0 \int_{T_1}^s \Vert \bscA(\tau, p) \Vert \, \dd \tau
    \\
    & \quad + \int_{T_1}^s 
        \Vert \bscA(\tau, p) \Vert 
        \left( 
            \int_{T_1}^\tau \Vert \bscW^{(j)}(u, p) \Vert \, \dd u
        \right) 
        \, \dd \tau.
\end{split}
\label{E2_}
\ee
However, recalling \eqref{bscA_bscB_norms}, it is obvious that
\be
    \int_{T_1}^s \Vert \bscA(\tau, p) \Vert \, \dd \tau
    \leq c_0 \int_{T_1}^s e^{-\tau \delta_0} \, \dd \tau 
    \leq c_0 \int_{T_1}^\infty e^{-\tau \delta_0} \, \dd \tau 
    = \frac{c_0}{\delta_0} e^{- T_1 \delta_0}
    \leq \frac{c_0}{\delta_0}.
\label{E3}
\ee
As concerns the double integral appearing on the right hand side of \eqref{E2_}, 
the Fubini--Tonelli theorem is clearly applicable, and we find that
\be
    \int_{T_1}^s 
        \Vert \bscA(\tau, p) \Vert 
        \left( 
            \int_{T_1}^\tau \Vert \bscW^{(j)}(u, p) \Vert \, \dd u
        \right) 
        \, \dd \tau
    = \int_{T_1}^s
        \left(
            \int_u^s \Vert \bscA(\tau, p) \Vert \, \dd \tau
        \right)
        \Vert \bscW^{(j)}(u, p) \Vert \, \dd u.
\label{double_integral}
\ee
Remembering \eqref{E3}, at this point let us notice that on the right hand 
side of the above equation we can apply the inequality
\be
    \int_u^s \Vert \bscA(\tau, p) \Vert \, \dd \tau 
    \leq \frac{c_0}{\delta_0} e^{-u \delta_0}.
\label{E4}
\ee
Thus, keeping in mind \eqref{E3}, \eqref{double_integral} and \eqref{E4}, 
from \eqref{E2_} we infer at once that
\be
    \int_{T_1}^s 
        \Vert \bscA(\tau, p) \Vert 
        \Vert \bscV^{(j)}(\tau, p) \Vert 
        \, \dd \tau
    \leq 
    v_0 \frac{c_0}{\delta_0} 
    + \frac{c_0}{\delta_0} 
        \int_{T_1}^s 
            e^{-u \delta_0} \Vert \bscW^{(j)}(u, p) \Vert 
            \, \dd u.
\label{E5}
\ee

Finally, by inspecting the last two terms on the right hand side of
\eqref{bscW_norm_estim}, the application of the estimates displayed in
\eqref{bscA_bscB_norms} leads to the inequality
\be
    \int_{T_1}^s 
        \Vert \bscB(\tau, p) \Vert 
        \Vert \bscW^{(j)}(\tau, p) \Vert 
        \, \dd \tau
    \leq c_0 \int_{T_1}^s 
                e^{-\tau \delta_0} 
                \Vert \bscW^{(j)}(\tau, p) \Vert 
                \, \dd \tau,
\label{E6}
\ee
whereas \eqref{bscZ_norm} yields
\be
    \int_{T_1}^s \Vert \bscZ^{(j)}(\tau, p) \Vert \, \dd \tau 
    \leq
    c_0 \int_{T_1}^s (\tau + 1) e^{- \tau \delta_0} \, \dd \tau
    \leq
    c_0 \int_{T_1}^\infty (\tau + 1) e^{- \tau \delta_0} \, \dd \tau
    \leq c_2
\label{E7}
\ee
with some constant $c_2 \geq 0$.

To sum up, by applying the estimates displayed in 
\eqref{bscV_bscW_norms_at_T_1}, \eqref{E5}, \eqref{E6} and \eqref{E7}, from 
the inequality \eqref{bscW_norm_estim} we can derive that
\be
    \Vert \bscW^{(j)}(s, p) \Vert
    \leq
    c_3
    + c_4 \int_{T_1}^s 
            e^{- \tau \delta_0} 
            \Vert \bscW^{(j)}(\tau, p) \Vert 
            \, \dd \tau,
\label{integral_inequality}
\ee
where $c_3 = c_1 + v_0 c_0 \delta_0^{-1} +c_2$ and 
$c_4 = c_0 (1 + \delta_0^{-1})$. As a consequence, by invoking Gr\"onwall's 
lemma (see e.g. Theorem 1.1 in Chapter III of \cite{Hartman}), we obtain that 
\be
    \Vert \bscW^{(j)}(s, p) \Vert 
    \leq c_3 \exp 
                \left( 
                    c_4 \int_{T_1}^s e^{-\tau \delta_0} \, \dd \tau 
                \right).
\label{Gronwall}
\ee
So, with the non-negative constant
\be
    w_0 
    = c_3 \exp 
            \left( 
                c_4 \int_{T_1}^\infty e^{-\tau \delta_0} \, \dd \tau 
            \right)
    = c_3 \exp \left(\frac{c_4}{\delta_0} e^{-T_1 \delta_0} \right)
\label{w_0}
\ee
we end up with the desired inequality 
$\Vert \bscW^{(j)}(s, p) \Vert \leq w_0$. Moreover, utilizing 
\eqref{bscV_norm_estim} and \eqref{bscV_bscW_norms_at_T_1}, it also 
follows that
\be
    \Vert \bscV^{(j)}(s, p) \Vert \leq v_0 + w_0 (s - T_1), 
\label{bscV_norm_OK}
\ee
whence the proof is complete.
\end{proof}

Of course, later on we shall need much sharper estimates on the functions
defined in \eqref{cV_and_cW} than the rudimentary inequalities given in the 
previous Lemma. To make progress, let $a \in \bN_n$, $p \in K_0$, and suppose 
that $s' \geq s \geq T_1$. Recalling the time derivatives appearing in Lemma 
\ref{LEMMA_cE_and_cF_time_derivatives}, it is clear that
\be
    \cE_a(s', p) - \cE_a(s, p) = \int_s^{s'} \cF_a(\tau, p) \, \dd \tau
    \midand
    \cF_a(s', p) - \cF_a(s, p) 
    = \int_s^{s'} \varphi_a(\Phi(\tau, p)) \, \dd \tau.
\label{X1}
\ee
Notice that by Lemma \ref{LEMMA_asymptotics_1} we have
\be
    \cE_a(s', p) \to 0 
    \midand
    \cF_a(s', p) \to 0
    \qquad
    (s' \to \infty). 
\label{X2}
\ee
Furthermore, from the decay conditions given in Lemma \ref{LEMMA_asymptotics_1} 
we also see that the function
\be
    [s, \infty) \ni \tau \mapsto \cF_a(\tau, p) \in \bC
\label{X3}
\ee
is integrable in Lebesgue's sense over the interval $[s, \infty)$. Combining 
the estimate given in \eqref{theta_estim_0} with Propositions 
\ref{PROPOSITION_u_asymptotics} and \ref{PROPOSITION_cL_a_b_estim}, from 
\eqref{varphi_a} it is also clear that $\varphi_a (\Phi(\tau, p))$ decays 
exponentially fast for $\tau \to \infty$, whence the function
\be
    [s, \infty) \ni \tau \mapsto \varphi_a(\Phi(\tau, p)) \in \bC
\label{X4}
\ee
also belongs to $L^1[s, \infty)$. Therefore, a trivial application of 
Lebesgue's dominated convergence theorem yields immediately that for 
$s' \to \infty$ we can write
\be
    \int_s^{s'} \cF_a(\tau, p) \, \dd \tau 
    \to \int_s^\infty \cF_a(\tau, p) \, \dd \tau
    \midand
    \int_s^{s'} \varphi_a(\Phi(\tau, p)) \, \dd \tau 
    \to \int_s^\infty \varphi_a(\Phi(\tau, p)) \, \dd \tau.
\label{X5}
\ee
Thus, combining \eqref{X2} and \eqref{X5} with \eqref{X1}, we end up with 
the integral representations
\be
    \cE_a(s, p) = - \int_s^\infty \cF_a(\tau, p) \, \dd \tau
    \midand
    \cF_a(s, p) = - \int_s^\infty \varphi_a(\Phi(\tau, p)) \, \dd \tau.
\label{cE_and_cF_integral_representation}
\ee

Making use of the above observations, now we turn back to the study of the 
partial derivatives of the functions $\cE_a$ and $\cF_a$ $(a \in \bN_n)$. 
More precisely, we wish to set up integral representations for the first 
order partial derivatives defined in \eqref{cV_and_cW}. In this respect the 
only non-trivial question is whether in the relationships displayed in 
\eqref{cE_and_cF_integral_representation} the differentiation with respect 
to $\tx_j$ $(j \in \bN_N)$ can be performed under the integral sign. 

Starting with the integral representation of $\cF_a$, notice that the 
integrand $\varphi_a \circ \Phi$ is smooth, and from Lemmas 
\ref{LEMMA_cE_and_cF_time_derivatives} and 
\ref{LEMMA_cV_and_cW_time_derivatives} it is evident that
\be
    \PD{(\varphi_a \circ \Phi)}{\tx_j} 
    = \PD{}{\tx_j} \PD{\cF_a}{t}
    = \PD{}{t} \PD{\cF_a}{\tx_j}
    = \PD{\cW^{(j)}_a}{t}
    = \sum_{c = 1}^n \cA_{a, c} \cV^{(j)}_c 
        + \sum_{c = 1}^n \cB_{a, c} \cW^{(j)}_c + \cZ^{(j)}_a.
\label{Y1}
\ee 
Since still $s \in [T_1, \infty)$ and $p \in K_0$, from Proposition 
\ref{PROPOSITION_bscA_bscB_bscZ_norms} and Lemma 
\ref{LEMMA_bscV_and_bscW_estim} we see that for any $\tau \geq s$ the
above partial derivative can be estimated from above as
\be
\begin{split}
    \left\vert \PD{(\varphi_a \circ \Phi)}{\tx_j}(\tau, p) \right\vert
    & \leq 
    \sum_{c = 1}^n 
        \vert \cA_{a, c}(\tau, p) \vert \vert \cV^{(j)}_c(\tau, p) \vert 
    + \sum_{c = 1}^n 
        \vert \cB_{a, c}(\tau, p) \vert \vert \cW^{(j)}_c(\tau, p) \vert 
    + \vert \cZ^{(j)}_a(\tau, p) \vert
    \\
    & \leq
    c_1 (\tau + 1) e^{- \tau \delta_0}
\end{split}
\label{Y2}
\ee
with some constant $c_1 \geq 0$. Now, the point is that the majorizing 
function
\be
    [s, \infty) \ni \tau \mapsto c_1 (\tau + 1) e^{- \tau \delta_0} \in \bC
\label{Y3}
\ee
is integrable on $[s, \infty)$. Moreover, the majorization \eqref{Y2} is 
uniform in the sense that it is independent of $p \in K_0$. So, a trivial 
application of Lebesgue's dominated convergence theorem yields that if $p$ 
belongs to the interior of $K_0$, then the differentiation can be performed 
under the integral sign. (This nice fact can be found in any textbook on real
analysis, see e.g. Theorem 2.27 in \cite{Folland}.) That is, it is fully 
justified to write that
\be
    \cW_a^{(j)}(s, p) = \PD{\cF_a}{\tx_j}(s, p) 
    = - \int_s^\infty \PD{(\varphi_a \circ \Phi)}{\tx_j}(\tau, p) \, \dd \tau
    \qquad
    (s \in [T_1, \infty), \, p \in \text{int}(K_0)).
\label{cW_integral_representation}
\ee
Now, due to \eqref{Y2}, it follows that $\forall s \in [T_1, \infty)$ and
$\forall p \in \text{int}(K_0)$ we can write
\be
    \vert \cW^{(j)}_a(s, p) \vert 
    \leq 
    \int_s^\infty 
        \left\vert 
            \PD{(\varphi_a \circ \Phi)}{\tx_j}(\tau, p) 
        \right\vert \, \dd \tau
    \leq 
    c_1 \int_s^\infty (\tau + 1) e^{- \tau \delta_0} \, \dd \tau
    \leq 
    c_2 (s + 1) e^{-s \delta_0} 
\label{Y4}
\ee
with some constant $c_2 \geq 0$. Since the closure of $\text{int}(K_0)$ 
coincides with $K_0$, by continuity it follows that
\be
    \vert \cW^{(j)}_a(s, p) \vert \leq c_2 (s + 1) e^{-s \delta_0}
    \qquad
    (s \in [T_1, \infty), \, p \in K_0).
\label{cW_estim_OK}
\ee

Turning to the integral representation of $\cE_a$, from 
\eqref{cE_and_cF_integral_representation} we see that in this case the 
integrand in question is given by the smooth function $\cF_a$. By definition 
\eqref{cV_and_cW}, the derivative of $\cF_a$ with respect to $\tx_j$ is the 
smooth function $\cW_a^{(j)}$, for which we have just derived the estimate 
\eqref{cW_estim_OK}. We proceed by noting that the upper bound in 
\eqref{cW_estim_OK} is independent of the choice of $p \in K_0$. Moreover, 
this upper bound is an integrable function of $s$ on $[T_1, \infty)$, thus 
performing the differentiation under the integral sign is completely 
legitimate. More precisely, we are entitled to write that
\be
    \cV_a^{(j)}(s, p) = \PD{\cE_a}{\tx_j}(s, p) 
    = - \int_s^\infty \cW_a^{(j)}(\tau, p) \, \dd \tau
    \qquad
    (s \in [T_1, \infty), \, p \in \text{int}(K_0)).
\label{cV_integral_representation}
\ee
As a consequence, $\forall s \in [T_1, \infty)$ and 
$\forall p \in \text{int}(K_0)$ we have
\be
    \vert \cV^{(j)}_a(s, p) \vert 
    \leq 
    \int_s^\infty \vert \cW_a^{(j)}(\tau, p) \vert \, \dd \tau
    \leq 
    c_2 \int_s^\infty (\tau + 1) e^{- \tau \delta_0} \, \dd \tau
    \leq 
    c_3 (s + 1) e^{-s \delta_0} 
\label{Y5}
\ee
with some constant $c_3 \geq 0$. Again, by continuity, we conclude that
\be
    \vert \cV^{(j)}_a(s, p) \vert \leq c_3 (s + 1) e^{-s \delta_0}
    \qquad
    (s \in [T_1, \infty), \, p \in K_0).
\label{cV_estim_OK}
\ee
At this point we are in a position to formulate the most important technical
result of this section.

\begin{THEOREM}
\label{THEOREM_cV_and_cW_estim}
There is a constant $C \geq 0$ such that $\forall a \in \bN_n$, 
$\forall j \in \bN_N$, $\forall s \in [T_1, \infty)$ and 
$\forall p \in K_0$ we can write that
\be
    \vert \cV^{(j)}_a(s, p) \vert \leq C (s + 1) e^{-s \delta_0}
    \midand
    \vert \cW^{(j)}_a(s, p) \vert \leq C (s + 1) e^{-s \delta_0}.
\label{cV_and_cW_estim_OK}
\ee
\end{THEOREM}

Just as earlier in Lemma \ref{LEMMA_asymptotics_1}, in the above Theorem we 
face the problem that the relevant asymptotic information about the particle 
rapidities is somewhat hidden. Although the above estimate on $\cW_a^{(j)}$ 
proves to be crucial, what we rather need is the control over the partial
derivative
\be
    \cU_a^{(j)} = \PD{\cG_a}{\tx_j} \in C^\infty(\bR \times P)
    \qquad
    (a \in \bN_n, \, j \in \bN_N)
\label{cU}
\ee
with $\cG_a$ defined in \eqref{cG_a}. Since by construction we have
\be
    \PD{(\theta_a \circ \Phi)}{\tx_j} = \cU_a^{(j)} + \PD{\theta_a^+}{x_j},
\label{TTTT1}
\ee
notice that the application of the formula \eqref{PDs_theta_b_circ_Phi} 
entails that
\be
\begin{split}
    \cU_a^{(j)} 
    & = \frac{1}{\cosh(\theta_a(\Phi)) u_a(\Phi)} \cW^{(j)}_a
        - \frac{\cosh(\theta_a(\Phi)) u_a(\Phi) - \cosh(\theta_a^+)}
            {\cosh(\theta_a(\Phi)) u_a(\Phi)}
            \PD{\theta_a^+}{x_j} 
    \\
    & \quad 
        - \tanh(\theta_a(\Phi)) 
            \sum_{c = 1}^n 
                \cL^a_c(\Phi) 
                \left( 
                    \cV^{(j)}_c + t \cosh(\theta_c^+) \PD{\theta_c^+}{x_j} 
                        + \PD{\lambda_c^+}{x_j} 
                \right).
\end{split}
\label{TTTT2}
\ee
Thus, it is evident that
\be
\begin{split}
    \vert \cU_a^{(j)} \vert 
    & \leq \vert \cW^{(j)}_a \vert
        + \vert \cosh(\theta_a(\Phi)) u_a(\Phi) - \cosh(\theta_a^+) \vert
            \left\vert \PD{\theta_a^+}{x_j} \right\vert 
    \\
    & \quad + \sum_{c = 1}^n 
                \vert \cL^a_c(\Phi) \vert 
                \left( 
                    \vert \cV^{(j)}_c \vert 
                    + \vert t \vert \cosh(\theta_c^+) 
                        \left\vert \PD{\theta_c^+}{x_j} \right\vert 
                    + \left\vert \PD{\lambda_c^+}{x_j} \right\vert 
                \right).
\end{split}
\label{TTTT3}
\ee
Let $s \in [T_1, \infty)$ and $p \in K_0$ be arbitrary, and examine the right 
hand side of the above inequality at the point $(s, p)$. Due to Theorem 
\ref{THEOREM_cV_and_cW_estim} and Proposition \ref{PROPOSITION_cL_a_b_estim}, 
the first term and also the terms in the last sum are under direct control, 
so only the second term requires attention. However, by exploiting the 
relationship \eqref{cH_0} together with Proposition 
\ref{PROPOSITION_u_asymptotics} and Lemma \ref{LEMMA_asymptotics_2}, it is a 
routine exercise to verify that
\be
    \vert 
        \cosh(\theta_a(\Phi(s, p))) u_a(\Phi(s, p)) - \cosh(\theta_a^+(p)) 
    \vert
    \leq
    c'_1 e^{-s \delta_0} 
\label{TTTT4}
\ee
with some constant $c'_1 \geq 0$. Thus, the following result is immediate.

\begin{PROPOSITION}
\label{PROPOSITION_cU_estim}
There is a constant $C' \geq 0$ such that $\forall a \in \bN_n$, 
$\forall j \in \bN_N$, $\forall s \in [T_1, \infty)$ and 
$\forall p \in K_0$ we can write that
\be
    \vert \cU^{(j)}_a(s, p) \vert \leq C' (s + 1) e^{-s \delta_0}.
\label{cU_estim_OK}
\ee
\end{PROPOSITION}

As a trivial corollary of Lemma \ref{LEMMA_asymptotics_1}, Lemma
\ref{LEMMA_asymptotics_2}, Theorem \ref{THEOREM_cV_and_cW_estim} and 
Proposition \ref{PROPOSITION_cU_estim}, we conclude this subsection with 
the following observation.

\begin{PROPOSITION}
\label{PROPOSITION_important_limits}
As $s \to \infty$, $\forall a \in \bN_n$, $\forall j \in \bN_N$ and 
$\forall m \in \bN$ we have
\be
    s^m \cE_a(s, p_0) \to 0, 
    \quad 
    s^m \cG_a(s, p_0) \to 0, 
    \quad 
    s^m \cV^{(j)}_a(s, p_0) \to 0, 
    \quad 
    s^m \cU^{(j)}_a(s, p_0) \to 0.
\label{important_limits}
\ee
\end{PROPOSITION}

\subsection{Canonicity of the dual variables}
\label{SUBSECTION_dual_variables_PBs}
Our first goal in this subsection is to compute the Poisson brackets of
the asymptotic variables defined in \eqref{lambda_+} and \eqref{theta_+}. 
Recalling \eqref{cE_a}, \eqref{cG_a}, and the notation
for the sections introduced in \eqref{sections}, it is clear that 
$\forall a \in \bN_n$ and $\forall s \in \bR$ we can write that
\be
    \lambda_a \circ \Phi_s = s \sinh(\theta_a^+) + \lambda_a^+ + (\cE_a)_s
    \midand
    \theta_a \circ \Phi_s = \theta_a^+ + (\cG_a)_s.
\label{lambda_and_theta_circ_Phi_s}
\ee
Let us keep in mind that in the previous subsections we established tight 
control on the `error terms' $\cE_a$ and $\cG_a$. Therefore, by exploiting 
the fundamental Poisson brackets \eqref{lambda_and_theta_PBs} and the fact 
that $\forall s \in \bR$ the time-$s$ flow $\Phi_s \colon P \rightarrow P$ 
is a \emph{symplectomorphism}, the Poisson brackets of the asymptotic 
variables also become accessible.

\begin{LEMMA}
\label{LEMMA_asymptotic_variables_PBs}
For all $a, b \in \bN_n$ we have 
$\PB{\lambda_a^+}{\lambda_b^+} = 0 = \PB{\theta_a^+}{\theta_b^+}$ and
$\PB{\lambda_a^+}{\theta_b^+} = \delta_{a, b}$.
\end{LEMMA}

\begin{proof}
Throughout the proof let $a, b \in \bN_n$ and $p \in P$ be arbitrary, 
but fixed elements. Starting with the Poisson brackets of the asymptotic 
rapidities, from \eqref{lambda_and_theta_circ_Phi_s} it is evident that 
$\forall s \in \bR$ we have
\be
    0 = \PB{\theta_a}{\theta_b} \circ \Phi_s
    = \PB{\theta_a \circ \Phi_s}{\theta_b \circ \Phi_s}
    = \PB{\theta_a^+}{\theta_b^+} + \PB{\theta_a^+}{(\cG_b)_s}
        + \PB{(\cG_a)_s}{\theta_b^+} + \PB{(\cG_a)_s}{(\cG_b)_s}.
\label{P1}
\ee
Now, recalling the definitions \eqref{tx} and \eqref{cU}, at point $p$ 
the Poisson bracket formula \eqref{PB} allows us to write that 
\be
\begin{split}
    \PB{\theta_a^+}{(\cG_b)_s}(p) 
    & = \sum_{k, l = 1}^N \Omega_{k, l} 
        \PD{\theta_a^+}{x_k}(p) \PD{(\cG_b)_s}{x_l}(p) 
    = \sum_{k, l = 1}^N \Omega_{k, l} 
        \PD{\theta_a^+}{x_k}(p) \PD{\cG_b}{\tx_l}(s, p)
    \\
    & = \sum_{k, l = 1}^N \Omega_{k, l} 
        \PD{\theta_a^+}{x_k}(p) \cU_b^{(l)}(s, p).
\end{split}
\label{P2}
\ee
However, on account of Proposition \ref{PROPOSITION_important_limits} it
is clear that 
\be
    \PB{\theta_a^+}{(\cG_b)_s}(p) \to 0
    \qquad
    (s \to \infty).
\label{P3}
\ee 
It is obvious that these arguments also apply to the last two terms 
appearing on the right hand side of the equation \eqref{P1}. Indeed, 
one can easily verify that
\be
    \PB{(\cG_a)_s}{\theta_b^+}(p) \to 0
    \midand
    \PB{(\cG_a)_s}{(\cG_b)_s}(p) \to 0
    \qquad
    (s \to \infty),
\label{P4}
\ee
and so from \eqref{P1} we infer that 
\be
    \PB{\theta_a^+}{\theta_b^+}(p) = 0.
\label{PB_theta+_theta+}
\ee

Next, remembering \eqref{lambda_and_theta_circ_Phi_s}, let us notice that 
$\forall s \in \bR$ we can write
\be
\begin{split}
    \delta_{a, b} & = \PB{\lambda_a}{\theta_b} \circ \Phi_s
    = \PB{\lambda_a \circ \Phi_s}{\theta_b \circ \Phi_s}
    \\
    & = s \PB{\sinh(\theta_a^+)}{\theta_b^+}
        + \PB{\lambda_a^+}{\theta_b^+}
        + \PB{(\cE_a)_s}{\theta_b^+}
        \\
        & \quad + s \PB{\sinh(\theta_a^+)}{(\cG_b)_s} 
        + \PB{\lambda_a^+}{(\cG_b)_s}
        + \PB{(\cE_a)_s}{(\cG_b)_s}.
\end{split}
\label{PP1}
\ee
As before, we inspect the right hand side of the above equation on a 
term-by-term basis. Due to the relationship \eqref{PB_theta+_theta+},
at point $p$ we can write
\be
    \PB{\sinh(\theta_a^+)}{\theta_b^+}(p) 
    = \cosh(\theta_a^+(p)) \PB{\theta_a^+}{\theta_b^+}(p) 
    = 0. 
\label{PP2}
\ee    
Repeating the ideas surrounding the equations \eqref{P2} and \eqref{P3}, from 
Proposition \ref{PROPOSITION_important_limits} and the Poisson bracket formula 
\eqref{PB} one can also infer that in the $s \to \infty$ limit we have
\begin{align}
    & \PB{(\cE_a)_s}{\theta_b^+}(p) \to 0,
    \quad
    s \PB{\sinh(\theta_a^+)}{(\cG_b)_s}(p) \to 0,
    \label{PP3}
    \\
    & \PB{\lambda_a^+}{(\cG_b)_s}(p) \to 0,
    \quad
    \PB{(\cE_a)_s}{(\cG_b)_s}(p) \to 0.
    \label{PP3'}
\end{align}
Thus, \eqref{PP1} immediately leads to the relationship 
\be
    \PB{\lambda_a^+}{\theta_a^+}(p) = \delta_{a, b}.
\label{PB_lambda+_theta+}
\ee

Finally, we turn our attention to the remaining Poisson brackets involving 
only the asymptotic positions. From \eqref{lambda_and_theta_circ_Phi_s} it 
is clear that for each $s \in \bR$ we have
\be
\begin{split}
    0 & = \PB{\lambda_a}{\lambda_b} \circ \Phi_s
    = \PB{\lambda_a \circ \Phi_s}{\lambda_b \circ \Phi_s}
    \\
    & = \PB{\lambda_a^+}{\lambda_b^+} 
        + s^2 \PB{\sinh(\theta_a^+)}{\sinh(\theta_b^+)}
        + s \left( 
                \PB{\sinh(\theta_a^+)}{\lambda_b^+} 
                + \PB{\lambda_a^+}{\sinh(\theta_b^+)}
            \right)
    \\
    & \quad 
        + \PB{\sinh(\theta_a^+)}{s (\cE_b)_s}
        + \PB{s (\cE_a)_s}{\sinh(\theta_b^+)}
        + \PB{\lambda_a^+}{(\cE_b)_s}
        + \PB{(\cE_a)_s}{\lambda_b^+}
        + \PB{(\cE_a)_s}{(\cE_b)_s}.
\end{split}
\label{PPP1}
\ee
Due to the equation \eqref{PB_theta+_theta+} we can write 
\be
    \PB{\sinh(\theta_a^+)}{\sinh(\theta_b^+)}(p) 
    = \cosh(\theta_a^+(p)) \PB{\theta_a^+}{\theta_b^+}(p) \cosh(\theta_b^+(p))
    = 0, 
\label{PPP2}
\ee
whilst the Poisson bracket \eqref{PB_lambda+_theta+} entails
\be
    \PB{\sinh(\theta_a^+)}{\lambda_b^+}(p) 
    + \PB{\lambda_a^+}{\sinh(\theta_b^+)}(p) 
    = - \cosh(\theta_a^+(p)) \delta_{b, a} 
        + \delta_{a, b} \cosh(\theta_b^+(p))
    = 0.
\label{PPP3}
\ee
As before, the application of \eqref{PB} and Proposition 
\ref{PROPOSITION_important_limits} immediately yields that at point $p$
the last five terms appearing on the right hand side of \eqref{PPP1} 
vanish as $s \to \infty$, and so we end up with the desired equation 
\be
    \PB{\lambda_a^+}{\lambda_b^+}(p) = 0.
\label{PB_lambda+_lambda+}
\ee    
Now, on account of the relationships \eqref{PB_theta+_theta+}, 
\eqref{PB_lambda+_theta+} and \eqref{PB_lambda+_lambda+} the proof is 
complete.
\end{proof}

To proceed further, we still have to establish the relationships between the 
asymptotic and the dual variables. Recalling \eqref{theta_+}, we already 
know that $\theta_c^+ = 2 \htheta_c$ $(c \in \bN_n)$, but the relationship 
\eqref{lambda_+} is less explicit. To make it more transparent, we shall
need the smooth function $\Delta_c \colon Q \rightarrow \bR$ defined on
the configuration space $Q$ \eqref{Q} by the formula
\be
\begin{split}
    \Delta_c(\bsxi)
    & = - \half \sum_{d = 1}^{c - 1} 
        \ln \left( 1 + \frac{\sin(\mu)^2}{\sinh(\xi_c - \xi_d)^2} \right)
    + \half \sum_{d = c + 1}^n 
        \ln \left( 1 + \frac{\sin(\mu)^2}{\sinh(\xi_c - \xi_d)^2} \right) 
    \\
    & \quad 
    + \half \sum_{\substack{d = 1 \\ (d \neq c)}}^n 
        \ln \left( 1 + \frac{\sin(\mu)^2}{\sinh(\xi_c + \xi_d)^2} \right)
    + \half \ln \left( 1 + \frac{\sin(\nu)^2}{\sinh(2 \xi_c)^2} \right),
\end{split}
\label{Delta_c}
\ee
where $c \in \bN_n$ and $\bsxi = (\xi_1, \ldots, \xi_n) \in Q$.

\begin{LEMMA}
\label{LEMMA_lambda_+_OK}
For all $c \in \bN_n$ we have
$\lambda_c^+ 
    = \half \hlambda_c + \half \Delta_c(\htheta_1, \ldots, \htheta_n)$.
\end{LEMMA}

\begin{proof}
From \eqref{hL_entries} and the definition \eqref{bsTheta_+_and_L_tilde} it 
is evident that $\forall a, b \in \bN_n$ we can write that
\be
    \tilde{L}_{a, b} = \hL_{a, b} 
    = \hF_a 
        \frac{\sinh(\ri \hmu)}{\sinh(\ri \hmu + \htheta_a - \htheta_b)} 
        \bar{\hF}_b.
\label{Z1}
\ee
Thus, the $n \times n$ submatrix in the upper left hand corner of $\tilde{L}$
is essentially a Cauchy type matrix multiplied by diagonal matrices from
both sides. Therefore, by utilizing the determinant formula 
\eqref{Cauchy_det_formula}, it is straightforward to verify that for the 
$c$-th $(c \in \bN_n)$ leading principal minor \eqref{pi_j} of $\tilde{L}$ 
we have
\be
    \pi_c(\tilde{L}) 
    = \prod_{d = 1}^c \vert \hF_d \vert^2 
        \prod_{1 \leq a < b \leq c} 
            \frac{\sinh(\htheta_a - \htheta_b)^2}
                {\vert \sinh(\ri \hmu + \htheta_a - \htheta_b) \vert^2}
    = \prod_{d = 1}^c \vert \hF_d \vert^2 
        \prod_{1 \leq a < b \leq c} 
            \left( 
                1 + \frac{\sin(\mu)^2}{\sinh(\htheta_a - \htheta_b)^2} 
            \right)^{-1}.
\label{Z2}
\ee
Recalling the definitions of $m_j$ \eqref{m_j} and the components of $\hF_j$ 
\eqref{hF_components}, it is clear that 
\be 
    m_1(\tilde{L}) = \vert \hF_1 \vert^2 = e^{\hlambda_1} \hu_1.
\label{Z3_1}
\ee
Moreover, for any $c \in \{ 2, \ldots, n \}$ we can write that
\be
    m_c(\tilde{L}) = \frac{\pi_c(\tilde{L})}{\pi_{c - 1}(\tilde{L})}
    = e^{\hlambda_c} \hu_c 
        \prod_{d = 1}^{c - 1} 
            \left( 
                1 + \frac{\sin(\mu)^2}{\sinh(\htheta_c - \htheta_d)^2} 
           \right)^{-1}.
\label{Z3_2}
\ee
Thus, from the formula of $\hu_c$ \eqref{hu} and the definition 
\eqref{lambda_+} the Lemma follows at once.
\end{proof}

The above Lemma motivates the closer inspection of the smooth functions 
defined in \eqref{Delta_c}. By taking the compositions 
$\Delta_c(\lambda_1, \ldots, \lambda_n) \in C^\infty(P)$, elementary 
differentiations immediately reveal the symmetry property
\be
    \PD{\Delta_a(\lambda_1, \ldots, \lambda_n)}{\lambda_b}
    = \PD{\Delta_b(\lambda_1, \ldots, \lambda_n)}{\lambda_a}
    \qquad
    (a, b \in \bN_n).
\label{Delta_symmetry}
\ee
Therefore, since the phase space $P$ \eqref{P} is connected and simply 
connected, there is a globally defined function $\cS \in C^\infty(Q)$,
unique up to a constant, such that for the composition 
$\cS(\lambda_1, \ldots, \lambda_n)$ we have
\be
    \dd (\cS(\lambda_1, \ldots, \lambda_n))
    = \sum_{c = 1}^n \Delta_c(\lambda_1, \ldots, \lambda_n)\dd \lambda_c.
\label{dd_cS}
\ee
At this point we are in a position to prove the most important result of the 
paper. Making use of the relationships displayed in equation \eqref{theta_+} 
and in Lemma \ref{LEMMA_lambda_+_OK}, the Theorem below can be seen as a 
corollary of our scattering theoretical analysis culminating in Lemma 
\ref{LEMMA_asymptotic_variables_PBs}.

\begin{THEOREM}
\label{THEOREM_main_theorem}
The dual variables 
$\hlambda_1, \ldots, \hlambda_n, \htheta_1, \ldots, \htheta_n$ defined in 
equations \eqref{hlambda} and \eqref{htheta} form a global Darboux system 
on the phase space $P$ \eqref{P}. That is, we have
\be
    \PB{\hlambda_a}{\hlambda_b} = 0,
    \quad
    \PB{\htheta_a}{\htheta_b} = 0,
    \quad
    \PB{\hlambda_a}{\htheta_b} = \delta_{a, b}
    \qquad
    (a, b \in \bN_n).
\label{dual_variables_PBs}
\ee
As a consequence, the duality map $\Psi$ \eqref{Psi} is an 
anti-symplectomorphism, i.e., $\Psi^* \omega = - \omega$.
\end{THEOREM}

\begin{proof}
As a preliminary step, we introduce the temporary shorthand notation
\be
    \hDelta_c = \Delta_c(\htheta_1, \ldots, \htheta_n) \in C^\infty(P) 
    \qquad
    (c \in \bN_n).
\label{hat_Delta}
\ee
To proceed further, take arbitrary indices $a, b \in \bN_n$ and keep them 
fixed. Recalling \eqref{theta_+} and Lemma \ref{LEMMA_lambda_+_OK}, from 
the application of Lemma \ref{LEMMA_asymptotic_variables_PBs} it comes 
effortlessly that
\begin{align}
    & \PB{\htheta_a}{\htheta_b} 
    = \frac{1}{4} \PB{\theta_a^+}{\theta_b^+}  
    = 0,
    \label{PB_OK1}
    \\
    & \PB{\hlambda_a}{\htheta_b} 
    = \PB{\lambda_a^+}{\theta_b^+} 
        - \PB{\hDelta_a}{\htheta_b} 
    = \delta_{a, b} 
        - \sum_{c = 1}^n \PD{\hDelta_a}{\htheta_c} \PB{\htheta_c}{\htheta_b}
    = \delta_{a, b}.
    \label{PB_OK2}
\end{align}
Furthermore, the above two equations and the symmetry relation 
\eqref{Delta_symmetry} entail 
\be
\begin{split}
    \PB{\hlambda_a}{\hlambda_b} 
    & = 4 \PB{\lambda_a^+}{\lambda_b^+}  
    - 2 \left( 
            \PB{\lambda_a^+}{\hDelta_b} + \PB{\hDelta_a}{\lambda_b^+} 
        \right) 
    + \PB{\hDelta_a}{\hDelta_b}
    \\
    & = - 2 \sum_{c = 1}^n 
        \left( 
            \PB{\lambda_a^+}{\htheta_c} \PD{\hDelta_b}{\htheta_c} 
            + \PD{\hDelta_a}{\htheta_c} \PB{\htheta_c}{\lambda_b^+} 
        \right)
    + \sum_{c, d = 1}^n
        \PD{\hDelta_a}{\htheta_c} 
        \PB{\htheta_c}{\htheta_d} 
        \PD{\hDelta_b}{\htheta_d}
    \\
    & = - \sum_{c = 1}^n 
        \left( 
            \PB{\lambda_a^+}{\theta_c^+} \PD{\hDelta_b}{\htheta_c} 
            + \PD{\hDelta_a}{\htheta_c} \PB{\theta_c^+}{\lambda_b^+} 
        \right)
    = \PD{\hDelta_a}{\htheta_b} - \PD{\hDelta_b}{\htheta_a} 
    = 0, 
\end{split}
\label{PB_OK3}
\ee
and so the proof of the fundamental Poisson brackets 
\eqref{dual_variables_PBs} is complete. 

In terms of the symplectic form \eqref{omega}, the canonicity of the dual 
coordinates can be rephrased as
\be
    \omega = \sum_{c = 1}^n \dd \hlambda_c \wedge \dd \htheta_c.
\label{omega_and_dual_variables}
\ee
Recalling $\Psi$ \eqref{Psi} and \eqref{lambda_and_theta}, it is clear that 
$\forall c \in \bN_n$ we have $\Psi^* \lambda_c = \htheta_c$ and 
$\Psi^* \theta_c = \hlambda_c$, whence
\be
    \Psi^* \omega
    = \Psi^* \sum_{c = 1}^n \dd \lambda_c \wedge \dd \theta_c
    = \sum_{c = 1}^n \dd (\Psi^* \lambda_c) \wedge \dd (\Psi^* \theta_c)
    = \sum_{c = 1}^n \dd \htheta_c \wedge \dd \hlambda_c
    = - \omega
\label{JJJJJ1}
\ee
also follows immediately.
\end{proof}

\subsection{The wave and the scattering maps}
\label{SUBSECTION_W_and_S}
So far we have analyzed the asymptotics of the trajectories only for large
positive values of time. Recalling \eqref{lambda_and_theta_circ_Phi_s}, for
any $a \in \bN_n$ our results can be succinctly summarized as
\be
    \lambda_a \circ \Phi_s \sim s \sinh(\theta_a^+) + \lambda_a^+
    \midand
    \theta_a \circ \Phi_s \sim \theta_a^+
    \qquad
    (s \to \infty).
\label{AS+}
\ee
However, a thorough analysis of the scattering properties does require the 
study of the asymptotics for $s \to -\infty$, too. For this reason, let us
conjugate the matrix flow \eqref{matrix_flow_OK} with $\cR_N$ \eqref{cR}. 
Thereby from \eqref{spectral_identification_OK} we obtain
\be
    \{ e^{\pm 2 \lambda_a \circ \Phi(s, p)} \mid a \in \bN_n \}
    = \Spec 
        \left( 
            \cR_N \tilde{L}(p) \cR_N^{-1} 
            e^{2 s \sinh(\cR_N \bsTheta^+(p) \cR_N^{-1})} 
    \right)
    \qquad
    (s \in \bR, \, p \in P),
\label{spectral_identification_OK_-}
\ee
where 
\be
    \cR_N \bsTheta^+ \cR_N^{-1} = -\bsTheta^+. 
\label{bsTheta+_and_cR}
\ee    
The point is that, due to the appearance of the negative sign in the above 
equation, by applying Theorem \ref{THEOREM_Ruijsenaars} to the matrix flow
\be
    \bR \ni s
    \mapsto
    \cR_N \tilde{L}(p) \cR_N^{-1} e^{-2 s \sinh(\bsTheta^+(p))} \in GL(N, \bC)
\label{matrix_flow_OK_-}
\ee
we can infer the desired asymptotic results for $s \to -\infty$ as well. 
More precisely, by the methods of Subsections 
\ref{SUBSECTION_Applying_Ruijsenaars_thm} and \ref{SUBSECTION_Asymptotics}, 
for any $a \in \bN_n$ we can easily establish the asymptotics
\be
    \lambda_a \circ \Phi_s \sim s \sinh(\theta_a^-) + \lambda_a^-
    \midand
    \theta_a \circ \Phi_s \sim \theta_a^-
    \qquad
    (s \to - \infty),
\label{AS-}
\ee
where the asymptotic rapidities obey
\be
    \theta_a^- = - \theta_a^+ = - 2 \htheta_a,
\label{theta_-}
\ee 
whereas, in complete analogy with \eqref{lambda_+}, for the asymptotic 
positions we can write
\be
    \lambda_a^- 
    = \half \ln(m_a(\cR_N \tilde{L} \cR_N^{-1})).
\label{lambda_-}
\ee
As a matter of fact, mimicking the proof of Lemma \ref{LEMMA_lambda_+_OK} one 
finds immediately that
\be
    \lambda_a^- 
    = -\half \hlambda_a + \half \Delta_a(\htheta_1, \ldots, \htheta_n).
\label{lambda_-_OK}
\ee

Now, let us introduce the smooth manifolds
\be
    P^\pm 
    = \{ \zeta 
            = (\bsxi, \bseta) 
            = (\xi_1, \ldots, \xi_n, \eta_1, \ldots, \eta_n) \in \bR^N 
        \mid
        \eta_1 \gtrless \ldots \gtrless \eta_n \gtrless 0 \}
\label{P_pm}
\ee
endowed with the symplectic forms
\be 
    \omega^\pm = \sum_{a = 1}^n \dd x^\pm_a \wedge \dd y^\pm_a,
\label{omega_pm}
\ee 
where $x_a^\pm, y_a^\pm \in C^\infty(P^\pm)$ are convenient global coordinates 
on $P^\pm$ defined by the formulae
\be
    x_a^\pm(\zeta) = \xi_a 
    \midand
    y_a^\pm(\zeta) = \eta_a
    \qquad
    (a \in \bN_n, \, \zeta \in P^\pm).
\label{x_and_y_coords}
\ee
Utilizing the asymptotic positions and rapidities, at this point we define 
the wave maps
\be
    W_\pm \colon P \rightarrow P^\pm,
    \quad
    p 
    \mapsto 
    (\lambda_1^\pm(p), \ldots, \lambda_n^\pm(p), 
        \theta_1^\pm(p), \ldots, \theta_1^\pm(p)),
\label{W_pm}
\ee
that are of central interest in scattering theory.

\begin{THEOREM}
\label{THEOREM_wave_maps}
Both wave maps $W_\pm$ \eqref{W_pm} are symplectomorphisms. Moreover, the 
corresponding scattering map 
\be
    S = W_+ \circ W_-^{-1} \colon P^- \rightarrow P^+
\label{scattering_map}
\ee
is also a symplectomorphism of the form
\be
    S(\bsxi, \bseta) 
    = \left( 
        -\xi_1 + \Delta_1 (- \bseta / 2), 
        \ldots, 
        -\xi_n + \Delta_n (- \bseta / 2), 
        - \eta_1, \ldots, -\eta_n
    \right),
\label{scattering_formula}
\ee
where 
$(\bsxi, \bseta) = (\xi_1, \ldots, \xi_n, \eta_1, \ldots, \eta_n) \in P^-$.
\end{THEOREM}

\begin{proof}
Let us define the auxiliary maps $\Upsilon_\pm \colon P \rightarrow P^\pm$ by 
the formulae
\be
    \Upsilon_\pm(p) 
    = \left(
        \pm \half \eta_1 + \half \Delta_1(\bsxi), 
        \ldots, 
        \pm \half \eta_n + \half \Delta_n(\bsxi),
        \pm 2 \xi_1, \ldots, \pm 2 \xi_n
    \right),
\label{Upsilon_pm}
\ee
where 
$p = (\bsxi, \bseta) = (\xi_1, \ldots, \xi_n, \eta_1, \ldots, \eta_n) \in P$.
It is easy to see that both $\Upsilon_+$ and $\Upsilon_-$ are diffeomorphisms 
with inverses $\Upsilon_\pm^{-1} \colon P^\pm \rightarrow P$ given by 
\be
    \Upsilon_\pm^{-1}(\zeta) 
    = \left(
        \pm \half \eta_1, \ldots, \pm \half \eta_n,
        \pm 2 \xi_1 \mp \Delta_1(\pm \bseta / 2), 
        \ldots, 
        \pm 2 \xi_n \mp \Delta_n(\pm \bseta / 2)
    \right),
\label{Upsilon_pm_inv}
\ee
where 
$\zeta = (\bsxi, \bseta) = (\xi_1, \ldots, \xi_n, \eta_1, \ldots, \eta_n) 
\in P^\pm$. 

Next, recalling the coordinate functions displayed in 
\eqref{lambda_and_theta} and \eqref{x_and_y_coords}, from the definition 
\eqref{Upsilon_pm} it is clear that
\be
    \Upsilon_\pm^* x_a^\pm 
    = x_a^\pm \circ \Upsilon_\pm 
    = \pm \half \theta_a + \half \Delta_a(\lambda_1, \ldots, \lambda_n)
    \midand
    \Upsilon_\pm^* y_a^\pm 
    = y_a^\pm \circ \Upsilon_\pm
    = \pm 2 \lambda_a.
\label{Upsilon_pull_backs}
\ee
Thus, focusing first only on $\Upsilon_+$, the observation we made in 
\eqref{dd_cS} allows us to write that
\be
    \Upsilon_+^* \sum_{a = 1}^n x_a^+ \dd y_a^+
    = \sum_{a = 1}^n 
        \left(
            \theta_a + \Delta_a(\lambda_1, \ldots, \lambda_n)
        \right) \dd \lambda_a
    = \sum_{a = 1}^n \theta_a \dd \lambda_a 
        + \dd (\cS(\lambda_1, \ldots, \lambda_n)).
\label{B1}
\ee
Now, by taking the exterior derivative of the above equation, from 
the definitions \eqref{omega} and \eqref{omega_pm} it is clear that
\be
    \Upsilon_+^* \omega^+ 
    = \sum_{a = 1}^n \dd \theta_a \wedge \dd \lambda_a 
        + \dd^2 (\cS(\lambda_1, \ldots, \lambda_n))
    = - \omega.
\label{B2}
\ee
To put it simple, the diffeomorphism $\Upsilon_+$ is an 
anti-symplectomorphism. The same technique allows us to infer the 
relationship $\Upsilon_-^* \omega^- = -\omega$, too. We mention in 
passing that, on account of \eqref{B1}, the composition 
$\cS(\lambda_1, \ldots, \lambda_n) \in C^\infty(P)$ can be seen as 
a generating function of $\Upsilon_+$.

Now, looking back to the definitions \eqref{Psi}, \eqref{W_pm} and 
\eqref{Upsilon_pm}, our key observation is that the wave maps can be 
realized as the compositions
\be
    W_\pm = \Upsilon_\pm \circ \Psi.
\label{W_and_Upsilon_and_Psi}
\ee
Therefore, by invoking Theorem \ref{THEOREM_main_theorem}, we conclude at
once that both $W_+$ and $W_-$ are symplectomorphisms. As concerns the 
scattering map \eqref{scattering_map}, from \eqref{W_and_Upsilon_and_Psi} 
it is evident that
\be
    S = \Upsilon_+ \circ \Upsilon_-^{-1}.
\label{S_map_and_Upsilon} 
\ee
Thus, by exploiting the explicit expressions appearing in \eqref{Upsilon_pm} 
and \eqref{Upsilon_pm_inv}, the formula \eqref{scattering_formula} is also 
immediate.
\end{proof}

\section{Discussion}
\label{SECTION_Discussion}
Now, we are in a position to harvest some interesting consequences of our
results. Take an arbitrary smooth function
\be
    \chi \colon \cP \rightarrow \bR,
    \quad
    Y \mapsto \chi(Y)
\label{chi}
\ee
defined on the symmetric space $\cP$ \eqref{cP}, and suppose that it is 
\emph{invariant} under the action of conjugations by the elements of the 
compact subgroup $K$ \eqref{K}; that is,
\be
    \chi(k Y k^{-1}) = \chi(Y)
    \qquad
    (Y \in \cP, \, k \in K).
\label{chi_invariant}
\ee
Recalling the Lax matrix $L$ \eqref{L} and Lemma \ref{LEMMA_hy}, it is 
immediate that for the function
\be
    H^g_\chi = \chi(L^g) \in C^\infty(P)
\label{H_chi}
\ee
we can write that
\be
    H^g_\chi 
    = \chi \big( \hy^g e^{2 \hbsTheta^g} (\hy^g)^{-1} \big) 
    = \chi \big( e^{2 \hbsTheta^g} \big).
\label{H_chi_and_htheta}
\ee
In other words, $H^g_\chi$ \eqref{H_chi} can be seen as a function of the 
coordinates $\htheta^g_1, \ldots, \htheta^g_n$. As a consequence, the members
of the global Darboux system featuring in Theorem \ref{THEOREM_main_theorem} 
provide \emph{action-angle variables} for the Hamiltonian system 
$(P, \omega, H^g_\chi)$. Notice that the van Diejen type model \eqref{H} 
belongs to this distinguished family of Hamiltonian systems. Indeed, due
to the invariance of the trace functional, the relationship \eqref{H_and_L}
entails that
\be
    H^g 
    = \half \tr(L^g) 
    = \half \tr(e^{2 \hbsTheta^g})
    = \half \sum_{j = 1}^N e^{2 \hTheta^g_j}
    = \sum_{a = 1}^n \cosh(2 \htheta^g_a).
\label{H_and_htheta}
\ee 
From a practical point of view, the Hamiltonian $H^g$ \eqref{H} is singled 
out only by the fact that it has a relatively simple form. 

In the theory of integrable Hamiltonian systems one of the principal goals 
is the construction of action-angle variables. Due to the rather explicit 
descriptions of the variables $\htheta^g_a$ \eqref{L_spec} and $\hlambda^g_a$ 
\eqref{hlambda}, in this paper this task is completely accomplished for
the family of systems $(P, \omega, H^g_\chi)$, including the van Diejen 
model $H^g$ \eqref{H}. However, as discovered by Ruijsenaars 
\cite{Ruij_CMP1988, Ruij_RIMS_2, Ruij_RIMS_3}, it is a remarkable feature 
of the CMS and the RSvD systems that they can be arranged into pairs based 
on their action-angle maps. Now we are in a position to reveal the duality 
property of the $2$-parameter family of van Diejen type systems \eqref{H}, 
too. For this reason, with the aid of the dual Lax matrix $\hL$ \eqref{hL}, 
for each $K$-invariant smooth function $\chi$ \eqref{chi} we introduce the 
`dual Hamiltonian'
\be
    \hH^g_\chi = \chi(\hL^g) \in C^\infty(P).
\label{hH_chi}
\ee
Notice that
\be
    \hH^g_\chi = \chi \big( (y^g)^{-1} e^{2 \bsLambda} y^g \big) 
    = \chi \big( e^{2 \bsLambda} \big),
\label{hH_and_bsLambda}
\ee
whence the original canonical coordinates \eqref{lambda_and_theta} provide 
action-angle variables for $(P, \omega, \hH^g_\chi)$. Moreover, due to the 
relationship \eqref{L_and_hL}, we can write that
\be
    \hH^g_\chi = H^{\hg}_\chi \circ \Psi^g.
\label{H_and_hH}
\ee
The upshot of this observation is that, up to the anti-symplectomorphism 
$\Psi^g$ \eqref{Psi}, the Hamiltonian systems $(P, \omega, \hH^g_\chi)$ and 
$(P, \omega, H^{\hg}_\chi)$ can be identified. We also see that, at the level 
of the parameter space, the systems in duality are related by the involution 
\eqref{tmfM_involution}. Since under this involution the Hamiltonian \eqref{H} 
transforms into itself, i.e., $H^{\hg} = H^g$, we may say that the van Diejen 
systems of our interest are \emph{self-dual}.

In order to uncover this new case of duality, in this paper we adapted 
Ruijsenaars' ideas \cite{Ruij_CMP1988} to the geometric picture introduced 
in \cite{Pusztai_Gorbe}. Indeed, the construction of the functions $\htheta_a$ 
\eqref{htheta} and $\hlambda_a$ \eqref{hlambda} is built upon the careful 
analysis of the Ruijsenaars type commutation relation \eqref{commut_rel}. 
However, to prove their smoothness and canonicity, we departed from the 
complex analytic approach advocated by \cite{Ruij_CMP1988}. We believe that 
our method presented in Section \ref{SECTION_Scattering_theory_and_duality} 
for proving the canonicity is a bit more general in the sense that, once 
the temporal asymptotics with the uniformity assertion is established as in 
Lemma \ref{LEMMA_asymptotics_1}, this technique may be applied to a much wider 
class of Hamiltonian systems describing repulsive particles, even under weaker 
smoothness conditions. Of course, for real-analytic Hamiltonians, the question 
of the real-analyticity of the pertinent objects would require further 
analysis. 

Turning to the scattering properties of the $2$-parameter family of 
hyperbolic van Diejen systems \eqref{H}, from Theorem \ref{THEOREM_wave_maps} 
we see that, up to an overall sign, the asymptotic rapidities are preserved. 
Moreover, from \eqref{Delta_c} it is also clear that the classical phase 
shifts are completely determined by the $1$-particle and the $2$-particle 
scattering processes. In other words, the scattering map 
\eqref{scattering_formula} has a \emph{factorized form}. Note that this 
peculiar feature seems to be characteristic to the CMS and the RSvD type 
many-particle systems. Indeed, for the models associated with the $A$-type 
root systems it is known for a long time (see e.g. \cite{Kulish_1976, 
Moser_1977, Ruij_CMP1988}), and recently it has been proved for the 
hyperbolic $BC_n$ Sutherland and the rational $BC_n$ RSvD models, too 
\cite{Pusztai_JPA2011, Pusztai_NPB2013}. Thus, following 
\cite{Ruij_FiniteDimSolitonSystems}, the content of Theorem 
\ref{THEOREM_wave_maps} can be rephrased by saying that the hyperbolic 
van Diejen system \eqref{H} is also a $BC$-type 
\emph{finite dimensional pure soliton system}.

To conclude the paper, let us recall that in \cite{Pusztai_Gorbe} we 
conjectured a Lax matrix for a bit more general $3$-parameter family of 
hyperbolic van Diejen systems (see equations (6.3-7) in \cite{Pusztai_Gorbe}). 
Therefore, by generalizing the arguments of the present paper, it would be 
the natural next step to construct action-angle variables for these systems, 
too. We expect that the members of this $3$-parameter family are also 
self-dual with factorized scattering maps.

\medskip
\noindent
\textbf{Acknowledgments.}
We are grateful to L.~Feh\'er for useful comments on the manuscript.
The work of B.G.P. was supported by the J\'anos Bolyai Research Scholarship 
of the Hungarian Academy of Sciences, by the Hungarian Scientific Research 
Fund (OTKA grant K116505), and also by a Lend\"ulet Grant; he wishes to thank 
Z.~Bajnok for hospitality in the MTA Lend\"ulet Holographic QFT Group.

\appendix
\numberwithin{THEOREM}{section}

\newcommand{\de}{\delta}
\newcommand{\C}{{\mathbb C}} 
\newcommand{\R}{{\mathbb R}}
\newcommand{\re}{{\rm Re}\, }

\section{Spectral asymptotics revisited (by Simon Ruijsenaars)}
\label{AppA}

The results encoded in Theorem~A2 of Ref.~\cite{Ruij_CMP1988} concerning spectral asymptotics play a crucial role in G\'abor's work on the van Diejen systems. My proof in Ref.~\cite{Ruij_CMP1988} contains elaborate computations involving resolvents and nested Neumann expansions. By contrast, in this appendix the essence of Theorem~A2 is recovered by exploiting the holomorphic implicit function theorem. The present method gives rise to a novel proof that is not only shorter and simpler, but also yields more explicit information. At the end of the appendix we also reconsider Theorem~A1 of Ref.~\cite{Ruij_CMP1988} along the same lines. Hence the counterparts of these older theorems are the present Theorems~\ref{THEOREM_THM_A1} and~\ref{THEOREM_THM_A2}, resp.
 
The appendix is concerned with two families of complex $N\times N$ matrices. The first one consists of diagonal matrices
\be\label{cD}
\cD \equiv \{ D=\diag (d_1,\ldots,d_N)\mid \re (d_N)<\cdots <\re (d_1)\}.
\ee
We shall use the notation
\be\label{SSSSS1}
\mu_j\equiv \re (d_j-d_{j+1}),\ \ \ j=1,\ldots,N-1,
\ee
\be\label{defR}
R\equiv \min (\mu_1,\ldots,\mu_{N-1}).
\ee

The second family consists of matrices $M$ that satisfy a restriction on their principal minors for our first theorem. Letting $1\le i_1<\cdots <i_j\le N$, $j=1,\ldots,N$, we denote the $j\times j$ minor involving these indices by $M(i_1,\ldots,i_j)$. Also, for the special case $i_k=k$, $k=1,\ldots, j$, we denote the corresponding principal minor by $\pi_j$. Thus we have in particular
\be\label{SSSSS2}
\pi_1 =M_{11},\ \ \ \pi_2=M_{11}M_{22}-M_{12}M_{21},\ \  \ \pi_N=|M|,
\ee
with $|M|$ the determinant of~$M$. The second family is now given by
\be\label{cM}
\cM\equiv \{ M\in \C^{N\times N}\mid \pi_j\ne 0,\ \ j=1,\ldots,N\}.
\ee
As a consequence, for $M\in\cM$ we may introduce nonzero complex numbers
\be\label{mj}
m_1\equiv M_{11},\ \ \ m_j\equiv \pi_j/\pi_{j-1},\ \ \ j=2,\ldots,N.
\ee

Our goal is now to elucidate the spectral asymptotics for matrices of the form
\be\label{SSSSS3}
E(t)\equiv M\exp(tD),\ \ \ M\in\cM,\ \ D\in\cD,\ \ t\in\R,
\ee
as $t\to\infty$. As will transpire, the spectrum is simple for $t$ sufficiently large, and the dominant asymptotics of the eigenvalues $\lambda_1(t),\ldots,\lambda_N(t)$   is given by
\be\label{domas}
\lambda_j(t)\sim m_j\exp(td_j),\ \ \ j=1,\ldots,N,\ \ \ t\to\infty.
\ee
(Note this yields an ordering $|\lambda_N(t)|<\cdots <|\lambda_1(t)|$ for~$t$ large,  due to the $d$-restriction in~\eqref{cD} and $m_j$ being nonzero. Note also that~\eqref{domas} is plain for a diagonal $M\in\cM$.)

For applications, however, it is crucial to improve considerably on the dominant asymptotics~\eqref{domas}. For the $N=2$ case this is readily done, since the eigenvalues can be calculated explicitly. However, this direct calculation yields no clue how to proceed for arbitrary~$N$. 

For a better understanding of the method followed for general~$N$, we begin by detailing it for the $N=2$ case. This serves to exemplify all steps of the flow chart without the inevitable notational clutter associated with the general case.

First, we define quantities~$c_j(t)$, $j=1,2$, by setting
\be\label{cj}
\lambda_j(t)=c_j(t)m_j\exp(td_j),
\ee
and note that the eigenvalues solve the system
\be\label{SSSSS4}
\lambda_1+\lambda_2= M_{11}\exp(td_1)+M_{22}\exp(td_2),\ \ \ \ \  \ \  \lambda_1 \lambda_2= |M| \exp(td_1+td_2).
\ee
Introducing
\be\label{SSSSS5}
\epsilon \equiv \exp(td_2-td_1),
\ee
this can be rewritten as
\be\label{Fjsys}
F_j(\epsilon;c_1,c_2)=0,\ \ \ j=1,2,
\ee
with
\be\label{F12}
F_1\equiv c_1+c_2 \frac{|M|}{M_{11}^2}\epsilon-1-\frac{M_{22}}{M_{11}}\epsilon,
\ \ \ \ \ \ 
F_2\equiv c_1c_2-1.
\ee

It is plain that this system has a solution
\be\label{SSSSS6}
F_j(0;1,1)=0,\ \ \ j=1,2.
\ee
Moreover, the matrix
\be\label{SSSSS7}
D_cF(0;1,1)=\left( 
\begin{array}{cc}
1 & 0 \\
1 & 1 \\
\end{array}
\right),
\ee
is regular and $F$ is entire in $\epsilon,c_1$ and $c_2$. Therefore, we may invoke the holomorphic implicit function theorem.
This theorem implies that for $|\epsilon|<a$ with~$a$ sufficiently small, there exists a unique holomorphic solution~$c=s(\epsilon)$ with the properties
\be\label{sprop}
s(0)=(1,1),\ \ |s_j(\epsilon)-1|\le 1/2,\ \ j=1,2,\ \ |\epsilon|<a.
\ee
 
Defining $T_0$ by
\be\label{SSSSS8}
\exp(-T_0\mu_1)= a/2,
\ee
we now put
\be\label{cjsj}
c_j(t)\equiv s_j(\exp( td_2-td_1)),\ \ \ t\ge T_0,\ \ \ j=1,2.
\ee
By~\eqref{sprop}, this yields
\be\label{SSSSS9} 
|c_j(t)|\in [1/2,3/2],\ \ j=1,2,\ \ t\ge T_0.
\ee
Choosing next $T_1\ge T_0$ such that
\be\label{SSSSS10}
|M_{11}|\exp(T_1\re d_1)>3||M|/M_{11}|\exp(T_1\re d_2),
\ee
we get from~\eqref{cj} eigenvalues satisfying
\be\label{SSSSS11}
|\lambda_2(t)|<|\lambda_1(t)|,\ \ \ \forall t\in [T_1,\infty).
\ee
In particular, it follows that $\sigma(E(t))$ is simple for all $t\ge T_1$.

To proceed, we note that by holomorphy of the functions $s_j(z)$ for $|z|<a$, we can find $\de\in(0,a/2]$ such that
\be\label{SSSSS12}
\sup_{|z|\le\de} |s'_j(z)|\le |s'_j(0)|+\eta,\ \ \ j=1,2,
\ee
with $\eta$ an arbitrary fixed positive number.
Defining $T_2$ by
\be\label{SSSSS13}
\exp(-T_2\mu_1)=\de,
\ee
we now put
\be\label{SSSSS14}
T_E\equiv \max(T_1,T_2).
\ee
Then we conclude from~\eqref{cjsj} that for all $t\ge T_E$ we have
\be\label{SSSSS15}
|c_j(t)-1|\le \exp (-t\mu_1)(|s'_j(0)|+\eta),
\ee
\be\label{SSSSS16}
|\dot{c}_j(t)|\le \exp (-t\mu_1)|d_2-d_1|(|s'_j(0)|+\eta).
\ee

We can make these bounds more explicit by calculating $s'_j(0)$ from the system~\eqref{Fjsys}. Indeed, using
\be\label{SSSSS17}
s_j(\epsilon)=1 +\epsilon v_j +O(\epsilon^2),\ \ \ \epsilon\to 0,
\ee
we deduce from~\eqref{F12} that $v_2$ equals $-v_1$ and  that we have
\be\label{v1}
v_1=\frac{1}{M_{11}}\Big( M_{22}-\frac{|M|}{M_{11}}\Big).
\ee
 
The upshot of our reasoning is that for all $t\ge T_E$ the spectrum of $E(t)$ is simple, with eigenvalues of the form
\be\label{SSSSS18}
\lambda_1(t)=M_{11}\exp(td_1)(1+\rho_1(t)),
 \ \ \ \ \ \lambda_2(t)=\frac{|M|}{M_{11}}\exp(td_2)(1+\rho_2(t)).
\ee
Here, the remainder terms are majorized by
\be\label{SSSSS19}
|\rho_j(t)|\le \exp(-t\mu_1)(|v_1|+\eta),\ \ \ j=1,2,
\ee
\be\label{SSSSS20}
|\dot{\rho}_j(t)|\le \exp(-t\mu_1)|d_2-d_1|(|v_1|+\eta),\ \ \ j=1,2,
\ee
with $v_1$ given by~\eqref{v1}. (Note that $v_1$ vanishes when $M$ is upper or lower triangular, as it should.)

We are now prepared for the general case.
\begin{THEOREM}
\label{THEOREM_THM_A1}
There exists $T_E\in\R$ such that the $N\times N$ matrix
\be\label{SSSSS21}
E(t)=M\exp(tD),\ \ M\in\cM,\ \ D\in\cD,\ \ t\in\R,
\ee
has nondegenerate eigenvalues $\lambda_1(t),\ldots,\lambda_N(t)$ satisfying
\be\label{lamord}
|\lambda_N(t)|<\cdots <|\lambda_1(t)|,\ \ \ \forall t\ge T_E.
\ee
They are of the form
\be\label{lamfo}
\lambda_j(t)=m_j\exp(td_j)[1+\rho_j(t)],\ \ \ t\ge T_E.
\ee
Here, the remainder functions are real-analytic on~$(T_E,\infty)$ and  satisfy
\be\label{SSSSS22}
|\rho_1(t)|\le  \exp(-t\mu_{1})q_1, 
\ee
\be\label{SSSSS23}
|\rho_j(t)|\le \exp(-t\mu_{j-1})q_{j-1} +\exp(-t\mu_{j})q_{j},\ \ \ j=2,\ldots,N-1,
\ee
\be\label{SSSSS24}
|\rho_N(t)|\le  \exp(-t\mu_{N-1})q_{N-1}, 
\ee
while their time derivatives satisfy
\be\label{SSSSS25}
|\dot{\rho}_1(t)|\le  \exp(-t\mu_{1})|d_{2}-d_{1}|q_{1}, 
\ee
\be\label{SSSSS26}
|\dot{\rho}_j(t)|\le \exp(-t\mu_{j-1})|d_{j}-d_{j-1}|q_{j-1} +\exp(-t\mu_{j})|d_{j+1}-d_{j}|q_{j},\ j=2,\ldots,N-1,
\ee
\be\label{rdN}
|\dot{\rho}_N(t)|\le  \exp(-t\mu_{N-1})|d_{N}-d_{N-1}|q_{N-1}. 
\ee
With $\eta$ an arbitrary fixed positive number, the $q_j$'s are given by
\be\label{defqj}
q_j\equiv |p_j| +\eta,\ \ \ \ j=1,\ldots,N-1,
\ee
where
\be\label{defp1}
p_1\equiv \frac{1}{M_{11}}\Big( M_{22}-\frac{M(1,2)}{M_{11}}\Big),
\ee
and
\be\label{defpj}
p_j\equiv \frac{M(1,\ldots,j-1,j+1)}{M(1,\ldots,j)}-\frac{m_{j+1}}{m_j},\ \ \ j=2,\ldots,N-1.
\ee
Furthermore, $T_E$ can be chosen uniformly for $(M,D)$ varying over an arbitrary compact subset of $\cM\times \cD$. 
\end{THEOREM}
\begin{proof}
We define $c_1,\ldots,c_N$ by
\be\label{csjN}
\lambda_j=c_jm_j\exp(td_j),\ \ \ j=1,\ldots,N,
\ee
where $\lambda_1,\ldots,\lambda_N$ are the solutions to $|E-\lambda {\bf 1}_N|=0$. Thus the latter solve the system
\be\label{SSSSS27}
\sum_{i_1<\cdots < i_j}\lambda_{i_1}\cdots \lambda_{i_j}=\sum_{i_1<\cdots < i_j}M(i_1,\ldots,i_j)\exp(td_{i_1}+\cdots +td_{i_j}),\ \ \ j=1,\ldots,N.
\ee
Introducing
\be\label{epsj}
\epsilon_j \equiv \exp(td_{j+1}-td_j),\ \ \ j=1,\ldots,N-1,
\ee
this can be rewritten
\be\label{FjsysN}
F_j(\epsilon_1,\ldots,\epsilon_{N-1};c_1,\ldots,c_N)=0,\ \ \ j=1,\ldots,N,
\ee
where
\begin{multline}\label{defFj}
F_j \equiv \sum_{i_1<\cdots <i_j}\Big( \frac{m_{i_1}\cdots m_{i_j}}{m_1\cdots m_j}c_{i_1}\cdots c_{i_j}-\frac{M(i_1,\ldots,  i_j)}{\pi_j}\Big)
\prod_{k_1=1}^{i_1-1} \epsilon_{k_1}\cdots \prod_{k_j=j}^{i_j-1}\epsilon_{k_j},\ \ \ j=1,\ldots,N.
\end{multline}

Clearly, this system has a solution
\be\label{SSSSS28}
F_j(0,\ldots,0;1,\ldots,1)=0,\ \ \ j=1,\ldots,N,
\ee
and it is not hard to verify
\be\label{SSSSS29}
(D_cF)(0,\ldots,0;1,\ldots,1)_{ij}=\left\{ 
\begin{array}{cc}
 0,  &  i<j, \\
 1, & i\ge j.\\
 \end{array}
 \right.
 \ee
 Since this yields a regular matrix  and $F$ is entire in 
 \be\label{SSSSS30}
 \epsilon\equiv (\epsilon_1,\ldots,\epsilon_{N-1}),\ \ \  c\equiv (c_1,\ldots,c_N),
 \ee
   the holomorphic implicit function theorem can be invoked.
It entails that for $\epsilon$ in the polydisc 
\be\label{polyd}
|\epsilon_1|,\ldots,|\epsilon_{N-1}|<a,
\ee
 with~$a$ small enough, we get a unique holomorphic solution~$c=s(\epsilon)$ fulfilling
\be\label{spropN}
s_j(0)=1,\ \ |s_j(\epsilon)-1|\le 1/2,\ \ j=1,\ldots,N.
\ee
 
We now choose $T_0$ satisfying
\be\label{T0}
\exp(-T_0 R)= a/2,
\ee
with $R$ given by~\eqref{defR}.
Defining
\be\label{cjsjN}
c_j(t)\equiv s_j(\exp( td_2-td_1),\ldots,\exp( td_N-td_{N-1})),\ \ \ t\ge T_0,\ \ \ j=1,\ldots,N,
\ee
we then deduce from~\eqref{spropN} that we have
\be\label{cjbo} 
|c_j(t)|\in [1/2,3/2],\ \ j=1,\ldots,N, \ \ \ t\ge T_0.
\ee

Next, we choose $T_1\ge T_0$ such that
\be\label{mjbo}
|m_j|\exp(T_1\mu_j)>3|m_{j+1}|,\ \ \ j=1,\ldots,N-1.
\ee
Then we get from~\eqref{csjN} the ordering~\eqref{lamord}, which implies that $\sigma(E(t))$ is simple for all $t\ge T_1$.

We now observe that since  the functions $s_j(\epsilon)$ are holomorphic in the polydisc~\eqref{polyd}, for an arbitrary $\eta>0$ we can find $\de\in(0,a/2]$ with
\be\label{sjsup}
\sup_{|z_1|,\ldots,|z_{N-1}|\le\de} |\partial_ks_j(z)|\le |\partial_ks_j(0)|+\eta,\ \ k=1,\ldots,N-1,\  \ j=1,\ldots,N.
\ee
Using
\begin{multline}\label{SSSSS100}
s_j(z_1,\ldots,z_{N-1})-s_j(0)=\int_0^{z_1}\partial_1 s_j(w_1,0,\ldots,0)dw_1 \\
+\int_0^{z_2}\partial_2 s_j(z_1,w_2,0,\ldots,0)dw_2 +\cdots + \int_0^{z_{N-1}}\partial_{N-1} s_j(z_1,\ldots,z_{N-2},w_{N-1})dw_{N-1},
\end{multline}
we infer from this that we have bounds
\be\label{sjbou}
|s_j(z)-1|\le \sum_{k=1}^{N-1}|z_k|(|\partial_k s_j(0)|+\eta),\ \ \ |z_1|,\ldots,|z_{N-1}|\le \de,\ \ j=1,\ldots,N.
\ee

Defining now $T_2$ by
\be\label{SSSSS40}
\exp(-T_2 R)=\de,
\ee
we set
\be\label{TE}
T_E\equiv \max(T_1,T_2).
\ee
Then we see from~\eqref{cjsjN} that for all $t\ge T_E$ we have
\be\label{cjN}
|c_j(t)-1|\le \sum_{k=1}^{N-1}\exp (-t\mu_k)(|\partial_k s_j(0)|+\eta),\ \ j=1,\ldots,N,
\ee
\be\label{cjNd}
|\dot{c}_j(t)|\le \sum_{k=1}^{N-1}|d_{k+1}-d_k|\exp (-t\mu_k)(|\partial_k s_j(0)|+\eta),\ \ \ \ j=1,\ldots,N.
\ee

We proceed to calculate the partials $\partial_ks_j(0)$ from the system~\eqref{FjsysN}. To this end we write
\be\label{SSSSS41}
s_j(\epsilon)=1 +\sum_{k=1}^{N-1}v_{jk}\epsilon_k +\mathrm{h.~o.},\ \ \ v_{jk}\equiv \partial_ks_j(0),
\ee
where h.~o.~denotes terms of higher order in the power series expansion around $\epsilon=0$. Substituting this for $c_j$ in~\eqref{FjsysN}, we collect all terms linear in $\epsilon_1,\ldots,\epsilon_{N-1}$. Defining $p_j$ by~\eqref{defpj}, the resulting linear system can then be written as
\be\label{SSSSS42}
\sum_{k=1}^{N-1}\sum_{i=1}^{j}v_{ik}\epsilon_k-p_j\epsilon_j=0, \ j=1,\ldots,N-1,
\ \ \  \ \ \ \ \sum_{k=1}^{N-1}\sum_{i=1}^{N}v_{ik}\epsilon_k=0.
\ee

Now by holomorphy the coefficients of $\epsilon_1,\ldots,\epsilon_{N-1}$ must vanish. From this we conclude first that for $j=1$ the partials $ \partial_ks_j(0)$ are given by
\be\label{SSSSS43}
v_{1k}=\left\{ 
\begin{array}{ll}
 p_1,  & \ \  k=1, \\
 0, & \ \ k>1,\\
 \end{array}
 \right.
 \ee
 and then recursively for $j=2,\ldots,N-1$, by
\be\label{SSSSS44}
v_{jk}=\left\{ 
\begin{array}{ll}
 -p_{j-1},  & \ \ \  k=j-1, \\
 p_j, & \ \ \ k=j,\\
 0, & \ \ \ \mathrm{otherwise},\\
 \end{array}
 \right.
 \ee
 while for $j=N$ we finally get
 \be\label{SSSSS45}
 v_{Nk}=\left\{ 
\begin{array}{ll}
 -p_{N-1},  & \ \  k=N-1, \\
 0, & \ \ k<N-1.\\
 \end{array}
 \right.
 \ee
 Substituting this in \eqref{cjN}--\eqref{cjNd}, we arrive at~\eqref{lamfo}--\eqref{rdN}. As a consequence, it remains to prove the uniformity claim.
 
To this end we begin by observing that the functions $F_j$ defined by~\eqref{defFj} are not only entire functions of $\epsilon$ and $c$, but also holomorphic functions of the matrix elements $M_{ij}$ on the open subset $\cM$ of $\C^{N^2}$.  Moreover, recalling~\eqref{epsj}, we see that they are entire in~$t$ and in~$d_1,\ldots,d_N$. We now fix $M_0\in\cM$ and $D_0\in\cD$ and view the set~$\cD$ as an open subset of~$\C^N$ in the obvious way. 

Next, we consider closed polydiscs in $\cM$ and $\cD$ around $M_0$ and $D_0$, with nonzero radii~$r_1$ and~$r_2$. We need only show that we can choose $T_E$ uniformly for $M$ and $D$ varying over such polydiscs $P_0(r_1)\times Q_0(r_2)\subset \cM\times\cD$ when we suitably decrease~$r_1$ and~$r_2$ if need be. (Indeed, any compact subset of $\cM\times\cD$ is covered by a finite number of such polydiscs.)

To prove this, we retrace the above steps, as applied to $M_0\exp (tD_0)$. First, we note that the holomorphic implicit function theorem implies that when we choose not only~$a$ in~\eqref{polyd}, but also $r_1$ small enough, then we get a unique solution $s_j(\epsilon,M)$ obeying~\eqref{spropN} and holomorphic in the Cartesian product of the~$\epsilon$- and $M$-polydiscs.

Second, we define $T_0$ by
\be\label{SSSSS46}
\exp(-T_0m(r_2))=a/2,\ \ m(r_2)\equiv \min_{D\in Q_0(r_2)} (R(D)),
\ee
with~$R(D)$ given by~\eqref{defR}. Then~\eqref{cjbo} follows again.

Third, we choose $T_1\ge T_0$ such that the inequalities~\eqref{mjbo}
 hold true on  $P_0(r_1)\times Q_0(r_2)$. Then we obtain the ordering~\eqref{lamord} and nondegeneracy of $\sigma(E(t))$ for all $(M,D)\in P_0(r_1)\times Q_0(r_2)$ and all $t\ge T_1$.
 
Fourth, choosing $\de\in (0,a/2]$ and eventually decreasing~$r_1$ such that \eqref{sjsup} holds on the product of the closed $\epsilon$-polydisc with radius~$\de$ and $P_0(r_1)$, the bounds~\eqref{sjbou} follow as before.

Finally, defining $T_2$ by
\be\label{SSSSS47}
\exp(-T_2m(r_2))=\de, 
\ee
it follows that our uniformity assertion holds true for $T_E$~\eqref{TE}.
\end{proof}

In Appendix~A of Ref.~\cite{Ruij_CMP1988} we also studied the spectral asymptotics for $t\to\infty$ of matrices of the form $M+tD$,  with $M$ an arbitrary $N\times N$ matrix and $D\in\cD$. We have meanwhile realised that (a slightly different version of) the pertinent result (namely, Theorem A1 in Ref.~\cite{Ruij_CMP1988}) can be readily understood via Rayleigh--Schr\"odinger perturbation theory, cf.~for example \cite{Reed_Simon_VOL_4}, p.~7.

Specifically, the spectral asymptotics of $M+tD$  for $t\to\infty$ can be readily deduced from the behavior of the spectrum $\sigma(D+\epsilon M)$ for $|\epsilon|$ small enough. Indeed, since $D$ has simple spectrum, we can find $r>0$ such that this is still true for $D+\epsilon M$ with $|\epsilon|\le r$.  Then Rayleigh--Schr\"odinger perturbation theory can be invoked to infer that $D+\epsilon M$ has eigenvalues of the form
\be\label{SSSSS50}
d_j +\epsilon M_{jj} +\epsilon^2 \alpha_j +O(\epsilon^3),\ \ \epsilon \to 0,\  \ j=1,\ldots,N,
\ee
where
\be\label{alj}
\alpha_j\equiv \sum_{k\ne j} \frac{M_{jk}M_{kj}}{d_j-d_k},\ \ \ j=1,\ldots,N.
\ee
As a consequence, $M+tD$ has eigenvalues that satisfy estimates
\be\label{SSSSS51}
\lambda_j(t)=M_{jj}+ td_j+t^{-1}\alpha_j +O(t^{-2}),\ \ t\to\infty,\ \ \ j=1,\ldots,N.
\ee

This result can also be recovered and slightly improved by the method we followed to prove Theorem~\ref{THEOREM_THM_A1}. This yields the following theorem, which concludes this appendix.

\begin{THEOREM}
\label{THEOREM_THM_A2}
There exists $T_E\in\R$ such that the matrix
\be\label{SSSSS52}
E(t)=M+tD,\ \ M\in\C^{N\times N},\ \ D\in\cD,\ \ t\in\R,
\ee
has nondegenerate eigenvalues $\lambda_1(t),\ldots,\lambda_N(t)$, satisfying 
\be\label{lamord2}
|\lambda_N(t)|<\cdots <|\lambda_1(t)|,\ \ \ \forall t\ge T_E.
\ee
They are of the form
 \be\label{lamM}
\lambda_j(t)=M_{jj}+ td_j+\rho_j(t),\ \ \ j=1,\ldots,N,
\ee
where the $\rho_j(t)$'s are real-analytic functions on~$(T_E,\infty)$ such that
\be\label{rhob}
|\rho_j(t)|\le \frac{1}{t}(|\alpha_j|+\eta),\ \ \ 
|\dot{\rho}_j(t)|\le \frac{1}{t^2}(|\alpha_j|+\eta),\ \ \ j=1,\ldots,N.
\ee
Here, $\eta$ is an arbitrary fixed positive number and the $\alpha_j$ are given by~\eqref{alj}.
Moreover, $T_E$ can be chosen uniformly for $(M,D)$ varying over an arbitrary compact subset of $\C^{N\times N}\times \cD$. 
\end{THEOREM}
\begin{proof}
The crux is that the holomorphic function theorem again applies to the case at hand, provided we choose a suitable starting point. Specifically, let $\lambda_j(t)$, $j=1,\ldots,N$, be the solutions to $|E(t)-\lambda{\bf 1}_N|=0$. Then we set
\be\label{clam}
c_j(\epsilon)\equiv \epsilon\lambda_j(1/\epsilon),\ \ \ j=1,\ldots,N.
\ee
In this case the associated spectral system is given by
\be\label{SSSSS53}
F_j(\epsilon;c_1,\ldots,c_N)=0,\ \ \ j=1,\ldots,N,
\ee
where the functions $F_j$ are of the form
\begin{multline}\label{Fjep}
F_j  =\sum_{i_1<\cdots <i_j}\Big(  c_{i_1}\cdots c_{i_j} - d_{i_1}\cdots d_{i_j} -\epsilon \sum_{k=1}^{j }M_{i_ki_k}d_{i_1}\cdots \hat{d}_{i_k}\cdots d_{i_j}\\
 -\epsilon^2 \sum_{1\le k<l\le j}M(i_k,i_l)d_{i_1}\cdots \hat{d}_{i_k}\cdots \hat{d}_{i_l}\cdots d_{i_j} \Big)+O(\epsilon^3).
\end{multline}
Here the hat signifies that the pertinent $d_i$ should be omitted.

This system has an obvious solution
\be\label{SSSSS54}
F_j(0;d_1,\ldots,d_N)=0,\ \ \ j=1,\ldots,N,
\ee
with corresponding partial matrix
\be\label{SSSSS55}
(D_c F)(0;d_1,\ldots,d_N)=\left(
\begin{array}{ccc}
1 & \cdots & 1 \\
d_2+\cdots +d_N  &  \cdots & d_1+\cdots +d_{N-1} \\
\vdots & \vdots & \vdots \\
d_2\cdots d_N & \cdots & d_1\cdots d_{N-1} \\
\end{array}\right),
\ee
and it is not difficult to verify
\be\label{SSSSS56}
|(D_c F)(0;d_1,\ldots,d_N)|=\prod_{1\le i<j\le N}(d_i-d_j).
\ee
(The determinant vanishes for $d_i=d_j$ and it is a polynomial of degree $N(N-1)/2$ in $d_1,\ldots,d_N$. By antisymmetry it then must be a nonzero multiple of the right-hand side. It is easy to see this multiple equals 1.)

Since the numbers $d_1,\ldots,d_N$ are distinct, the product on the right-hand side is nonzero. Hence the holomorphic implicit function theorem may be invoked. From this we deduce that for $|\epsilon|<a$ with~$a$ small enough, there exists a unique solution $c(\epsilon)$ satisfying
\be\label{SSSSS60}
|c_j(\epsilon)-d_j|\le r_j/3,\ \ \ j=1,\ldots,N,
\ee
where
\be\label{SSSSS61}
r_1\equiv \mu_1,\ \ r_N\equiv \mu_{N-1},\ \ r_j\equiv \min(\mu_{j-1},\mu_j),\ \ j=2,\ldots, N-1.
\ee
This entails that the corresponding $\lambda_j(t)$ (cf.~\eqref{clam}) satisfy~\eqref{lamord}. Hence $\sigma(E(t))$ is simple for all $t\ge T_0$, where 
\be\label{SSSSS62}
T_0\equiv 1/a.
\ee

Writing next
\be\label{SSSSS63}
c_j(\epsilon)= d_j+c'_j(0)\epsilon +c''_j(0)\epsilon^2/2 + O(\epsilon^3),
\ee
it is immediate from the above system that $c'_j(0)$ equals $M_{jj}$. To prove that we have
\be\label{SSSSS64}
c''_j(0)=2\alpha_j,\ \ \ j=1,\ldots,N,
\ee
with $\alpha_j$ given by~\eqref{alj}, is arduous, but straightforward. (One need only check that $F_j(\epsilon;c)$ vanishes to second order in $\epsilon$ when $c_j$ is replaced by $d_j+\epsilon M_{jj} +\epsilon^2 \alpha_j$ in~\eqref{Fjep}. Using permutations, it suffices to verify that the sum of all second-order terms involving   $M_{11}M_{22}$ and $M_{12}M_{21}$ vanishes. Noting that the latter product can only arise from $\alpha_1$ and $\alpha_2$, this can be readily achieved.) 

Now by holomorphy there exists for a given $\eta>0$ a number $\de\in (0,a/2]$ such that
\be\label{SSSSS65}
|c_j(\epsilon)-d_j-M_{jj}\epsilon|\le (|\alpha_j|+\eta)|\epsilon|^2.
\ee
Defining
\be\label{SSSSS66}
T_E\equiv 1/\de,
\ee
we then obtain \eqref{lamM}--\eqref{rhob}. Finally, the uniformity assertion follows as in the proof of the above theorem.
\end{proof}



\begin{thebibliography}{99}

    \bibitem{Calogero}
        F.~Calogero, Solution of the one-dimensional $N$-body problem with 
        quadratic and/or inversely quadratic pair potentials, 
        \emph{J. Math. Phys.} \textbf{12} (1971) 419-436.
	
    \bibitem{Sutherland}
        B.~Sutherland, 
        Exact results for a quantum many body problem in one dimension, 
        \emph{Phys. Rev.} \textbf{A 4} (1971) 2019-2021.
	
    \bibitem{Moser_1975}
        J.~Moser,
        Three integrable Hamiltonian systems connected with isospectral 
        deformations, 
        \emph{Adv. Math.} \textbf{16} (1975) 197-220.

    \bibitem{Olsha_Pere_1976}
        M.A.~Olshanetsky and A.M.~Perelomov,
        Completely integrable Hamiltonian systems connected with semisimple 
        Lie algebras,
        \emph{Invent. Math.} \textbf{37} (1976) 93-108.
        
    \bibitem{Olsha_Pere_1981}
        M.A.~Olshanetsky and A.M.~Perelomov,
        Classical integrable finite-dimensional systems related to 
        Lie algebras, 
        \emph{Phys. Rep.} \textbf{71} (1981) 313-400.

    \bibitem{Sutherland_2004}
        B.~Sutherland,
        \emph{Beautiful models: 70 years of exactly solved quantum many-body 
        problems},
        World Scientific, Singapore, 2004.
        
    \bibitem{Ruij_Schneider}
        S.N.M.~Ruijsenaars and H.~Schneider,
        A new class of integrable models and its relation to solitons,
        \emph{Ann. Phys. (N.Y.)} \textbf{170} (1986) 370-405.

    \bibitem{Ruij_FiniteDimSolitonSystems}
        S.N.M.~Ruijsenaars,
        Finite-dimensional soliton systems, 
        pp. 165-206 in: 
        B. Kupershmidt (Ed.), 
        \emph{Integrable and superintegrable systems}, 
        World Scientific, 1990.

    \bibitem{van_Diejen_ComposMath}
        J.F.~van Diejen,
        Commuting difference operators with polynomial eigenfunctions,
        \emph{Compositio Math.} \textbf{95} (1995) 183-233.

    \bibitem{van_Diejen_TMP1994}
        J.F.~van Diejen,
        Deformations of Calogero--Moser systems,
        \emph{Theor. Math. Phys.} {\bf 99} (1994) 549-554.

    \bibitem{van_Diejen_JMP1995}
        J.F.~van Diejen,
        Difference Calogero--Moser systems and finite Toda chains,
        \emph{J. Math. Phys.} {\bf 36} (1995) 1299-1323.

    \bibitem{Ruij_CMP1988}
        S.N.M.~Ruijsenaars, 
        Action-angle maps and scattering theory for some finite dimensional 
        integrable systems I. The pure soliton case,
        \emph{Commun. Math. Phys.} \textbf{115} (1988) 127-165.

    \bibitem{Ruij_RIMS_2}
        S.N.M.~Ruijsenaars, 
        Action-angle maps and scattering theory for some finite dimensional 
        integrable systems II. Solitons, antisolitons and their bound states,
        \emph{Publ. RIMS} \textbf{30} (1994) 865-1008.

    \bibitem{Ruij_RIMS_3}
        S.N.M.~Ruijsenaars, 
        Action-angle maps and scattering theory for some finite dimensional 
        integrable systems III. Sutherland type systems and their duals,
        \emph{Publ. RIMS} \textbf{31} (1995) 247-353.

    \bibitem{Nekrasov_1999}
        N.~Nekrasov, 
        Infinite-dimensional algebras, many-body systems and gauge theories,
        pp. 263-299 in: 
        A.Yu.~Morozov and M.A.~Olshanetsky (Eds.), 
        \emph{Moscow Seminar in Mathematical Physics}, 
        AMS Transl. Ser. 2, vol. 191,
        American Mathematical Society, Providence, 1999.

    \bibitem{Fock_Rosley_1999}
        V.V.~Fock and A.A.~Rosly, 
        Poisson structure on moduli of flat connections on Riemann surfaces 
        and the $r$-matrix,
        pp. 67-86 in: 
        A.Yu.~Morozov and M.A.~Olshanetsky (Eds.), 
        \emph{Moscow Seminar in Mathematical Physics}, 
        AMS Transl. Ser. 2, vol. 191,
        American Mathematical Society, Providence, 1999.
        
    \bibitem{Fock_et_al_JHEP07}
        V.~Fock, A.~Gorsky, N.~Nekrasov and V.~Rubtsov,
        Duality in Integrable Systems and Gauge Theories,
        \emph{JHEP} \textbf{07} (2000) 028.

    \bibitem{Feher_Klimcik_JPA2009}
        L.~Feh\'er and C.~Klim\v{c}\'ik,
        On the duality between the hyperbolic Sutherland and the rational 
        Ruijsenaars--Schneider models,
        \emph{J. Phys. A: Math. Theor.} \textbf{42} (2009) 185202.

    \bibitem{Feher_Ayadi_JMP2010}
        L.~Feh\'er and V.~Ayadi,
        Trigonometric Sutherland systems and their Ruijsenaars duals from
        symplectic reduction,
        \emph{J. Math. Phys.} \textbf{51} (2010) 103511.

    \bibitem{Feher_Klimcik_CMP2011}
        L.~Feh\'er and C.~Klim\v{c}\'ik,
        Poisson--Lie interpretation of trigonometric Ruijsenaars duality,
        \emph{Commun. Math. Phys.} \textbf{301} (2011) 55-104. 
        
    \bibitem{Feher_Klimcik_NPB2012}
        L.~Feh\'er and C.~Klim\v{c}\'ik,
        Self-duality of the compactified Ruijsenaars--Schneider system 
        from quasi-Hamiltonian reductions,
        \emph{Nucl. Phys.} \textbf{B 860} (2012) 464-515.
      
    \bibitem{Feher_Kluck_NPB2014}  
        L.~Feh\'er and T.J.~Kluck,
        New compact forms of the trigonometric Ruijsenaars--Schneider system,
        \emph{Nucl. Phys.} \textbf{B 882} (2014) 97-127.

    \bibitem{Chen_Hou_JPA2001}
        K.~Chen and B.Y.~Hou,
        The $D_n$ Ruijsenaars--Schneider model,
        \emph{J. Phys. A} \textbf{34} (2001) 7579-7589.

    \bibitem{Pusztai_NPB2011}
        B.G.~Pusztai,
        Action-angle duality between the $C_n$-type hyperbolic Sutherland 
        and the rational Ruijsenaars--Schneider--van Diejen models,
        \emph{Nucl. Phys.} \textbf{B 853} (2011) 139-173.

    \bibitem{Pusztai_NPB2012}
        B.G.~Pusztai,
        The hyperbolic $BC_n$ Sutherland and the rational $BC_n$ 
        Ruijsenaars--Schnei\-der--van Diejen models: Lax matrices and duality,
        \emph{Nucl. Phys.} \textbf{B 856} (2012) 528-551.

    \bibitem{Feher_Gorbe_JMP2014}
        L.~Feh\'er and T.F.~G\"orbe,
        Duality between the trigonometric $BC_n$ Sutherland system and 
        a completed rational Ruijsenaars--Schneider--van Diejen system,
        \emph{J. Math. Phys.} \textbf{55} (2014) 102704.

    \bibitem{Ruij_TMP2008}
        S.N.M.~Ruijsenaars, 
        The classical hyperbolic Askey--Wilson dynamics without bound states,
        \emph{Theor. Math. Phys.} \textbf{154} (2008) 418-432.

    \bibitem{Pusztai_Gorbe}
        B.G.~Pusztai and T.F.~G\"orbe,
        Lax representation of the hyperbolic van Diejen dynamics 
        with two coupling parameters, accepted for publication in
        \emph{Commun. Math. Phys.},
        \texttt{arXiv:1603.0671 [math-ph]}.

    \bibitem{Knapp}
        A.W.~Knapp,
        \emph{Lie groups beyond an introduction},
        Progress in Mathematics, vol. 140, Birkh\"auser, Boston, MA, 2002.

    \bibitem{Pusztai_JPA2011}
        B.G.~Pusztai, 
        On the scattering theory of the classical hyperbolic $C_n$ 
        Sutherland model,
        \emph{J. Phys. A: Math. Theor.} \textbf{44} (2011) 155306.

    \bibitem{Prasolov}
        V.V.~Prasolov,
        \emph{Problems and theorems in linear algebra}, 
        American Mathematical Society, Providence, 1994.        

    \bibitem{Ruij_CMP1987}
        S.N.M.~Ruijsenaars, 
        Complete integrability of relativistic Calogero--Moser systems 
        and elliptic function identities, 
        \emph{Commun. Math. Phys.} \textbf{110} (1987) 191-213.
        
    \bibitem{Lee}
        J.M.~Lee,
        \emph{Introduction to Smooth Manifolds},
        2nd ed.,
        Graduate Texts in Mathematics 218, Springer, New York, 2013.
        
    \bibitem{Hartman}
        P.~Hartman,
        \emph{Ordinary Differential Equations}, 
        2nd ed., 
        SIAM, Philadelphia, 2002.
        
    \bibitem{Folland}
        G.B.~Folland,
        \emph{Real Analysis: Modern Techniques and Their Applications},
        2nd ed., 
        Wiley-In\-ter\-science, John Wiley \& Sons, New York, 1999.

    \bibitem{Kulish_1976}
        P.P.~Kulish,
        Factorization of the classical and the quantum $S$ matrix and 
        conservation laws,
        \emph{Theor. Math. Phys.} \textbf{26} (1976) 132-137.
		
    \bibitem{Moser_1977}
        J.~Moser,
        The scattering problem for some particle systems on the line, 
        pp. 441-463 in: 
        \emph{Lecture Notes in Mathematics} \textbf{597}, Springer, 1977. 

    \bibitem{Pusztai_NPB2013}
        B.G.~Pusztai,
        Scattering theory of the hyperbolic $BC_n$ Sutherland and
        the rational $BC_n$ Ruij\-senaars--Schneider--van Diejen models,
        \emph{Nucl. Phys.} \textbf{B 874} (2013) 647-662.

    \bibitem{Reed_Simon_VOL_4}
        M.~Reed and B.~Simon, 
        \emph{Methods of Modern Mathematical Physics. 
        IV: Analysis of Operators}, 
        Academic Press, New York, 1978.

\end{thebibliography}
\end{document}